\documentclass[pra,twocolumn,superscriptaddress]{revtex4-2}%
\usepackage{amsfonts}
\usepackage{amsmath}
\usepackage{amssymb}
\usepackage{graphicx}
\usepackage[colorlinks=true,linkcolor=blue,citecolor=red,plainpages=false,pdfpagelabels]%
{hyperref}%
\providecommand{\U}[1]{\protect\rule{.1in}{.1in}}
\newtheorem{theorem}{Theorem}

\newtheorem{corollary}[theorem]{Corollary}

\newtheorem{definition}[theorem]{Definition}

\newtheorem{proposition}[theorem]{Proposition}
\newtheorem{remark}[theorem]{Remark}

\newenvironment{proof}[1][Proof]{\noindent\textbf{#1.} }{\ \rule{0.5em}{0.5em}}
\allowdisplaybreaks
\begin{document}
\title{Entropy of a quantum channel}
\author{Gilad Gour}
\email{gour@ucalgary.ca}
\affiliation{Department of Mathematics and Statistics, Institute for Quantum Science and
Technology, University of Calgary, Alberta, Canada T2N 1N4}
\author{Mark M. Wilde}
\email{mwilde@lsu.edu}
\affiliation{Hearne Institute for Theoretical Physics, Department of Physics and Astronomy, and 
Center for Computation and Technology, Louisiana State University, Baton
Rouge, Louisiana 70803, USA}

\begin{abstract}
The von Neumann entropy of a quantum state is a central concept in physics and
information theory, having a number of compelling physical interpretations.
There is a certain perspective that the most fundamental notion in quantum
mechanics is that of a quantum channel, as quantum states, unitary evolutions,
measurements, and discarding of quantum systems can each be regarded as
certain kinds of quantum channels. Thus, an important goal is to define a
consistent and meaningful notion of the entropy of a quantum channel.
Motivated by the fact that the entropy of a state $\rho$ can be formulated as
the difference of the number of physical qubits and the \textquotedblleft
relative entropy distance\textquotedblright\ between $\rho$ and the maximally
mixed state, here we define the entropy of a channel $\mathcal{N}$\ as the
difference of the number of physical qubits of the channel output with the
\textquotedblleft relative entropy distance\textquotedblright\ between
$\mathcal{N}$ and the completely depolarizing channel. We prove that this
definition satisfies all of the axioms, recently put forward in [Gour, IEEE
Trans.~Inf.~Theory \textbf{65}, 5880 (2019)], required for a channel entropy
function. The task of quantum channel merging, in which the goal is for the
receiver to merge his share of the channel with the environment's share, gives
a compelling operational interpretation of the entropy of a channel. The entropy of a channel can be negative for certain channels, but this negativity has an operational interpretation in terms of the channel merging protocol. We define
R\'{e}nyi and min-entropies of a channel and prove that they satisfy the
axioms required for a channel entropy function. Among other results, we also
prove that a smoothed version of the min-entropy of a channel satisfies the
asymptotic equipartition property. 

\end{abstract}
\maketitle


\tableofcontents

\section{Introduction}

In his foundational work on quantum statistical mechanics, von Neumann
extended the classical Gibbs entropy concept to the quantum realm \cite{VN32}.
This extension, known as the von Neumann or quantum entropy, plays a key role
in physics and information theory. It is defined by the following formula
\cite{VN32}:%
\begin{equation}
H(A)_{\rho}\equiv-\operatorname{Tr}\{\rho_{A}\log_{2}\rho_{A}\},
\end{equation}
where $\rho_{A}$ is the state of a system $A$. The entropy has operational
interpretations in terms of quantum data compression \cite{S95} and optimal
entanglement manipulation rates of pure bipartite quantum states
\cite{BBPS96}, where the choice of base two for the logarithm becomes clear.
In recent developments of quantum thermodynamics, it was shown that the free
energy, namely, the difference of the energy and the product of the
temperature and the von Neumann entropy, can be interpreted as the rate at
which work can be extracted from a large number of copies of a quantum system
in a thermal bath at fixed temperature, by using only thermal operations
\cite{BHORS13}.

By defining the quantum relative entropy of a state $\rho_{A}$ and a positive
semi-definite operator $\sigma_{A}$ as \cite{U62}%
\begin{equation}
D(\rho_{A}\Vert\sigma_{A})\equiv\operatorname{Tr}\{\rho_{A}\left[  \log
_{2}\rho_{A}-\log_{2}\sigma_{A}\right]  \},
\end{equation}
if $\operatorname{supp}(\rho_{A})\subseteq\operatorname{supp}(\sigma_{A})$ and
$D(\rho_{A}\Vert\sigma_{A})=+\infty$ otherwise, we can rewrite the formula for
quantum entropy as follows:%
\begin{equation}
H(A)_{\rho}=\log_{2}\left\vert A\right\vert -D(\rho_{A}\Vert\pi_{A}),
\label{eq:ent-rel-ent}%
\end{equation}
where $\left\vert A\right\vert $ denotes the dimension of the system $A$ and
$\pi_{A}\equiv I_{A}/\left\vert A\right\vert $ denotes the maximally mixed
state. In this way, we can think of entropy as quantifying the difference of
the number of physical qubits contained in the system~$A$ and the
\textquotedblleft relative entropy distance\textquotedblright\ of the state
$\rho_{A}$ to the maximally mixed state $\pi_{A}$. This way of thinking about
quantum entropy is relevant in the resource theory of purity
\cite{OHHH02,HHHH0SS03,HHO03,OHHHH03,D05,KD07}, in which the goal is to
distill local pure states from a given state (or vice versa) by allowing local
unitary operations for free. Furthermore, the quantum relative entropy
$D(\rho_{A}\Vert\pi_{A})$ has an operational meaning as the optimal rate at
which the state $\rho_{A}$ can be distinguished from the maximally mixed state
$\pi_{A}$ in the Stein setting of quantum hypothesis testing \cite{HP91,ON00}.
In what follows, we use the formula in \eqref{eq:ent-rel-ent}\ as the basis
for defining the entropy of a quantum channel.

For some time now, there has been a growing realization that the fundamental
constituents of quantum mechanics are quantum channels. Recall that a quantum
channel $\mathcal{N}_{A\rightarrow B}$\ is a completely positive, trace
preserving map that takes a quantum state for system$~A$ to one for system$~B$
\cite{H12}. Indeed, all the relevant components of the theory, including
quantum states, measurements, unitary evolutions, etc.,~can be written as
quantum channels. A quantum state can be understood as a preparation channel,
sending a trivial quantum system to a non-trivial one prepared in a given
state. A quantum measurement can be understood as a quantum channel that sends
a quantum system to a classical one; and of course a unitary evolution is a
kind of quantum channel, as well as the discarding of a quantum system. One
might even boldly go as far as to say that there is really only a single
postulate of quantum mechanics, and it is that \textquotedblleft everything is
a quantum channel.\textquotedblright\ With this perspective, one could start
from this unified postulate and then understand from there particular kinds of
channels, i.e., states, measurements, and unitary evolutions.

Due to the fundamental roles of quantum channels and the entropy of a quantum
state, as highlighted above, it is thus natural to ask whether there is a
meaningful notion of the entropy of a quantum channel, i.e., a quantifier of
the uncertainty of a quantum channel. As far as we are aware, this question
has not been fully addressed in prior literature (see Remark~\ref{rem:other-ent-func} for further discussion), and it is the aim of the
present paper to provide a convincing notion of a quantum channel's entropy.
To define such a notion, we look to \eqref{eq:ent-rel-ent} for inspiration. As
such, we need generalizations of the quantum relative entropy and the
maximally mixed state to the setting of quantum channels:

\begin{enumerate}
\item The quantum relative entropy of channels $\mathcal{N}_{A\rightarrow B}$
and $\mathcal{M}_{A\rightarrow B}$ is defined as \cite{CMW16,LKDW18}
\begin{equation}
D(\mathcal{N}\Vert\mathcal{M})\equiv\sup_{\rho_{RA}}D(\mathcal{N}%
_{A\rightarrow B}(\rho_{RA})\Vert\mathcal{M}_{A\rightarrow B}(\rho_{RA})),
\label{eq:channel-rel-ent}%
\end{equation}
where the optimization is with respect to bipartite states $\rho_{RA}$ of a
reference system $R$ of arbitrary size and the channel input system~$A$. Due
to state purification, the data-processing inequality \cite{Lindblad1975}%
,\ and the Schmidt decomposition theorem, it suffices to optimize over states
$\rho_{RA}$ that are pure and such that system $R$ is isomorphic to system
$A$. This observation significantly reduces the complexity of computing the
channel relative entropy.

\item The channel that serves as a generalization of the maximally mixed state
is the channel $\mathcal{R}_{A\rightarrow B}$ that completely randomizes or
depolarizes the input state as follows:%
\begin{equation}
\mathcal{R}_{A\rightarrow B}(X_{A})=\operatorname{Tr}\{X_{A}\}\pi_{B},
\label{eq:cm-random-ch}%
\end{equation}
where $X_{A}$ is an arbitrary operator for system $A$. That is, its action is
to discard the input and replace with a maximally mixed state $\pi_{B}$.
\end{enumerate}

With these notions in place, we can now define the entropy of a quantum channel:

\begin{definition}
[Entropy of a quantum channel]\label{def:entropy-channel}Let $\mathcal{N}%
_{A\rightarrow B}$ be a quantum channel. Its entropy is defined as%
\begin{equation}
H(\mathcal{N})\equiv\log_{2}\left\vert B\right\vert -D(\mathcal{N}%
\Vert\mathcal{R}),
\end{equation}
where $D(\mathcal{N}\Vert\mathcal{R})$ is the channel relative entropy in
\eqref{eq:channel-rel-ent}\ and $\mathcal{R}_{A\rightarrow B}$ is the
completely randomizing channel in \eqref{eq:cm-random-ch}.
\end{definition}

We remark here that, in analogy to the operational interpretation for
$D(\rho_{A}\Vert\pi_{A})$ mentioned above, it is known that $D(\mathcal{N}%
\Vert\mathcal{R})$ is equal to the optimal rate at which the channel
$\mathcal{N}_{A\rightarrow B}$ can be distinguished from the completely
randomizing channel $\mathcal{R}_{A\rightarrow B}$, by allowing for any
possible quantum strategy to distinguish the channels \cite{CMW16}. Again,
this statement holds in the Stein setting of quantum hypothesis testing (see
\cite{CMW16} for details). We also emphasize here that the entropy of a channel can be negative for some channels, but this negativity has an operational interpretation in terms of the channel merging protocol (see Remark~\ref{rem:channel-entr-bounds-interp} in this context). 

The remainder of our paper contains arguments advocating for this definition
of a channel's entropy. In the next section, we show that it satisfies the
three basic axioms, put forward in \cite{G18}, for any function to be called
an entropy function for a quantum channel, including non-decrease under the
action of a random unitary superchannel, additivity, and normalization. After
that, we provide several alternate representations for the entropy of a
channel, the most significant of which is the completely bounded entropy of
\cite{DJKR06}. Section~\ref{sec:q-ch-merge} delivers an operational
interpretation of a channel's entropy in terms of an information-theoretic
task that we call quantum channel merging, which is a dynamical counterpart of
the well known task of quantum state merging \cite{HOW05,HOW07}. We calculate
channel entropies for several example channels in Section~\ref{sec:examples},
which include erasure, dephasing, depolarizing, and Werner--Holevo channels.
In the same section, we introduce the energy-constrained and unconstrained
entropies of a quantum channel and calculate them for thermal, amplifier, and
additive-noise bosonic Gaussian channels. In Section~\ref{sec:renyi-ent-ch},
we define the $\alpha$-R\'{e}nyi entropy of a channel, prove that it satisfies
the basic axioms for certain values of the R\'{e}nyi parameter$~\alpha$, and
provide alternate representations for it. In Section~\ref{sec:min-ent-ch}, we
define the min-entropy of a channel, establish that it satisfies the basic
axioms, and provide alternate representations for it. In Section~\ref{sec:AEP}%
, we define the smoothed min-entropy of a channel, and then we prove an
asymptotic equipartition property, which relates the smoothed min-entropy of a
channel to its entropy. In Section~\ref{sec:other-ents}, we discuss other
entropies of a channel, noting that several of them collapse to the (von
Neumann) entropy of a channel. We finally conclude in Section~\ref{sec:concl}
with a summary and some open questions.

\textit{Note on related work}---After completing the results in our related
preprint \cite{GW18}, we noticed \cite[Eq.~(6)]{Yuan2018}, in which Yuan
proposed to define the entropy of a quantum channel in the same way as we have
proposed in Definition~\ref{def:entropy-channel}. Yuan's work is now published
as \cite{PhysRevA.99.032317}.

\begin{remark}
\label{rem:other-ent-func}
We note here that \textquotedblleft the entropy of a channel\textquotedblright%
\ was also defined in \cite{RZF11,R11}, but the definition given there does
not satisfy \textquotedblleft reduction to states\textquotedblright\ or the
basic axiom of normalization. For this reason, it cannot be considered an
entropy function according to the approach of \cite{G18}.
\end{remark}

\section{Entropy of a quantum channel}

Proceeding with Definition~\ref{def:entropy-channel} for the entropy of a
quantum channel, we now establish several of its properties, and then we
provide alternate representations for it.

\subsection{Properties of the entropy of a quantum channel}

In \cite{G18}, it was advocated that a function of a quantum channel is an
entropy function if it satisfies\ non-decrease under random unitary
superchannels, additivity, and normalization. As shown in the next three
subsections, the entropy of a channel, as given in
Definition~\ref{def:entropy-channel}, satisfies all three axioms, and in fact,
it satisfies stronger properties that imply these.

\subsubsection{Non-decrease under the action of a uniformity preserving
superchannel}

Before addressing the first axiom, let us first briefly review the notion of
superchannels \cite{CDP08}, which are linear maps that take as input a quantum
channel and output a quantum channel. To define them, let $\mathcal{L}%
(A\rightarrow B)$ denote the set of all linear maps from $\mathcal{L}(A)$ to
$\mathcal{L}(B)$. Similarly, let $\mathcal{L}(C\rightarrow D)$ denote the set
of all linear maps from $\mathcal{L}(C)$ to $\mathcal{L}(D)$. Let
$\Theta:\mathcal{L}(A\rightarrow B)\rightarrow\mathcal{L}(C\rightarrow D)$
denote a linear supermap, taking $\mathcal{L}(A\rightarrow B)$ to
$\mathcal{L}(C\rightarrow D)$. A quantum channel is a particular kind of
linear map, and any linear supermap $\Theta$ that takes as input an arbitrary
quantum channel $\Psi_{A\rightarrow B}\in\mathcal{L}(A\rightarrow B)$ and is
required to output a quantum channel $\Phi_{C\rightarrow D}\in\mathcal{L}%
(C\rightarrow D)$ should preserve the properties of complete positivity (CP)
and trace preservation (TP). That is, the supermap should be CPTP preserving.
Furthermore, for the supermap to be physical, the same should be true when it
acts on subsystems of bipartite quantum channels, so that the supermap
$\operatorname{id}\otimes\Theta$ should be CPTP preserving, where
$\operatorname{id}$ represents an arbitrary identity supermap. A supermap
satisfying this property is said to be completely CPTP preserving and is then
called a superchannel. It was proven in \cite{CDP08} that any superchannel
$\Theta:\mathcal{L}(A\rightarrow B)\rightarrow\mathcal{L}(C\rightarrow D)$ can
be physically realized as follows. If
\begin{equation}
\Phi_{C\rightarrow D}=\Theta\lbrack\Psi_{A\rightarrow B}]
\end{equation}
for an arbitrary input channel $\Psi_{A\rightarrow B}\in\mathcal{L}%
(A\rightarrow B)$ and some output channel $\Phi_{C\rightarrow D}\in
\mathcal{L}(C\rightarrow D)$, then the physical realization of the
superchannel $\Theta$ is as follows:%
\begin{equation}
\Phi_{C\rightarrow D}=\Omega_{BE\rightarrow D}\circ\left(  \Psi_{A\rightarrow
B}\otimes\operatorname{id}_{E}\right)  \circ\Lambda_{C\rightarrow AE},
\label{eqn:superchannel}%
\end{equation}
where $\Lambda_{C\rightarrow AE}:\mathcal{L}(C)\rightarrow\mathcal{L}(AE)$ is
a pre-processing channel, system $E$ corresponds to some memory or environment
system, and $\Omega_{BE\rightarrow D}:\mathcal{L}(BE)\rightarrow
\mathcal{L}(D)$ is a post-processing channel.

A uniformity preserving superchannel $\Theta$ is a superchannel that takes the
completely randomizing channel $\mathcal{R}_{A\rightarrow B}$ in
\eqref{eq:cm-random-ch} to another completely randomizing channel
$\mathcal{R}_{C\rightarrow D}$, such that $\left\vert A\right\vert =\left\vert
C\right\vert $ and $\left\vert B\right\vert =\left\vert D\right\vert $, i.e.,%
\begin{equation}
\Theta(\mathcal{R}_{A\rightarrow B})=\mathcal{R}_{C\rightarrow D}.
\label{eq:unif-preserv}%
\end{equation}
For such superchannels, we have the following:

\begin{proposition}
\label{prop:monotone-unif-preserv}Let $\mathcal{N}_{A\rightarrow B}$ be a
quantum channel, and let $\Theta$ be a uniformity preserving superchannel as
defined above. Then the entropy of a channel does not decrease under the
action of such a superchannel:
\begin{equation}
H(\Theta(\mathcal{N}))\geq H(\mathcal{N}).
\end{equation}

\end{proposition}

\begin{proof}
This follows from the fact that the channel relative entropy is non-increasing
under the action of an arbitrary superchannel \cite{G18,Yuan2018}. That is,
for two channels $\mathcal{N}_{A\rightarrow B}$ and $\mathcal{M}_{A\rightarrow
B}$, and a superchannel $\Xi$, the following inequality holds%
\begin{equation}
D(\mathcal{N}\Vert\mathcal{M})\geq D(\Xi(\mathcal{N})\Vert\Xi(\mathcal{M})).
\end{equation}
Applying this, we find that%
\begin{align}
H(\mathcal{N})  &  =\log_{2}\left\vert B\right\vert -D(\mathcal{N}%
\Vert\mathcal{R})\label{eq:uniform-pres-1}\\
&  \leq\log_{2}\left\vert B\right\vert -D(\Theta(\mathcal{N})\Vert
\Theta(\mathcal{R}))\\
&  =\log_{2}\left\vert B\right\vert -D(\Theta(\mathcal{N})\Vert\mathcal{R})\\
&  =\log_{2}\left\vert D\right\vert -D(\Theta(\mathcal{N})\Vert\mathcal{R})\\
&  =H(\Theta(\mathcal{N})). \label{eq:uniform-pres-last}%
\end{align}
The second equality follows by definition from \eqref{eq:unif-preserv}.
\end{proof}

\bigskip

In \cite{G18}, a superchannel $\Upsilon$ was called a random unitary
superchannel if its action on a channel $\mathcal{N}_{A\rightarrow B}$ can be
written as%
\begin{equation}
\Upsilon(\mathcal{N}_{A\rightarrow B})=\sum_{x}p_{X}(x)\mathcal{V}%
_{B\rightarrow D}^{x}\circ\mathcal{N}_{A\rightarrow B}\circ\mathcal{U}%
_{C\rightarrow A}^{x},
\end{equation}
where $\mathcal{U}_{C\rightarrow A}^{x}$ and $\mathcal{V}_{B\rightarrow D}%
^{x}$ are unitary channels and $p_{X}(x)$ is a probability distribution. In
\cite{G18}, it was proved that a random unitary superchannel is a special kind
of uniformity preserving superchannel. Thus, due to
Proposition~\ref{prop:monotone-unif-preserv}, it follows that the entropy of a
channel, as given in Definition~\ref{def:entropy-channel}, satisfies the first
axiom from \cite{G18} required for an entropy function.

\subsubsection{Additivity}

In this subsection, we prove that the entropy of a channel is additive, which
is the second axiom proposed in \cite{G18}\ for a channel entropy
function.\ The proof is related to many prior additivity results from
\cite{AC97,A04,DJKR06,CMW16,BHKW18}.

\begin{proposition}
[Additivity]\label{prop:entropy-additive}Let $\mathcal{N}$ and $\mathcal{M}$
be quantum channels. Then the channel entropy is additive in the following
sense:
\begin{equation}
H(\mathcal{N}\otimes\mathcal{M})=H(\mathcal{N})+H(\mathcal{M}).
\end{equation}

\end{proposition}

\begin{proof}
This can be understood as a consequence of the additivity results from
\cite{CMW16,BHKW18}, which in turn are related to the earlier additivity
results from \cite{AC97,A04,DJKR06}. For channels $\mathcal{N}_{A_{1}%
\rightarrow B_{1}}$ and $\mathcal{M}_{A_{2}\rightarrow B_{2}}$, and
corresponding randomizing channels $\mathcal{R}_{A_{1}\rightarrow B_{1}}%
^{(1)}$ and $\mathcal{R}_{A_{2}\rightarrow B_{2}}^{(2)}$, we have by
definition that%
\begin{align}
&  H(\mathcal{N}\otimes\mathcal{M})\nonumber\\
&  =\log_{2}(\left\vert B_{1}\right\vert \left\vert B_{2}\right\vert
)-D(\mathcal{N}\otimes\mathcal{M}\Vert\mathcal{R}^{(1)}\otimes\mathcal{R}%
^{(2)})\\
&  =\log_{2}\left\vert B_{1}\right\vert +\log_{2}\left\vert B_{2}\right\vert
-D(\mathcal{N}\otimes\mathcal{M}\Vert\mathcal{R}^{(1)}\otimes\mathcal{R}%
^{(2)}),
\end{align}
and so the result follows if%
\begin{equation}
D(\mathcal{N}\otimes\mathcal{M}\Vert\mathcal{R}^{(1)}\otimes\mathcal{R}%
^{(2)})=D(\mathcal{N}\Vert\mathcal{R}^{(1)})+D(\mathcal{M}\Vert\mathcal{R}%
^{(2)}). \label{eq:add-ch-rel-ent}%
\end{equation}
Note that the inequality \textquotedblleft$\geq$\textquotedblright\ for
\eqref{eq:add-ch-rel-ent}\ trivially follows, and so it remains to prove the
inequality \textquotedblleft$\leq$\textquotedblright\ for
\eqref{eq:add-ch-rel-ent}. To this end, let $\psi_{RA_{1}A_{2}}$ be an
arbitrary pure state, and define%
\begin{align}
\rho_{R^{\prime}A_{1}}  &  \equiv\mathcal{M}_{A_{2}\rightarrow B_{2}}%
(\psi_{RA_{1}A_{2}}),\\
\sigma_{R^{\prime}A_{1}}  &  \equiv\mathcal{R}_{A_{2}\rightarrow B_{2}}%
(\psi_{RA_{1}A_{2}}),
\end{align}
where system $R^{\prime}\equiv RB_{2}$. Then we find that\begin{widetext}
\begin{align}
&  D((\mathcal{N}_{A_{1}\rightarrow B_{1}}\otimes\mathcal{M}_{A_{2}\rightarrow
B_{2}})(\psi_{RA_{1}A_{2}})\Vert(\mathcal{R}_{A_{1}\rightarrow B_{1}}%
\otimes\mathcal{R}_{A_{2}\rightarrow B_{2}})(\psi_{RA_{1}A_{2}}))\nonumber\\
&  =D(\mathcal{N}_{A_{1}\rightarrow B_{1}}(\rho_{R^{\prime}A_{1}}%
)\Vert\mathcal{R}_{A_{1}\rightarrow B_{1}}(\sigma_{R^{\prime}A_{1}%
}))\label{eq:addit-1}\\
&  \leq D(\mathcal{N}_{A_{1}\rightarrow B_{1}}(\rho_{R^{\prime}A_{1}}%
)\Vert\mathcal{R}_{A_{1}\rightarrow B_{1}}(\rho_{R^{\prime}A_{1}}%
))+D(\rho_{R^{\prime}A_{1}}\Vert\sigma_{R^{\prime}A_{1}})\\
&  =D(\mathcal{N}_{A_{1}\rightarrow B_{1}}(\rho_{R^{\prime}A_{1}}%
)\Vert\mathcal{R}_{A_{1}\rightarrow B_{1}}(\rho_{R^{\prime}A_{1}%
}))+D(\mathcal{M}_{A_{2}\rightarrow B_{2}}(\psi_{RA_{1}A_{2}})\Vert
\mathcal{R}_{A_{2}\rightarrow B_{2}}(\psi_{RA_{1}A_{2}}))\\
&  \leq\sup_{\rho_{R^{\prime}A_{1}}}D(\mathcal{N}_{A_{1}\rightarrow B_{1}%
}(\rho_{R^{\prime}A_{1}})\Vert\mathcal{R}_{A_{1}\rightarrow B_{1}}%
(\rho_{R^{\prime}A_{1}}))\nonumber\\
&  \qquad+\sup_{\psi_{RA_{1}A_{2}}}D(\mathcal{M}_{A_{2}\rightarrow B_{2}}%
(\psi_{RA_{1}A_{2}})\Vert\mathcal{R}_{A_{2}\rightarrow B_{2}}(\psi
_{RA_{1}A_{2}}))\\
&  =D(\mathcal{N}_{A_{1}\rightarrow B_{1}}\Vert\mathcal{R}_{A_{1}\rightarrow
B_{1}})+D(\mathcal{M}_{A_{2}\rightarrow B_{2}}\Vert\mathcal{R}_{A_{2}%
\rightarrow B_{2}}). \label{eq:addit-last}%
\end{align}
\end{widetext}
The first inequality follows from the same steps given in the proof of
\cite[Lemma~38]{BHKW18}. This concludes the proof.
\end{proof}

\bigskip

Another approach to establishing additivity is to employ the first identity of
Proposition~\ref{prop:alt-reps-ent-ch} (in Section~\ref{sec:alt-reps-ent}) and
\cite[Eq.~(3.28)]{AC97}, the latter of which was independently formulated in
\cite[Section~2.3]{DJKR06}.

\subsubsection{Reduction to states and normalization}

We now prove that the entropy of a channel reduces to the entropy of a state
if the channel is one that replaces the input with a given state.

\begin{proposition}
[Reduction to states]\label{prop:reduction-to-states}Let the channel
$\mathcal{N}_{A\rightarrow B}$ be a replacer channel, defined such that
$\mathcal{N}_{A\rightarrow B}(\rho_{A})=\sigma_{B}$ for all states $\rho_{A}$
and some state $\sigma_{B}$. Then the following equality holds%
\begin{equation}
H(\mathcal{N})=H(B)_{\sigma}.
\end{equation}

\end{proposition}

\begin{proof}
For any input $\psi_{RA}$, the output is $\mathcal{N}_{A\rightarrow B}%
(\psi_{RA})=\psi_{R}\otimes\sigma_{B}$, and we find that%
\begin{align}
D(\mathcal{N}_{A\rightarrow B}(\psi_{RA})\Vert\mathcal{R}_{A\rightarrow
B}(\psi_{RA}))  &  =D(\psi_{R}\otimes\sigma_{B}\Vert\psi_{R}\otimes\pi
_{B})\nonumber\\
&  =D(\sigma_{B}\Vert\pi_{B}).
\end{align}
This implies that%
\begin{align}
H(\mathcal{N})  &  =\log_{2}\left\vert B\right\vert -D(\mathcal{N}%
\Vert\mathcal{R})\\
&  =\log\left\vert B\right\vert -D(\sigma_{B}\Vert\pi_{B})\\
& =H(B)_{\sigma},
\end{align}
concluding the proof.
\end{proof}

\bigskip

A final axiom (normalization) for a channel entropy function \cite{G18}\ is
that it should be equal to zero for any channel that replaces the input with a
pure state and it should be equal to the logarithm of the output dimension for
any channel that replaces the input with the maximally mixed state. Clearly,
Proposition~\ref{prop:reduction-to-states}\ implies the normalization property
if the replaced state is maximally mixed or pure.

\subsection{Alternate representations for the entropy of a channel}

\label{sec:alt-reps-ent}

The entropy of a quantum channel has at least three alternate representations,
in terms of the completely bounded entropy of \cite{DJKR06}, the entropy gain
of its complementary channel \cite{A04}, and the maximum output entropy of the
channel conditioned on its environment. We recall these various channel
functions now.

Recall that the completely bounded entropy of a quantum channel $\mathcal{N}%
_{A\rightarrow B}$\ is defined as \cite{DJKR06}%
\begin{equation}
H_{\text{CB},\min}(\mathcal{N})\equiv\inf_{\rho_{RA}}H(B|R)_{\omega},
\end{equation}
where $H(B|R)_{\omega}\equiv H(BR)_{\omega}-H(R)_{\omega}$ is the conditional
entropy of the state $\omega_{RB}=\mathcal{N}_{A\rightarrow B}(\rho_{RA})$ and
the system $R$ is unbounded. However, due to data processing, purification,
and the Schmidt decomposition theorem, it follows that%
\begin{equation}
H_{\text{CB},\min}(\mathcal{N})=\inf_{\psi_{RA}}H(B|R)_{\omega},
\label{eq:cb-min}%
\end{equation}
where $\psi_{RA}$ is a pure bipartite state with system $R$ isomorphic to the
channel input system $A$.

Due to the Stinespring representation theorem \cite{S55}, every channel
$\mathcal{N}_{A\rightarrow B}$ can be realized by the action of an isometric
channel $\mathcal{U}_{A\rightarrow BE}^{\mathcal{N}}$ and a partial trace as
follows:%
\begin{equation}
\mathcal{N}_{A\rightarrow B}=\operatorname{Tr}_{E}\circ\mathcal{U}%
_{A\rightarrow BE}^{\mathcal{N}}. \label{eq:iso-extend}%
\end{equation}
If we instead trace over the channel output $B$, this realizes a complementary
channel of $\mathcal{N}_{A\rightarrow B}$:%
\begin{equation}
\mathcal{N}_{A\rightarrow E}^{c}\equiv\operatorname{Tr}_{B}\circ
\mathcal{U}_{A\rightarrow BE}^{\mathcal{N}}. \label{eq:comp-chan}%
\end{equation}
Using these notions, we can define the entropy gain of a complementary channel
of $\mathcal{N}_{A\rightarrow B}$ as follows \cite{A04}:%
\begin{equation}
G(\mathcal{N}_{A\rightarrow E}^{c})\equiv\inf_{\rho_{A}}\left[  H(E)_{\tau
}-H(A)_{\rho}\right]  ,
\end{equation}
where $\tau_{BE}\equiv\mathcal{U}_{A\rightarrow BE}^{\mathcal{N}}(\rho_{A})$.
The entropy gain has been investigated for infinite-dimensional quantum
systems in \cite{Holevo2010,Holevo2011,H11ISIT}. We can also define the
maximum output entropy of the channel conditioned on its environment as%
\begin{equation}
\sup_{\rho_{A}}H(B|E)_{\tau},
\end{equation}
where again $\tau_{BE}\equiv\mathcal{U}_{A\rightarrow BE}^{\mathcal{N}}%
(\rho_{A})$.

We now prove that the entropy of a channel, as given in
Definition~\ref{def:entropy-channel}, is equal to the completely bounded
entropy, the entropy gain of a complementary channel, and the negation of the
maximum output entropy of the channel conditioned on its environment.

\begin{proposition}
\label{prop:alt-reps-ent-ch}Let $\mathcal{N}_{A\rightarrow B}$ be a quantum
channel, and let $\mathcal{U}_{A\rightarrow BE}^{\mathcal{N}}$ be an isometric
channel extending it, as in \eqref{eq:iso-extend}. Then%
\begin{equation}
H(\mathcal{N})=H_{\operatorname{CB},\min}(\mathcal{N})=G(\mathcal{N}%
_{A\rightarrow E}^{c})=-\sup_{\rho_{A}}H(B|E)_{\tau},
\label{eq:von-neu-collapse}%
\end{equation}
where $\tau_{BE}\equiv\mathcal{U}_{A\rightarrow BE}^{\mathcal{N}}(\rho_{A}%
)$.\ It then follows that%
\begin{equation}
\left\vert H(\mathcal{N})\right\vert \leq\log_{2}\left\vert B\right\vert .
\label{eq:dim-bnd-ent-ch}%
\end{equation}

\end{proposition}

\begin{proof}
Using the identity $D(\rho\Vert c\sigma)=D(\rho\Vert\sigma)-\log_{2}c$, for a
constant $c>0$, and the fact that the conditional entropy $H(B|R)_{\mathcal{N}%
(\psi)}=-D(\mathcal{N}_{A\rightarrow B}(\psi_{RA})\Vert\psi_{R}\otimes I_{B}%
)$, we find that%
\begin{align}
H(\mathcal{N})  &  =\log_{2}\left\vert B\right\vert -D(\mathcal{N}%
\Vert\mathcal{R})\label{eq:ent-to-CB-ent-1}\\
&  =\log_{2}\left\vert B\right\vert -\sup_{\psi_{RA}}D(\mathcal{N}%
_{A\rightarrow B}(\psi_{RA})\Vert\mathcal{R}_{A\rightarrow B}(\psi_{RA}))\\
&  =\log_{2}\left\vert B\right\vert -\sup_{\psi_{RA}}D(\mathcal{N}%
_{A\rightarrow B}(\psi_{RA})\Vert\psi_{R}\otimes\pi_{B})\\
&  =-\sup_{\psi_{RA}}D(\mathcal{N}_{A\rightarrow B}(\psi_{RA})\Vert\psi
_{R}\otimes I_{B})\\
&  =\inf_{\psi_{RA}}H(B|R)_{\mathcal{N}(\psi)}\label{eq:ent-to-CB-ent-last}\\
&  =H_{\operatorname{CB},\min}(\mathcal{N}).
\end{align}
We can then conclude the dimension bound in \eqref{eq:dim-bnd-ent-ch} from the
fact that it holds uniformly for the conditional entropy $\left\vert
H(B|R)\right\vert \leq\log_{2}\left\vert B\right\vert $. Defining $\tau
_{RBE}=\mathcal{U}_{A\rightarrow BE}^{\mathcal{N}}(\psi_{RA})$, from the
identity
\begin{equation}
H(B|R)_{\tau}=H(BR)_{\tau}-H(R)_{\tau}=H(E)_{\tau}-H(A)_{\rho},
\end{equation}
for $\rho_{A}=\operatorname{Tr}_{R}\{\psi_{RA}\}$, and where we used $\tau
_{R}=\psi_{R}$, we have that%
\begin{equation}
H(\mathcal{N})=G(\mathcal{N}_{A\rightarrow E}^{c})\equiv\inf_{\rho_{A}}\left[
H(E)_{\tau}-H(A)_{\rho}\right]  .
\end{equation}
We finally conclude that%
\begin{equation}
H(\mathcal{N})=-\sup_{\rho_{A}}H(B|E)_{\tau},
\label{eq:duality-for-channel-entropy}%
\end{equation}
which follows from the identity (duality of conditional entropy)%
\begin{equation}
H(B|R)_{\omega}=-H(B|E)_{\mathcal{U}(\psi_{A})}.
\end{equation}
This concludes the proof.
\end{proof}

\bigskip

\begin{remark}
\label{rem:channel-entr-bounds-interp}
We note here, as observed in \cite{DJKR06}, that the dimension lower bound
$H(\mathcal{N})\geq-\log_{2}\left\vert B\right\vert $ is saturated by the
identity channel, while the dimension upper bound $H(\mathcal{N})\leq\log
_{2}\left\vert B\right\vert $\ is saturated for the completely randomizing
(depolarizing) channel, which sends every state to the maximally mixed state.
Also, the entropy $H(\mathcal{N})$ is equal to zero for a replacer channel
that replaces the input with a pure quantum state. It is also known that the entropy of a channel is non-negative for all entanglement-breaking channels, as shown in \cite{DJKR06}. This includes all classical channels.

Thus, unlike entropy of a quantum state, the entropy of a quantum channel can
be negative. This negativity captures the ability of the channel to distill
quantum entanglement, in a sense made precise by the quantum channel merging
theorem stated as Theorem~\ref{thm:opt-int-ent-ch} in
Section~\ref{sec:q-ch-merge}. In the previous subsection we saw that for a
replacer channel with pure output state, the entropy of a channel is zero.
This replacer channel is also entanglement breaking. On the other hand, the
identity channel is the least noisy channel, and therefore should have the
least entropy possible. Indeed, as stated above, for the identity channel, our
entropy function equals the negative of the logarithm of the dimension (which
is the smallest possible value).
\end{remark}

\begin{corollary}
For any quantum channel $\mathcal{N}_{A\to B}$
\begin{equation}
H(\mathcal{N})\geq-\log|A|\;.
\end{equation}
with equality if and only if $\mathcal{N}_{A\to B}$ is an isometry.
\end{corollary}

\begin{proof}
The proof that $-\log|A|$ is the smallest possible value follows trivially
from the well known bound $D\big(\mathcal{N}_{A\to B}(\psi_{RA}%
)\big\|\mathcal{R}_{A\to B}(\psi_{RA})\big)\leq\log|AB|$. Now, from the
proposition above
\begin{align}
H(\mathcal{N})=\inf_{\psi_{RA}}H(B|R)_{\mathcal{N}(\psi)}\;.
\end{align}
Therefore, the smallest possible value $-\log|A|$ is achieved if and only if
$\mathcal{N}_{A\to B}(\psi_{RA})$ is the maximally entangled state (recall
$|R|=|A|$). This is only possible if $|B|\geq|A|$ and $\mathcal{N}$ is an isometry.
\end{proof}

\section{Quantum channel merging}

\label{sec:q-ch-merge}

Given a bipartite state $\rho_{BE}$, the goal of quantum state merging is for
Bob to use forward classical communication to Eve, as well as entanglement, to
merge his share of the state with Eve's share \cite{HOW05,HOW07}. The optimal
rate of entanglement consumed is equal to the conditional
entropy$~H(B|E)_{\rho}$. Alternatively, the optimal rate of entanglement
\textit{gained} is equal to the conditional entropy $H(B|R)_{\psi}$, where
$\psi_{RBE}$ is a purification of $\rho_{BE}$.

In this section, we define a task, called \textit{quantum channel merging},
that can be considered a dynamical counterpart of state merging. Given a
quantum channel $\mathcal{N}_{A\rightarrow B}$ with isometric extension
$\mathcal{U}_{A\rightarrow BE}^{\mathcal{N}}$, the goal is for Bob to merge
his share of the channel with Eve's share. We find here that the entanglement
cost of the protocol is equal to $\sup_{\rho_{A}}H(B|E)_{\omega}$, where
$\omega_{BE}=\mathcal{U}_{A\rightarrow BE}^{\mathcal{N}}(\rho_{A})$.
Equivalently, by employing \eqref{eq:duality-for-channel-entropy}, the
entanglement gain of the protocol is equal to $H(\mathcal{N})$, the entropy of
the channel $\mathcal{N}_{A\rightarrow B}$. Thus, the main result of this
section is a direct operational interpretation of the entropy of a channel as
the entanglement gain in quantum channel merging. We note here that the
completely bounded entropy of \cite{DJKR06}\ (i.e., entropy of a channel)\ was
recently interpreted in terms of a cryptographic task in \cite{YHW18}.

We now specify the quantum channel merging information-processing task in
detail. Let $\mathcal{N}_{A\rightarrow B}$ be a quantum channel, and suppose
that $\mathcal{U}_{A\rightarrow BE}^{\mathcal{N}}$ is an isometric channel
extending it. Here, we think of the isometric channel $\mathcal{U}%
_{A\rightarrow BE}^{\mathcal{N}}$ as a broadcast channel (three-terminal
device), which connects a source to the receivers Bob and Eve. Suppose that a
source generates an arbitrary state $\psi_{RA^{n}}$ and then sends the $A$
systems through the isometric channel $(\mathcal{U}_{A\rightarrow
BE}^{\mathcal{N}})^{\otimes n}$, which transmits the $B$ systems to Bob and
the $E$ systems to Eve. The goal is for Bob to use free one-way local
operations and classical communication (one-way LOCC) in order to generate
ebits at the maximum rate possible, while also merging his systems with Eve's.

\begin{figure*}
\begin{center}
\includegraphics[
width=6in
]{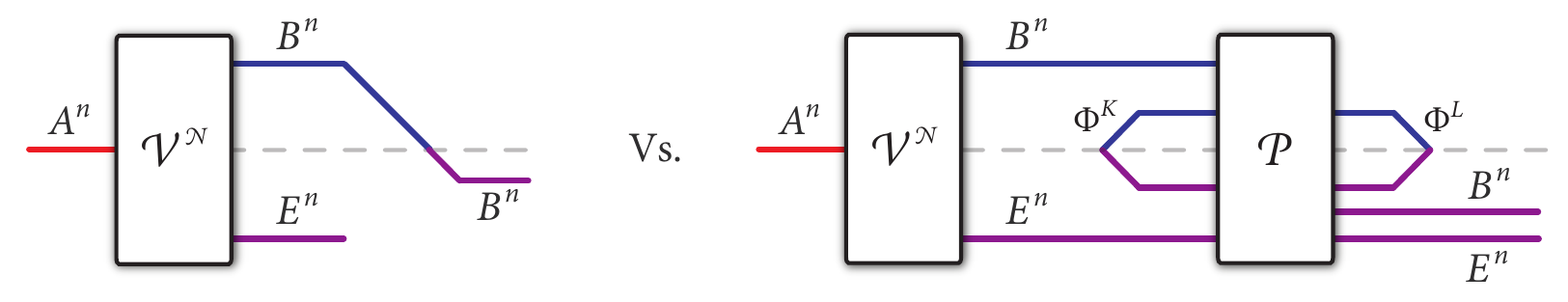}
\end{center}
\caption{The goal of quantum channel merging is for Bob to merge his share of
the channel with Eve's. Given a channel $\mathcal{N}_{A\rightarrow B}$, let
$\mathcal{V}^{\mathcal{N}}\equiv(\mathcal{U}_{A\rightarrow BE}^{\mathcal{N}%
})^{\otimes n}$, where $\mathcal{U}_{A\rightarrow BE}^{\mathcal{N}}$ is an
isometric channel extending $\mathcal{N}_{A\rightarrow B}$. By consuming a
maximally entangled state $\Phi^{K}$ of Schmidt rank $K$ and applying a
one-way LOCC\ protocol $\mathcal{P\,}$, Bob and Eve can distill a maximally
entangled state $\Phi^{L}$ of Schmidt rank $L$ and transfer Bob's systems
$B^{n}$ to Eve, in such a way that any third party having access to the inputs
$A^{n}$ and the outputs $B^{n}$ and $E^{n}$ would not be able to distinguish
the difference between the ideal situation on the left and the simulation on
the right. Theorem~\ref{thm:opt-int-ent-ch} states that the optimal asymptotic
rate of entanglement gain is equal to the entropy of the channel $\mathcal{N}%
$.}%
\label{fig:ch-merge}%
\end{figure*}

Let $n\in\mathbb{N}$, $M\in\mathbb{Q}$, and $\varepsilon\in\left[  0,1\right]
$. An $(n,M,\varepsilon)$ protocol for this task consists of a one-way
LOCC\ channel $\mathcal{P}_{B^{n}E^{n}\overline{B}_{0}\overline{E}%
_{0}\rightarrow\widetilde{B}_{E}^{n}\widetilde{E}^{n}\overline{B}_{1}%
\overline{E}_{1}}$ such that\begin{widetext}
\begin{multline}
\sup_{\psi_{RA^{n}}}\frac{1}{2}\bigg\Vert[\operatorname{id}_{BE\rightarrow
\widetilde{B}_{E}\widetilde{E}}^{\otimes n}\circ(\mathcal{U}_{A\rightarrow
BE}^{\mathcal{N}})^{\otimes n}](\psi_{RA^{n}})\otimes\Phi_{\overline{B}%
_{1}\overline{E}_{1}}^{L}\label{eq:criterion-CM}\\
-\mathcal{P}_{B^{n}E^{n}\overline{B}_{0}\overline{E}_{0}\rightarrow
\widetilde{B}_{E}^{n}\widetilde{E}^{n}\overline{B}_{1}\overline{E}_{1}%
}\!\left(  [(\mathcal{U}_{A\rightarrow BE}^{\mathcal{N}})^{\otimes n}%
(\psi_{RA^{n}})]\otimes\Phi_{\overline{B}_{0}\overline{E}_{0}}^{K}\right)
\bigg\Vert_{1}\leq\varepsilon,
\end{multline}
\end{widetext}
where $\Phi_{\overline{B}_{0}\overline{E}_{0}}^{K}$ and $\Phi_{\overline
{B}_{1}\overline{E}_{1}}^{L}$ are maximally entangled states of Schmidt rank
$K$ and $L$, respectively and $M=L/K$, so that the number of ebits gained in
the protocol is equal to $\log_{2}M=\log_{2}L-\log_{2}K$.
Figure~\ref{fig:ch-merge}\ depicts the task of quantum channel merging.

\begin{definition}
[Q.~channel merging capacity]A rate $R$ is achievable for quantum channel
merging if for all $\varepsilon\in(0,1]$, $\delta>0$, and sufficiently large
$n$, there exists an $(n,2^{n\left[  R-\delta\right]  },\varepsilon)$ protocol
of the above form. The quantum channel merging capacity $C_{M}(\mathcal{N})$
is defined to be the supremum of all achievable rates:%
\begin{multline}
C_{M}(\mathcal{N})\equiv\\
\sup\left\{  R\ |\ R\text{ is achievable for channel merging on }%
\mathcal{N}\right\}  .
\end{multline}

\end{definition}

\begin{theorem}
\label{thm:opt-int-ent-ch} The quantum channel merging capacity of a channel
$\mathcal{N}$ is equal to its entropy:%
\begin{equation}
C_{M}(\mathcal{N})=H(\mathcal{N}).
\end{equation}

\end{theorem}

We provide a detailed proof of Theorem~\ref{thm:opt-int-ent-ch} in
Appendix~\ref{sec:channel-merging-proof}.

\section{Examples}

\label{sec:examples}

In this section, we provide formulas for the entropy of several fundamental
channel models, including erasure channels, dephasing channels, depolarizing
channels, and Werner--Holevo channels. We also define the energy-constrained
and unconstrained entropies of a channel and determine formulas for them for
common bosonic channel models, including thermal, amplifier, and
additive-noise channels.

\subsection{Finite-dimensional channels}

A first observation to make is that, for any finite-dimensional channel, it is
an \textquotedblleft easy\textquotedblright\ optimization task to calculate
its entropy. This is a consequence of the identity $H(\mathcal{N})=-\sup
_{\rho_{A}}H(B|E)_{\mathcal{U}(\rho)}$\ from
Proposition~\ref{prop:alt-reps-ent-ch}\ and the concavity of conditional
entropy \cite{LR73,PhysRevLett.30.434} (in this context, see also
\cite[Eq.~(3.19)]{AC97}). Thus, one can exploit numerical optimizations to
calculate it \cite{Fawzi2018,FF18}.

For channels with symmetry, it can be much easier to evaluate a channel's
entropy, following from some observations from, e.g., \cite[Section~6]{KW17a}.
Let us begin by recalling the notion of a covariant channel $\mathcal{N}%
_{A\rightarrow B}$\ \cite{Hol02}. For a group $G$ with unitary channel
representations $\{\mathcal{U}_{A}^{g}\}_{g}$ and $\{\mathcal{V}_{B}^{g}%
\}_{g}$ acting on the input system $A$ and output system $B$ of the channel
$\mathcal{N}_{A\rightarrow B}$, the channel $\mathcal{N}_{A\rightarrow B}$ is
covariant with respect to the group$~G$ if the following equality holds for
all $g\in G$:%
\begin{equation}
\mathcal{N}_{A\rightarrow B}\circ\mathcal{U}_{A}^{g}=\mathcal{V}_{B}^{g}%
\circ\mathcal{N}_{A\rightarrow B}. \label{eq:cov-to-help}%
\end{equation}
If the averaging channel is such that $\frac{1}{\left\vert G\right\vert }%
\sum_{g}\mathcal{U}_{A}^{g}(X)=\operatorname{Tr}[X]I/\left\vert A\right\vert $
(implementing a unitary one-design), then we simply say that the channel
$\mathcal{N}_{A\rightarrow B}$ is covariant. It turns out that the entropy of
a channel is simple to calculate for covariant channels, with the optimal
$\psi_{RA}$ in \eqref{eq:cb-min} being the maximally entangled state, or
equivalently, the optimal $\rho_{A}$ in $-\sup_{\rho_{A}}H(B|E)_{\mathcal{U}%
(\rho)}$\ being the maximally mixed state.

\begin{proposition}
\label{prop:ch-ent-covariance}Let $\mathcal{N}_{A\rightarrow B}$ be a quantum
channel that is covariant with respect to a group $G$, in the sense of
\eqref{eq:cov-to-help}, and let $\mathcal{U}_{A\rightarrow BE}^{\mathcal{N}}$
be an isometric channel extending it. Then it suffices to perform the
optimization for the entropy of a channel over states that respect the
symmetry of the channel:%
\begin{equation}
H(\mathcal{N})=-\sup_{\rho_{A}=\mathcal{S}_{A}(\rho_{A})}H(B|E)_{\mathcal{U}%
(\rho)},
\end{equation}
where the symmetrizing channel $\mathcal{S}_{A}=\frac{1}{\left\vert
G\right\vert }\sum_{g\in G}\mathcal{U}_{A}^{g}$. Thus, if a channel is
covariant, then $H(\mathcal{N})=-H(B|E)_{\mathcal{U}(\pi)}$; i.e., the optimal
state $\rho_{A}$ is the maximally mixed state $\pi_{A}$.
\end{proposition}

\begin{proof}
First recall from Proposition~\ref{prop:alt-reps-ent-ch}\ that $H(\mathcal{N}%
)=-\sup_{\rho_{A}}H(B|E)_{\mathcal{U}(\rho)}$. Let $\rho_{A}$ be an arbitrary
state. If a channel $\mathcal{N}_{A\rightarrow B}$ is covariant as in
\eqref{eq:cov-to-help}, then it is known that there exists a unitary channel
$\mathcal{W}_{E}^{g}$ such that \cite{H06,H12}%
\begin{equation}
\mathcal{U}_{A\rightarrow BE}^{\mathcal{N}}\circ\mathcal{U}_{A}^{g}%
=(\mathcal{V}_{B}^{g}\otimes\mathcal{W}_{E}^{g})\circ\mathcal{U}_{A\rightarrow
BE}^{\mathcal{N}}.
\end{equation}
See also \cite[Appendix~A]{DBW17}\ for a simple proof. Then we find that%
\begin{align}
H(B|E)_{\mathcal{U}(\rho)}  &  =H(B|E)_{(\mathcal{V}^{g}\otimes\mathcal{W}%
^{g})\mathcal{U}(\rho)}\\
&  =\frac{1}{\left\vert G\right\vert }\sum_{g\in G}H(B|E)_{(\mathcal{V}%
^{g}\otimes\mathcal{W}^{g})\mathcal{U}(\rho)}\\
&  =\frac{1}{\left\vert G\right\vert }\sum_{g\in G}H(B|E)_{(\mathcal{U}%
\circ\mathcal{U}^{g})(\rho)}\\
&  \leq H(B|E)_{(\mathcal{U}\circ\mathcal{S})(\rho)}.
\end{align}
The first equality follows from invariance of conditional entropy under the
action of a local unitary (the equality holds for all $g\in G$). The third
equality follows from channel covariance. The inequality follows from
concavity of conditional entropy \cite{LR73,PhysRevLett.30.434}.
\end{proof}

\bigskip

A simple example of a channel that is covariant is the quantum erasure
channel, defined as \cite{GBP97}%
\begin{equation}
\mathcal{E}^{p}(\rho)\equiv(1-p)\rho+p|e\rangle\langle e|,
\end{equation}
where $\rho$ is a $d$-dimensional input state, $p\in\left[  0,1\right]  $ is
the erasure probability, and $|e\rangle\langle e|$ is a pure erasure state
orthogonal to any input state, so that the output state has $d+1$ dimensions.
A $d$-dimensional dephasing channel has the following action:%
\begin{equation}
\mathcal{D}^{\mathbf{p}}(\rho)=\sum_{\ell=0}^{d-1}p_{\ell}Z^{\ell}\rho
Z^{\ell\dag},
\end{equation}
where $\mathbf{p}$ is a vector containing the probabilities $p_{\ell}$ and $Z$
has the following action on the computational basis $Z|x\rangle=e^{2\pi
ix/d}|x\rangle$. This channel is covariant with respect to the
Heisenberg--Weyl group of unitaries, which is well known to form a unitary
one-design. A particular kind of Werner--Holevo channel performs the following
transformation on a $d$-dimensional input state $\rho$ \cite{WH02}:%
\begin{equation}
\mathcal{W}^{(d)}(\rho)\equiv\frac{1}{d-1}\left(  \operatorname{Tr}%
\{\rho\}I-T(\rho)\right)  ,
\end{equation}
where $d\geq2$ and $T$ denotes the transpose map $T(\cdot)=\sum_{i,j}%
|i\rangle\langle j|(\cdot)|i\rangle\langle j|$. As observed in
\cite[Section~II]{WH02}, this channel is covariant. The $d$-dimensional
depolarizing channel is a common model of noise in quantum information,
transmitting the input state with probability $1-p\in\left[  0,1\right]  $ and
replacing it with the maximally mixed state $\pi\equiv\frac{I}{d}$ with
probability~$p$:%
\begin{equation}
\Delta^{p}(\rho)=\left(  1-p\right)  \rho+p\pi.
\end{equation}

By applying Proposition~\ref{prop:ch-ent-covariance}\ and evaluating the
resulting entropy $-H(B|E)$ for each of the above channels when the maximally
mixed state $\pi$ is input, we arrive at the following formulas:%
\begin{align}
H(\mathcal{E}^{p})  &  =h_{2}(p)+\left(  p-1\right)  \log_{2}d,\\
H(\mathcal{D}^{\mathbf{p}})  &  =H(\mathbf{p})-\log_{2}d,\\
H(\mathcal{W}^{(d)})  &  =\log_{2}\left[  (d-1)/2\right]  ,\\
H(\Delta^{p})  &  =-\left(  1-p+\frac{p}{d^{2}}\right)  \log_{2}\left(
1-p+\frac{p}{d^{2}}\right) \nonumber\\
&  \qquad-\left(  d^{2}-1\right)  \frac{p}{d^{2}}\log_{2}\frac{p}{d^{2}}%
-\log_{2}d,
\end{align}
where $H(\mathbf{p})$ is the Shannon entropy of the probability vector
$\mathbf{p}$. These formulas are plotted and interpreted in
Figures~\ref{fig:erasure-example}--\ref{fig:depolarizing-example}.

\begin{figure}
\begin{center}
\includegraphics[
width=3.3in
]{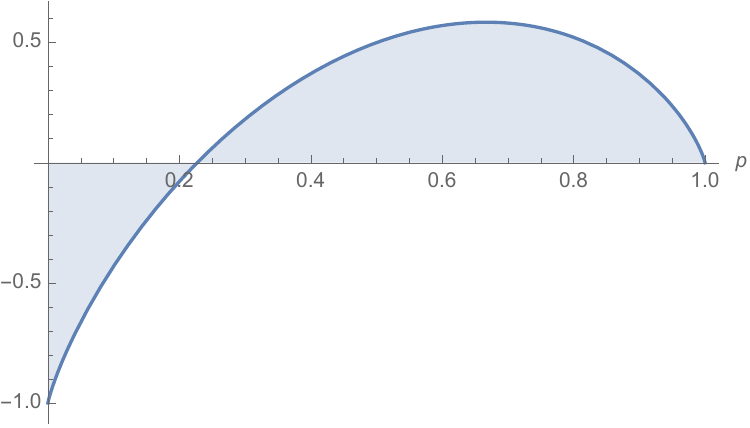}
\end{center}
\caption{Entropy of the qubit erasure channel as a function of the erasure
probability $p$. When $p=0$, the erasure channel is the identity qubit channel
and thus takes on its smallest value. When $p=1$, the erasure channel
deterministically replaces the input with the pure state $|e\rangle\langle e|$
and thus has entropy equal to zero.}%
\label{fig:erasure-example}%
\end{figure}

\begin{figure}
\begin{center}
\includegraphics[
width=3.3in
]{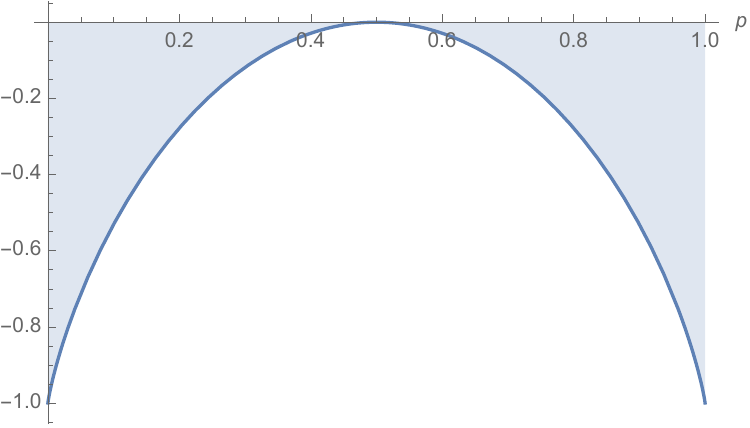}
\end{center}
\caption{Entropy of the qubit dephasing channel as a function of the dephasing
probability $p$. The optimal input state is the maximally entangled state, so
that the channel entropy is evaluated on the Choi state of the channel. When
$p=0$, the dephasing channel is the identity qubit channel and thus takes on
its smallest value. When $p=1/2$, the dephasing channel is a classical
channel, so that its Choi state is maximally classically correlated. For such
a state, $D(\mathcal{N}\Vert\mathcal{R})=1$ so that the channel entropy is
equal to zero.}%
\label{fig:dephasing-example}%
\end{figure}

\begin{figure}
\begin{center}
\includegraphics[
width=3.3in
]{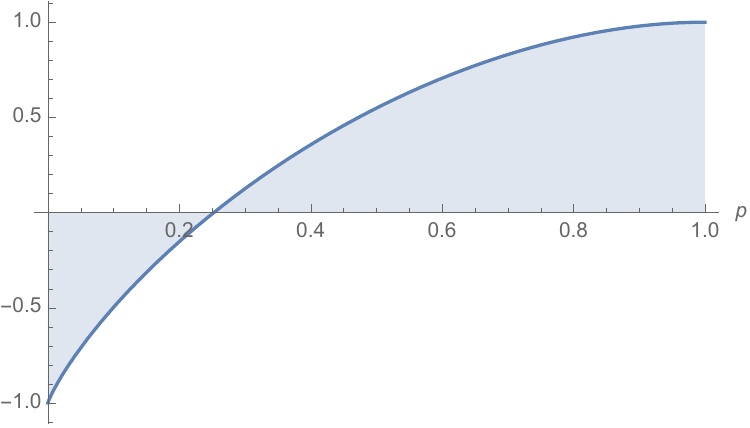}
\end{center}
\caption{Entropy of the qubit depolarizing channel as a function of the
depolarizing probability $p$. When $p=0$, the depolarizing channel is the
identity qubit channel and thus takes on its smallest value. When $p=1$, the
depolarizing channel replaces the channel input with the maximally mixed state
and thus takes on its maximal value.}%
\label{fig:depolarizing-example}%
\end{figure}

\subsection{Energy-constrained entropy of a channel}

We can define the energy-constrained entropy of a channel for
infinite-dimensional systems, by employing the identity in
Proposition~\ref{prop:alt-reps-ent-ch}\ and the definition of conditional
entropy from \cite{K11}.

To review the definition from \cite{K11}, recall that the quantum entropy of a
state $\rho$ acting on a separable Hilbert space is defined as%
\begin{equation}
H(\rho)\equiv\operatorname{Tr}\{\eta(\rho)\},
\end{equation}
where $\eta(x)=-x\log_{2}x$ if $x>0$ and $\eta(0)=0$. The trace in the above
equation can be taken with respect to any countable orthonormal basis of
$\mathcal{H}$ \cite[Definition~2]{AL70}. The quantum entropy is a
non-negative, concave, lower semicontinuous function on $\mathcal{D}%
(\mathcal{H})$ \cite{W76}. It is also not necessarily finite (see, e.g.,
\cite{BV13}). When $\rho_{A}$ is assigned to a system $A$, we write
$H(A)_{\rho}\equiv H(\rho_{A})$. Recall that the relative entropy of two
states $\rho$ and $\sigma$ acting on a separable Hilbert space is given by
\cite{F70,Lindblad1973}%
\begin{multline}
D(\rho\Vert\sigma)\equiv\\
\lbrack\ln2]^{-1}\sum_{i,j}|\langle\phi_{i}|\psi_{j}\rangle|^{2}%
[p(i)\ln\!\left(  \frac{p(i)}{q(j)}\right)  +q(j)-p(i)],
\end{multline}
where $\rho=\sum_{i}p(i)|\phi_{i}\rangle\langle\phi_{i}|$ and $\sigma=\sum
_{j}q(j)|\psi_{j}\rangle\langle\psi_{j}|$ are spectral decompositions of
$\rho$ and $\sigma$ with $\{|\phi_{i}\rangle\}_{i}$ and $\{|\psi_{j}%
\rangle\}_{j}$ orthonormal bases. The prefactor $[\ln2]^{-1}$ is there to
ensure that the units of the quantum relative entropy are bits. For a
bipartite state $\rho_{AB}$, the mutual information is defined as%
\begin{equation}
I(A;B)_{\rho}\equiv D(\rho_{AB}\Vert\rho_{A}\otimes\rho_{B}).
\end{equation}
Finally, for a bipartite state $\rho_{AB}$ such that $H(A)_{\rho}<\infty$, the
conditional entropy is defined as \cite{K11}%
\begin{equation}
H(A|B)_{\rho}\equiv H(A)_{\rho}-I(A;B)_{\rho},
\end{equation}
and it is known that $H(A|B)_{\rho}\in\lbrack-H(A)_{\rho},H(A)_{\rho}]$
\cite{K11}.

A Gibbs observable is a positive semi-definite operator $G$ acting on a
separable Hilbert space such that $\operatorname{Tr}\{e^{-\beta G}\}<\infty$
for all $\beta>0$ \cite{H03,H04,H12}. This condition for a Gibbs observable
means that there is always a well defined thermal state.

Finally, we say that a quantum channel $\mathcal{N}_{A\rightarrow B}$\ obeys
the finite-output entropy condition \cite{H03,H04,H12}\ with respect to a
Gibbs observable $G$\ if for all $P\geq0$, the following inequality holds%
\begin{equation}
\sup_{\rho_{A}:\operatorname{Tr}\{G\rho_{A}\}\leq P}H(B)_{\mathcal{N}(\rho
)}<\infty.
\end{equation}

We now define the energy-constrained and unconstrained channel entropy as follows:

\begin{definition}
Let $\mathcal{N}_{A\rightarrow B}$ be a quantum channel that satisfies the
finite-output entropy condition with respect to a Gibbs observable$~G$. For
$P\geq0$, the energy-constrained entropy of $\mathcal{N}_{A\rightarrow B}$ is
defined as%
\begin{equation}
H(\mathcal{N},G,P)\equiv\inf_{\psi_{RA}:\operatorname{Tr}\{G\psi_{A}\}\leq
P}H(B|R)_{\omega},
\end{equation}
where $\omega_{RB}\equiv\mathcal{N}_{A\rightarrow B}(\psi_{RA})$ and the
optimization is with respect to all pure bipartite states with system $R$
isomorphic to system $A$.\ The unconstrained entropy of $\mathcal{N}%
_{A\rightarrow B}$ with respect to $G$ is then defined as%
\begin{equation}
H(\mathcal{N},G)\equiv\inf_{P\geq0}H(\mathcal{N},G,P).
\end{equation}

\end{definition}

\subsection{Bosonic Gaussian channels}

In this section, we evaluate the energy-constrained and unconstrained entropy
of several important bosonic Gaussian channels \cite{H12,S17}, including the
thermal, amplifier, and additive-noise channels. Here we take the Gibbs
observable to be the photon number operator $\hat{n}$ \cite{H12,S17}, and we
note that each of these channels satisfies the finite-output entropy condition
mentioned above. From a practical perspective, we should be most interested in
these particular single-mode bosonic Gaussian channels, as these are of the
greatest interest in applications, as stressed in \cite[Section~12.6.3]{H12}
and \cite[Section~3.5]{HG12}. Each of these are defined respectively by the
following Heisenberg input-output relations:%
\begin{align}
\hat{b}  &  =\sqrt{\eta}\hat{a}+\sqrt{1-\eta}\hat{e}%
,\label{eq:thermal-channel}\\
\hat{b}  &  =\sqrt{G}\hat{a}+\sqrt{G-1}\hat{e}^{\dag}%
,\label{eq:amplifier-channel}\\
\hat{b}  &  =\hat{a}+\left(  x+ip\right)  /\sqrt{2},
\label{eq:additive-noise-channel}%
\end{align}
where $\hat{a}$, $\hat{b}$, and $\hat{e}$ are the field-mode annihilation
operators for the sender's input, the receiver's output, and the environment's
input of these channels, respectively.

The channel in \eqref{eq:thermal-channel} is a thermalizing channel, in which
the environmental mode is prepared in a thermal state $\theta(N_{B})$\ of mean
photon number $N_{B}\geq0$, defined as%
\begin{equation}
\theta(N_{B})\equiv\frac{1}{N_{B}+1}\sum_{n=0}^{\infty}\left(  \frac{N_{B}%
}{N_{B}+1}\right)  ^{n}|n\rangle\langle n|, \label{eq:bosonic-thermal-st}%
\end{equation}
where $\left\{  |n\rangle\right\}  _{n=0}^{\infty}$ is the orthonormal,
photonic number-state basis. When $N_{B}=0$, $\theta(N_{B})$ reduces to the
vacuum state, in which case the resulting channel in
\eqref{eq:thermal-channel} is called the pure-loss channel. The parameter
$\eta\in(0,1)$ is the transmissivity of the channel, representing the average
fraction of photons making it from the input to the output of the channel. Let
$\mathcal{L}_{\eta,N_{B}}$ denote this channel.

The channel in \eqref{eq:amplifier-channel} is an amplifier channel, and the
parameter $G>1$ is its gain. For this channel, the environment is prepared in
the thermal state $\theta(N_{B})$. If $N_{B}=0$, the amplifier channel is
called the pure-amplifier channel. Let $\mathcal{A}_{G,N_{B}}$ denote this channel.

Finally, the channel in \eqref{eq:additive-noise-channel} is an additive-noise
channel, representing a quantum generalization of the classical additive white
Gaussian noise channel. In \eqref{eq:additive-noise-channel}, $x$ and $p$ are
zero-mean, independent Gaussian random variables each having variance $\xi
\geq0$. Let $\mathcal{T}_{\xi}$ denote this channel. Note that the
additive-noise channel arises from the thermal channel in the limit
$\eta\rightarrow1$, $N_{B}\rightarrow\infty$, but with $\left(  1-\eta\right)
N_{B}\rightarrow\xi$ \cite{GGLMS04}.

Kraus representations for the channels in
\eqref{eq:thermal-channel}--\eqref{eq:additive-noise-channel}\ are available
in \cite{ISS11}, which can be helpful for further understanding their action
on input quantum states.

All of the above channels are phase-insensitive or phase-covariant Gaussian
channels \cite{H12,S17}. Let $N_{S}\geq0$. Since the function $\sup
_{\rho:\operatorname{Tr}\{\hat{n}\rho\}\leq N_{S}}H(B|E)_{\mathcal{U}(\rho)}$
we are evaluating is concave in the input and invariant under local unitaries,
\cite[Remark~22]{SWAT17} applies, implying that the optimal input state for
the entropies of these channels is the bosonic thermal state $\theta(N_{S})$.
We then find by employing well known entropy formulas from \cite{HW01,GLMS03}
(see also \cite{WHG12} in this context)\ that%
\begin{multline}
H(\mathcal{L}_{\eta,N_{B}},\hat{n},N_{S}) =\\
g_{2}(\left[  D_{1}+\left(  1-\eta\right)  \left(  N_{S}-N_{B}\right)
-1\right]  /2)\\
+g_{2}(\left[  D_{1}-\left(  1-\eta\right)  \left(  N_{S}-N_{B}\right)
-1\right]  /2)-g_{2}(N_{S}),
\end{multline}
\begin{multline}
H(\mathcal{A}_{G,N_{B}},\hat{n},N_{S}) =\\
g_{2}(\left[  D_{2}+\left(  G-1\right)  \left(  N_{S}+N_{B}+1\right)
-1\right]  /2)\\
+g_{2}(\left[  D_{2}-\left(  G-1\right)  \left(  N_{S}+N_{B}+1\right)
-1\right]  /2)-g_{2}(N_{S}),
\end{multline}
\begin{multline}
H(\mathcal{T}_{\xi},\hat{n},N_{S}) =g_{2}(\left[  D_{3}-\left(  \xi+1\right)
\right]  /2)\\
+g_{2}(\left[  D_{3}+\xi-1\right]  /2)-g_{2}(N_{S}),
\end{multline}
where $g_{2}$ is the bosonic entropy function defined in
\eqref{eq:bosonic-entropy}\ and\begin{widetext}
\begin{align}
D_{1}  &  \equiv\sqrt{\left[  \left(  \eta+1\right)  N_{S}+\left(
1-\eta\right)  N_{B}+1\right]  ^{2}-4\eta N_{S}\left(  N_{S}+1\right)  },\\
D_{2}  &  \equiv\sqrt{\left[  \left(  G+1\right)  N_{S}+\left(  G-1\right)
\left(  N_{B}+1\right)  +1\right]  ^{2}-4GN_{S}\left(  N_{S}+1\right)  },\\
D_{3}  &  \equiv\sqrt{\left(  \xi+1\right)  ^{2}+4\xi N_{S}}.
\end{align}
\end{widetext}
Note that we arrived at the formula for $H(\mathcal{T}_{\xi},\hat{n},N_{S})$
by considering the limit discussed above. Furthermore, by the same reasoning
as given in \cite[Section~6]{SWAT17}, these functions are decreasing with
increasing $N_{S}$, and so we find that%
\begin{align}
H(\mathcal{L}_{\eta,N_{B}},\hat{n})  &  =\inf_{N_{S}\geq0}H(\mathcal{L}%
_{\eta,N_{B}},\hat{n},N_{S})\\
&  =\lim_{N_{S}\rightarrow\infty}H(\mathcal{L}_{\eta,N_{B}},\hat{n},N_{S}),\\
H(\mathcal{A}_{G,N_{B}},\hat{n})  &  =\inf_{N_{S}\geq0}H(\mathcal{A}_{G,N_{B}%
},\hat{n},N_{S})\\
&  =\lim_{N_{S}\rightarrow\infty}H(\mathcal{A}_{G,N_{B}},\hat{n},N_{S}),\\
H(\mathcal{T}_{\xi},\hat{n})  &  =\inf_{N_{S}\geq0}H(\mathcal{T}_{\xi},\hat
{n},N_{S})\\
&  =\lim_{N_{S}\rightarrow\infty}H(\mathcal{T}_{\xi},\hat{n},N_{S}),
\end{align}
which leads to the following formulas for the unconstrained entropies of the
channels:%
\begin{align}
H(\mathcal{L}_{\eta,N_{B}},\hat{n})  &  =\log_{2}(1-\eta)+g_{2}(N_{B}%
),\label{eq:ent-unc-therm}\\
H(\mathcal{A}_{G,N_{B}},\hat{n})  &  =\log_{2}(G-1)+g_{2}(N_{B}%
),\label{eq:ent-unc-amp}\\
H(\mathcal{T}_{\xi},\hat{n})  &  =\log_{2}(\xi)+\frac{1}{\ln2}.
\label{eq:ent-unc-add-noise}%
\end{align}
These formulas are plotted and interpreted in
Figures~\ref{fig:thermal-example}--\ref{fig:additive-noise-example}. A
Mathematica file is available with the arXiv posting of this paper to automate
these calculations, but we note here that the expansion $g_{2}(x)=\log
_{2}(x)+1/\ln2+O(1/x)$ is helpful for this purpose. We also note that the
formulas in \eqref{eq:ent-unc-therm}--\eqref{eq:ent-unc-amp}\ were presented
in \cite[Eq.~(2)]{PGBL09}\ and the formula in
\eqref{eq:ent-unc-add-noise}\ was presented in \cite[Section~V]{HW01}.

\begin{figure}
\begin{center}
\includegraphics[
width=3.3in
]{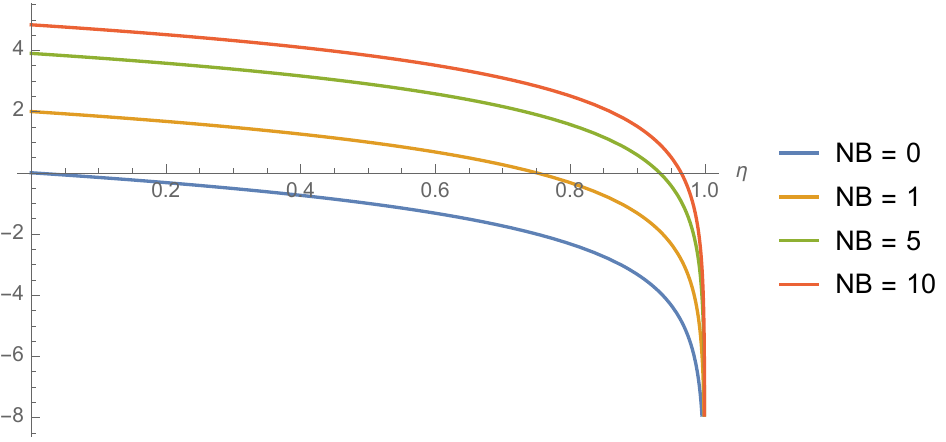}
\end{center}
\caption{Entropy of the bosonic thermal channel as a function of the
transmissivity $\eta$. When $\eta=1$, the thermal channel is the identity
channel and thus takes on its smallest value of $-\infty$, regardless of the
value of $N_{B}$. When $\eta=0$ and $N_{B}=0$, the thermal channel
deterministically replaces the input with the pure vacuum state $|0\rangle
\langle0|$ and thus has entropy equal to zero. As the thermal noise $N_{B}$
increases, the entropy of the thermal channel increases.}%
\label{fig:thermal-example}%
\end{figure}

\begin{figure}
\begin{center}
\includegraphics[
width=3.3in
]{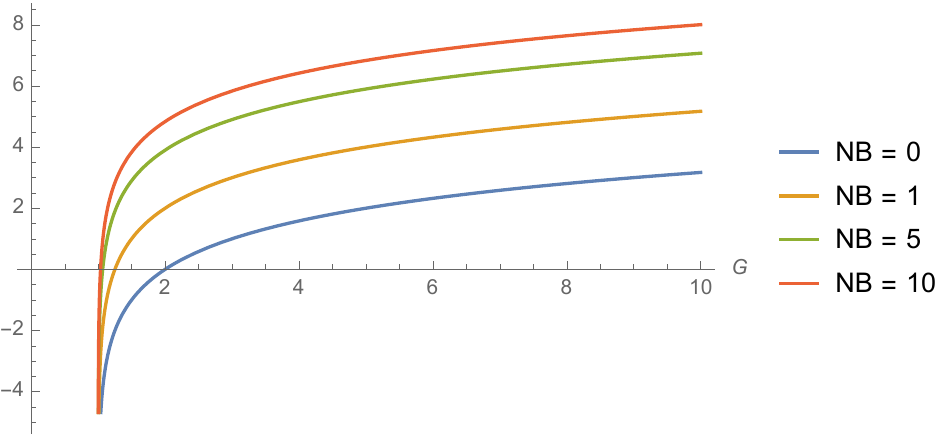}
\end{center}
\caption{Entropy of the bosonic amplifier channel as a function of the
amplifier gain $G$. When $G=1$, the amplifier channel is the identity channel
and thus takes on its smallest value of $-\infty$, regardless of the value of
$N_{B}$. As the amplifier gain $G$ and the thermal noise $N_{B}$ increase, the
entropy of the thermal channel increases.}%
\label{fig:amplifier-example}%
\end{figure}

\begin{figure}
\begin{center}
\includegraphics[
width=3.3in
]{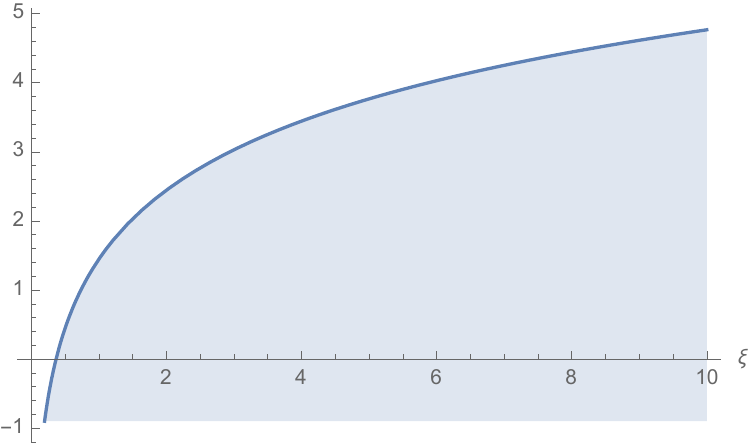}
\end{center}
\caption{Entropy of the bosonic additive-noise channel as a function of the
noise parameter $\xi$. When $\xi=0$, the additive-noise channel is the
identity channel and thus takes on its smallest value of $-\infty$. As $\xi$
increases, the entropy of the additive-noise channel increases
logarithmically.}%
\label{fig:additive-noise-example}%
\end{figure}

\section{R\'{e}nyi entropy of a quantum channel}

\label{sec:renyi-ent-ch}

Generalizing the von Neumann entropy of a quantum state, the R\'{e}nyi entropy
finds extensive application in physics and information theory. Given a pure
bipartite state, the R\'{e}nyi entropy of the reduced state is an entanglement
measure, which finds application in conformal field theory \cite{CC09},
holography \cite{NRT09}, and black holes \cite{Sol11}. The full range of
values of the R\'{e}nyi entropy is known as the entanglement spectrum. The
R\'{e}nyi entropy finds information-theoretic meaning in the expression for
the error exponent of entanglement concentration \cite{HKMMW02} and quantum
data compression \cite{Hay02}, indicating the exponential rate at which errors
in these settings decay to zero. As such, it is worthwhile to understand the
R\'{e}nyi entropy of a channel as a generalization of the R\'{e}nyi entropy of
a state.

In this section, we define the R\'{e}nyi entropy of a channel, following the
same approach discussed in the introduction. That is, we first write the
R\'{e}nyi entropy of a state as the difference of the number of physical
qubits and the R\'{e}nyi relative entropy of the state to the maximally mixed
state. Then we define the R\'{e}nyi entropy of a channel in the same way as in
Definition~\ref{def:entropy-channel}, but replacing the channel relative
entropy with the sandwiched R\'{e}nyi channel relative entropy from
\cite{CMW16}.

The R\'{e}nyi entropy of a quantum state $\rho_{A}$ of system $A$ is defined
for $\alpha\in(0,1)\cup(1,\infty)$ as%
\begin{align}
H_{\alpha}(A)_{\rho}  &  \equiv\frac{1}{1-\alpha}\log_{2}\operatorname{Tr}%
\{\rho_{A}^{\alpha}\}\\
&  =\frac{1}{1-\alpha}\log_{2}\left\Vert \rho_{A}\right\Vert _{\alpha}%
^{\alpha},
\end{align}
where $\left\Vert X\right\Vert _{\alpha}\equiv\lbrack\operatorname{Tr}%
\{\left\vert X\right\vert ^{\alpha}\}]^{1/\alpha}$ and $\left\vert
X\right\vert \equiv\sqrt{X^{\dag}X}$ for an operator $X$. The R\'{e}nyi
relative entropy of quantum states can be defined in two different ways, known
as the Petz--R\'{e}nyi relative entropy \cite{P85,P86} and the sandwiched
R\'{e}nyi relative entropy \cite{MDSFT13,WWY13}. The sandwiched R\'{e}nyi
relative entropy is defined for $\alpha\in(0,1)\cup(1,\infty)$, a state $\rho
$, and a positive semi-definite operator $\sigma$ as%
\begin{equation}
D_{\alpha}(\rho\Vert\sigma)\equiv\frac{1}{\alpha-1}\log_{2}\operatorname{Tr}%
\left\{  \left(  \sigma^{\left(  1-\alpha\right)  /2\alpha}\rho\sigma^{\left(
1-\alpha\right)  /2\alpha}\right)  ^{\alpha}\right\}  ,
\label{eq:sandwiched-Renyi}%
\end{equation}
whenever either $\alpha\in(0,1)$ or $\operatorname{supp}(\rho)\subseteq
\operatorname{supp}(\sigma)$ and $\alpha>1$. Otherwise, it is set to $+\infty
$. The sandwiched R\'{e}nyi relative entropy obeys the data processing
inequality for $\rho$ and $\sigma$ as above, a quantum channel $\mathcal{N}$,
and $\alpha\in\lbrack1/2,1)\cup(1,\infty)$ \cite{FL13} (see also
\cite{B13monotone,MO13,MDSFT13,WWY13,W18opt}):%
\begin{equation}
D_{\alpha}(\rho\Vert\sigma)\geq D_{\alpha}(\mathcal{N}(\rho)\Vert
\mathcal{N}(\sigma)). \label{eq:DP-sand}%
\end{equation}
It converges to the quantum relative entropy in the limit $\alpha\rightarrow1$
\cite{MDSFT13,WWY13}:%
\begin{equation}
\lim_{\alpha\rightarrow1}D_{\alpha}(\rho\Vert\sigma)=D(\rho\Vert\sigma).
\end{equation}
By inspection, the R\'{e}nyi entropy of a state can be written as%
\begin{equation}
H_{\alpha}(A)_{\rho}=\log_{2}\left\vert A\right\vert -D_{\alpha}(\rho_{A}%
\Vert\pi_{A}).
\end{equation}

The sandwiched R\'{e}nyi channel divergence of channels $\mathcal{N}%
_{A\rightarrow B}$ and $\mathcal{M}_{A\rightarrow B}$ is defined for
$\alpha\in\lbrack1/2,1)\cup(1,\infty)$ as~\cite{CMW16}%
\begin{equation}
D_{\alpha}(\mathcal{N}\Vert\mathcal{M})\equiv\sup_{\rho_{RA}}D_{\alpha
}(\mathcal{N}_{A\rightarrow B}(\rho_{RA})\Vert\mathcal{M}_{A\rightarrow
B}(\rho_{RA})),
\end{equation}
where the optimization is with respect to bipartite states $\rho_{RA}$ of a
reference system $R$ of arbitrary size and the channel input system~$A$. Due
to state purification, the data-processing inequality in
\eqref{eq:DP-sand},\ and the Schmidt decomposition theorem, it suffices to
optimize over states $\rho_{RA}$ that are pure and such that system $R$ is
isomorphic to system $A$.

We now define the R\'{e}nyi entropy of a quantum channel as follows:

\begin{definition}
[R\'{e}nyi entropy of a q.~channel]\label{def:renyi-ent-ch}Let $\mathcal{N}%
_{A\rightarrow B}$ be a quantum channel. For $\alpha\in\lbrack1/2,1)\cup
(1,\infty)$, the R\'{e}nyi entropy of the channel $\mathcal{N}$ is defined as%
\begin{equation}
H_{\alpha}(\mathcal{N})\equiv\log_{2}\left\vert B\right\vert -D_{\alpha
}(\mathcal{N}\Vert\mathcal{R}),
\end{equation}
where $\mathcal{R}_{A\rightarrow B}$ is the completely randomizing channel from~\eqref{eq:cm-random-ch}.
\end{definition}

We remark here that $D_{\alpha}(\mathcal{N}\Vert\mathcal{R})$, for $\alpha>1$,
has an operational interpretation as the strong converse exponent for
discrimination of the channel $\mathcal{N}_{A\rightarrow B}$ from the
completely randomizing channel $\mathcal{R}_{A\rightarrow B}$, when
considering any possible channel discrimination strategy \cite{CMW16}.

One could alternatively define a different R\'{e}nyi entropy of a channel
according to the above recipe, but in terms of the Petz--R\'{e}nyi relative
entropy. However, it is unclear whether the additivity property is generally
satisfied for the resulting R\'{e}nyi entropy of a channel, and so we do not
consider it further here, instead leaving this question open.

\subsection{Properties of the R\'{e}nyi entropy of a quantum channel}

The R\'{e}nyi entropy of a channel obeys the three desired axioms from
\cite{G18}, and in fact, the proofs are essentially the same as the previous
ones, but instead using properties of the sandwiched R\'{e}nyi relative entropy.

\begin{proposition}
Let $\mathcal{N}_{A\rightarrow B}$ be a quantum channel, and let $\Theta$ be a
uniformity preserving superchannel as defined above. Then for all
$[1/2,1)\cup(1,\infty)$:%
\begin{equation}
H_{\alpha}(\Theta(\mathcal{N}))\geq H_{\alpha}(\mathcal{N}).
\end{equation}

\end{proposition}

\begin{proof}
We follow the same steps as in
\eqref{eq:uniform-pres-1}--\eqref{eq:uniform-pres-last}, but making the
substitutions $H\rightarrow H_{\alpha}$ and $D\rightarrow D_{\alpha}$. Also,
we use the fact that, for $[1/2,1)\cup(1,\infty)$, the sandwiched R\'{e}nyi
channel divergence does not increase under the action of a superchannel, as
shown in \cite{G18}.
\end{proof}

\begin{proposition}
[Additivity]\label{prop:renyi-additive}Let $\mathcal{N}$ and $\mathcal{M}$ be
quantum channels. Then the channel R\'{e}nyi entropy is additive in the
following sense for $\alpha\in(1,\infty)$:%
\begin{equation}
H_{\alpha}(\mathcal{N}\otimes\mathcal{M})=H_{\alpha}(\mathcal{N})+H_{\alpha
}(\mathcal{M}).
\end{equation}

\end{proposition}

\begin{proof}
The proof here follows the same approach given in the proof of
Proposition~\ref{prop:entropy-additive}, making the substitutions
$H\rightarrow H_{\alpha}$ and $D\rightarrow D_{\alpha}$. The steps in
\eqref{eq:addit-1}--\eqref{eq:addit-last} follow from the same steps given in
the proof of Proposition~41 of \cite{BHKW18}, which in turn rely upon the
additivity result from \cite{DJKR06}. See also \cite{CMW16}\ in this context.
\end{proof}

\begin{proposition}
[Reduction to states]Let the channel $\mathcal{N}_{A\rightarrow B}$ be a
replacer channel, defined such that $\mathcal{N}_{A\rightarrow B}(\rho
_{A})=\sigma_{B}$ for all states $\rho_{A}$ and some state $\sigma_{B}$. Then
the following equality holds for all $\alpha\in(0,1)\cup(1,\infty)$:%
\begin{equation}
H_{\alpha}(\mathcal{N})=H_{\alpha}(B)_{\sigma}.
\end{equation}

\end{proposition}

\begin{proof}
The proof is essentially the same as the proof of
Proposition~\ref{prop:reduction-to-states}, making the substitutions
$H\rightarrow H_{\alpha}$ and $D\rightarrow D_{\alpha}$.
\end{proof}

\bigskip

We can then conclude that the R\'{e}nyi entropy of a channel satisfies the
normalization axiom from the fact that $H_{\alpha}(B)_{\sigma}=\log\left\vert
B\right\vert $ if $\sigma_{B}$ is maximally mixed and $H(B)_{\sigma}=0$ if
$\sigma_{B}$ is pure.

\subsection{Alternate representations for the R\'{e}nyi entropy of a quantum
channel}

Just as we showed in Section~\ref{sec:alt-reps-ent}\ that there are alternate
representations for the entropy of a quantum channel, here we do the same for
the R\'{e}nyi entropy of a channel. We define the conditional R\'{e}nyi
entropy of a bipartite state $\rho_{AB}$ as%
\begin{equation}
H_{\alpha}(A|B)_{\rho|\rho}\equiv-D_{\alpha}(\rho_{AB}\Vert I_{A}\otimes
\rho_{B}),
\end{equation}
where $D_{\alpha}(\rho\Vert\sigma)$ is the sandwiched R\'{e}nyi relative
entropy from \eqref{eq:sandwiched-Renyi}. The conditional Petz--R\'{e}nyi
entropy of a bipartite state $\rho_{AB}$ is defined as%
\begin{equation}
\overline{H}_{\alpha}(A|B)_{\rho}\equiv-\inf_{\sigma_{B}}\overline{D}_{\alpha
}(\rho_{AB}\Vert I_{A}\otimes\sigma_{B}),
\end{equation}
where the Petz--R\'{e}nyi relative entropy $\overline{D}_{\alpha}(\rho
\Vert\sigma)$ is defined for $\alpha\in(0,1)\cup(1,\infty)$ as \cite{P85,P86}%
\begin{equation}
\overline{D}_{\alpha}(\rho\Vert\sigma)\equiv\frac{1}{\alpha-1}\log
_{2}\operatorname{Tr}\left\{  \rho^{\alpha}\sigma^{1-\alpha}\right\}  ,
\end{equation}
whenever either $\alpha\in(0,1)$ or $\operatorname{supp}(\rho)\subseteq
\operatorname{supp}(\sigma)$ and $\alpha>1$. Otherwise, it is set to $+\infty
$. The Petz--R\'{e}nyi relative entropy obeys the data processing inequality
for $\rho$ and $\sigma$ as above, a quantum channel $\mathcal{N}$, and
$\alpha\in(0,1)\cup(1,2]$ \cite{P85,P86}:%
\begin{equation}
\overline{D}_{\alpha}(\rho\Vert\sigma)\geq\overline{D}_{\alpha}(\mathcal{N}%
(\rho)\Vert\mathcal{N}(\sigma)).
\end{equation}

The completely bounded $1\rightarrow\alpha$ norm of a quantum channel is
defined for $\alpha\geq1$ as \cite{DJKR06}%
\begin{equation}
\left\Vert \mathcal{N}_{A\rightarrow B}\right\Vert _{\text{CB},1\rightarrow
\alpha}\equiv\sup_{\rho_{R}}\left\Vert \rho_{R}^{1/2\alpha}\mathcal{N}%
_{A\rightarrow B}(\Gamma_{RA})\rho_{R}^{1/2\alpha}\right\Vert _{\alpha},
\end{equation}
where the optimization is with respect to a density operator $\rho_{R}$ and
$\Gamma_{RA}\equiv|\Gamma\rangle\langle\Gamma|_{RA}$ denotes the projection
onto the following maximally entangled vector:%
\begin{equation}
|\Gamma\rangle_{RA}\equiv\sum_{i}|i\rangle_{R}|i\rangle_{A},
\end{equation}
where $\{|i\rangle_{R}\}_{i}$ and $\{|i\rangle_{A}\}_{i}$ are orthonormal
bases and system$~R$ is isomorphic to the channel input system~$A$.

We can now state the alternate representations for the R\'{e}nyi entropy of a channel:

\begin{proposition}
\label{prop:identities-for-renyi}Let $\mathcal{N}_{A\rightarrow B}$ be a
quantum channel, and let $\mathcal{U}_{A\rightarrow BE}^{\mathcal{N}}$ be an
isometric channel extending it, as in \eqref{eq:iso-extend}. Then for
$\alpha\in(0,1)\cup(1,\infty)$,%
\begin{equation}
H_{\alpha}(\mathcal{N})=\inf_{\psi_{RA}}H_{\alpha}(B|R)_{\omega|\omega}%
=-\sup_{\rho_{A}}\overline{H}_{\beta}(B|E)_{\tau}.
\end{equation}
where the first optimization is with respect to bipartite pure states with
system $R$ isomorphic to system $A$, $\omega_{RB}\equiv\mathcal{N}%
_{A\rightarrow B}(\psi_{RA})$, $\tau_{BE}\equiv\mathcal{U}_{A\rightarrow
BE}^{\mathcal{N}}(\rho_{A})$, and $\beta=1/\alpha$.\ For $\alpha\in
\lbrack1/2,1)\cup(1,\infty)$,%
\begin{equation}
\left\vert H_{\alpha}(\mathcal{N})\right\vert \leq\log_{2}\left\vert
B\right\vert .
\end{equation}
For $\alpha\in(1,\infty)$, we have that%
\begin{equation}
H_{\alpha}(\mathcal{N})=\frac{\alpha}{1-\alpha}\log_{2}\left\Vert
\mathcal{N}_{A\rightarrow B}\right\Vert _{\operatorname{CB},1\rightarrow
\alpha}.
\end{equation}

\end{proposition}

\begin{proof}
To establish the equality%
\begin{equation}
H_{\alpha}(\mathcal{N})=\inf_{\psi_{RA}}H_{\alpha}(B|R)_{\omega|\omega},
\label{eq:simple-identity-H-alpha}%
\end{equation}
we follow the same reasoning as in
\eqref{eq:ent-to-CB-ent-1}--\eqref{eq:ent-to-CB-ent-last}, but making the
substitutions $H\rightarrow H_{\alpha}$ and $D\rightarrow D_{\alpha}$. To
establish the equality%
\begin{equation}
\inf_{\psi_{RA}}H_{\alpha}(B|R)_{\omega|\omega}=-\sup_{\rho_{A}}\overline
{H}_{\beta}(B|E)_{\tau},
\end{equation}
we employ the identity \cite[Theorem~2]{TBH14}%
\begin{equation}
H_{\alpha}(B|R)_{\omega|\omega}=-\overline{H}_{\beta}(B|E)_{\tau}.
\end{equation}

To establish the dimension bounds, consider from data processing that%
\begin{equation}
H_{\alpha}(B|R)_{\omega|\omega}\leq H_{\alpha}(B)_{\omega}\leq\log
_{2}\left\vert B\right\vert ,
\end{equation}
where the second inequality follows from a dimension bound for the R\'{e}nyi
entropy. To establish the other dimension bound, let us employ the identity
\cite[Theorem~2]{TBH14} again%
\begin{align}
H_{\alpha}(B|R)_{\omega|\omega}  &  =-\overline{H}_{\beta}(B|E)_{\tau}\\
&  \geq-\inf_{\sigma_{E}}D_{\beta}(\tau_{BE}\Vert I_{B}\otimes\sigma_{E})\\
&  \geq-H_{\beta}(B)_{\tau}\\
&  \geq-\log_{2}\left\vert B\right\vert .
\end{align}
The first inequality is stated in \cite[Corollary~4]{TBH14}, and the second
follows from data processing of the Petz--R\'{e}nyi relative entropy under
measurements, which holds for $\beta\in(0,1)\cup(1,\infty)$, as shown in
\cite[Section~2.2]{Hay06} (note that a measurement in the eigenbasis of
$\tau_{B}$ combined with the partial trace over system $E$ is a particular
kind of measurement).

To establish the connection to the completely bounded norm for $\alpha>1$, we
invoke \cite[Lemma~8]{CMW16} to find that%
\begin{align}
&  H_{\alpha}(\mathcal{N})\nonumber\\
&  =\log_{2}\left\vert B\right\vert -D_{\alpha}(\mathcal{N}\Vert\mathcal{R})\\
&  =\log_{2}\left\vert B\right\vert -\frac{\alpha}{\alpha-1}\log\left\Vert
\Omega_{\pi_{B}^{(1-\alpha)/\alpha}}\circ\mathcal{N}_{A\rightarrow
B}\right\Vert _{\text{CB},1\rightarrow\alpha}\\
&  =\frac{\alpha}{1-\alpha}\log\left\Vert \mathcal{N}_{A\rightarrow
B}\right\Vert _{\text{CB},1\rightarrow\alpha},
\end{align}
where
\begin{align}
\Omega_{\pi_{B}^{(1-\alpha)/\alpha}}(X_{B})  &  \equiv\pi_{B}^{\left(
1-\alpha\right)  /2\alpha}X_{B}\pi_{B}^{\left(  1-\alpha\right)  /2\alpha}\\
&  =\left\vert B\right\vert ^{\left(  \alpha-1\right)  /\alpha}X_{B},
\end{align}
concluding the proof.
\end{proof}

\bigskip

Again, the dimension lower bound is saturated for the identity channel, while
the dimension upper bound is saturated for the completely depolarizing channel.

\section{Min-entropy of a quantum channel}

\label{sec:min-ent-ch}

The min-entropy of a quantum state $\rho_{A}$ of a system $A$ is defined as
\cite{Renner2005}%
\begin{align}
H_{\min}(A)_{\rho}  &  \equiv-\log_{2}\left\Vert \rho\right\Vert _{\infty}\\
&  =\lim_{\alpha\rightarrow\infty}H_{\alpha}(A)_{\rho}.
\end{align}
It has found extensive application in the context of quantum cryptography
\cite{Renner2005}. The max-relative entropy of a state $\rho$ with a positive
semi-definite operator $\sigma$ is defined as \cite{Datta09}%
\begin{align}
D_{\max}(\rho\Vert\sigma)  &  \equiv\inf\left\{  \lambda:\rho\leq2^{\lambda
}\sigma\right\} \\
&  =\log_{2}\left\Vert \sigma^{-1/2}\rho\sigma^{-1/2}\right\Vert _{\infty},
\end{align}
whenever $\operatorname{supp}(\rho)\subseteq\operatorname{supp}(\sigma)$, and
otherwise, it is set to $+\infty$. The max-relative entropy was recently given
an information-theoretic meaning as the distinguishability cost of two quantum
states \cite{WW19states}. It is known that \cite{MDSFT13}%
\begin{equation}
D_{\max}(\rho\Vert\sigma)=\lim_{\alpha\rightarrow\infty}D_{\alpha}(\rho
\Vert\sigma). \label{eq:max-rel-as-limit}%
\end{equation}

Observe that the min-entropy of a quantum state $\rho$ can be written as the
difference of the number of physical qubits for the system~$A$ and the
max-relative entropy of $\rho$ to the maximally mixed state~$\pi_{A}$:
\begin{equation}
H_{\min}(A)_{\rho}=\log_{2}\left\vert A\right\vert -D_{\max}(\rho_{A}\Vert
\pi_{A}).
\end{equation}
Thus following the spirit of previous developments, we define the min-entropy
of a channel as follows:

\begin{definition}
[Min-entropy of a quantum channel]\label{def:min-entropy-ch}We define the
\textit{min-entropy of a quantum channel} $\mathcal{N}_{A\rightarrow B}$
according to the recipe given in the introduction of our paper:%
\begin{equation}
H_{\min}(\mathcal{N})\equiv\log_{2}\left\vert B\right\vert -D_{\max
}(\mathcal{N}\Vert\mathcal{R}),
\end{equation}
where $D_{\max}(\mathcal{N}\Vert\mathcal{R})$ is the max-channel divergence
\cite{CMW16,LKDW18} and $\mathcal{R}_{A\rightarrow B}$ is the completely
randomizing channel from \eqref{eq:cm-random-ch}.
\end{definition}

The max-channel divergence is defined for two arbitrary channels
$\mathcal{N}_{A\rightarrow B}$ and $\mathcal{M}_{A\rightarrow B}$ as
\cite{CMW16,LKDW18}%
\begin{align}
&  D_{\max}(\mathcal{N}\Vert\mathcal{M})\nonumber\\
&  \equiv\sup_{\rho_{RA}}D_{\max}(\mathcal{N}_{A\rightarrow B}(\rho_{RA}%
)\Vert\mathcal{M}_{A\rightarrow B}(\rho_{RA}))\\
&  =D_{\max}(\mathcal{N}_{A\rightarrow B}(\Phi_{RA})\Vert\mathcal{M}%
_{A\rightarrow B}(\Phi_{RA})). \label{eq:max-rel-ch-max-ent}%
\end{align}
The latter equality, that an optimal state is the maximally entangled state
$\Phi_{RA}$, was proved in \cite[Lemma~12]{BHKW18} (see also \cite[Eq.~(45)]%
{GFWRSCW18}\ and \cite[Remark~13]{BHKW18} in this context). In fact, an
optimal state is any pure bipartite state with full Schmidt rank (reduced
state has full support).

Due to the limit in \eqref{eq:max-rel-as-limit} and the equality in
\eqref{eq:max-rel-ch-max-ent}, it follows that%
\begin{equation}
D_{\max}(\mathcal{N}\Vert\mathcal{M})=\lim_{\alpha\rightarrow\infty}D_{\alpha
}(\mathcal{N}\Vert\mathcal{M}).
\end{equation}
As such, we can immediately conclude that the min-entropy of a channel
$H_{\min}(\mathcal{N})$ is equal to the following limit%
\begin{equation}
H_{\min}(\mathcal{N})=\lim_{\alpha\rightarrow\infty}H_{\alpha}(\mathcal{N}%
),\label{eq:min-ent-limit}%
\end{equation}
and that it satisfies non-decrease under a uniformity preserving superchannel,
additivity, and reduction to states (i.e., for a replacer channel, it reduces
to the min-entropy of the replacing state), which, as stated previously, imply
the three axioms from \cite{G18}.

\subsection{Alternate representation for the min-entropy of a channel in terms
of conditional min-entropies}

The conditional min-entropy of a bipartite quantum state $\rho_{AB}$ is
defined as \cite{Renner2005}%
\begin{equation}
H_{\min}(A|B)_{\rho}\equiv-\inf_{\sigma_{B}}D_{\max}(\rho_{AB}\Vert
I_{A}\otimes\sigma_{B}).
\end{equation}
We can also define the following related quantity:%
\begin{equation}
H_{\min}(A|B)_{\rho|\rho}\equiv-D_{\max}(\rho_{AB}\Vert I_{A}\otimes\rho_{B}),
\end{equation}
and clearly we have that%
\begin{equation}
H_{\min}(A|B)_{\rho}\geq H_{\min}(A|B)_{\rho|\rho}.
\end{equation}

The identities in \eqref{eq:cb-min} and \eqref{eq:von-neu-collapse}, as well
as the definition of conditional min-entropy, inspire the following quantity:%
\begin{equation}
H_{\min}^{\uparrow}(\mathcal{N})=\inf_{\psi_{RA}}H_{\min}(B|R)_{\omega}.
\end{equation}
In the above, $\omega_{RB}\equiv\mathcal{N}_{A\rightarrow B}(\psi_{RA})$ and
$\psi_{RA}$ is a pure state with system $R$ isomorphic to the channel input
system $A$.

This quantity might seem different from the min-entropy of a channel, but the
following proposition states that $H_{\min}^{\uparrow}(\mathcal{N})$ is
actually equal to the min-entropy of the channel $H_{\min}(\mathcal{N})$, thus
simplifying the notion of min-entropy of a quantum channel:

\begin{proposition}
\label{prop:alt-min-ent-ch}Let $\mathcal{N}_{A\rightarrow B}$ be a quantum
channel. Then%
\begin{align}
H_{\min}(\mathcal{N})  &  =\inf_{\psi_{RA}}H_{\min}(B|R)_{\omega|\omega}\\
&  =H_{\min}(B|R)_{\Phi^{\mathcal{N}}|\Phi^{\mathcal{N}}}\\
&  =H_{\min}^{\uparrow}(\mathcal{N}),
\end{align}
where $\omega_{RB}\equiv\mathcal{N}_{A\rightarrow B}(\psi_{RA})$ and
$\psi_{RA}$ is a pure state with system $R$ isomorphic to the channel input
system $A$.\ Also, the state $\Phi_{RB}^{\mathcal{N}}=\mathcal{N}%
_{A\rightarrow B}(\Phi_{RA})$ is the Choi state of the channel.
\end{proposition}

\begin{proof}
The first equality follows from the same steps in the proof of
Proposition~\ref{prop:identities-for-renyi} (see the reasoning around
\eqref{eq:simple-identity-H-alpha}). The second equality, i.e.,%
\begin{equation}
H_{\min}(\mathcal{N})=H_{\min}(B|R)_{\Phi^{\mathcal{N}}|\Phi^{\mathcal{N}}},
\label{eq:reduction-channel-min-ent}%
\end{equation}
follows by the observation in \eqref{eq:max-rel-ch-max-ent}.

The proof of the equality $H_{\min}(\mathcal{N})=H_{\min}^{\uparrow
}(\mathcal{N})$\ follows from semi-definite programming duality, similar to
what was done previously for conditional min-entropy in \cite{KRS09}. Consider
that%
\begin{align}
&  H_{\min}^{\uparrow}(\mathcal{N})\nonumber\\
&  =\inf_{\psi_{RA}}H_{\min}(B|R)_{\mathcal{N}_{A\rightarrow B}(\psi_{RA})}\\
&  =\inf_{\psi_{RA}}\left[  -\inf_{\sigma_{R}}D_{\max}(\mathcal{N}%
_{A\rightarrow B}(\psi_{RA})\Vert\sigma_{R}\otimes I_{B})\right] \\
&  =\inf_{\psi_{RA}}\left[  -\inf_{\sigma_{R}}\left\{  \log_{2}%
\operatorname{Tr}\{\sigma_{R}\}:\mathcal{N}_{A\rightarrow B}(\psi_{RA}%
)\leq\sigma_{R}\otimes I_{B}\right\}  \right] \\
&  =-\log_{2}\sup_{\psi_{RA}}\inf_{\sigma_{R}}\left\{  \operatorname{Tr}%
\{\sigma_{R}\}:\mathcal{N}_{A\rightarrow B}(\psi_{RA})\leq\sigma_{R}\otimes
I_{B}\right\}  .
\end{align}
Considering the innermost part of the last line above as the following
semi-definite program%
\begin{equation}
\inf_{\sigma_{R}}\left\{  \operatorname{Tr}\{\sigma_{R}\}:\mathcal{N}%
_{A\rightarrow B}(\psi_{RA})\leq\sigma_{R}\otimes I_{B}\right\}  ,
\end{equation}
its dual is given by%
\begin{equation}
\sup_{X_{RB}}\left\{  \operatorname{Tr}\{X_{RB}\mathcal{N}_{A\rightarrow
B}(\psi_{RA})\}:X_{R}\leq I_{R},\ X_{RB}\geq0\right\}  .
\end{equation}
Now let us write the pure state $\psi_{RA}$ as $\psi_{RA}=Y_{R}\Gamma
_{RA}Y_{R}^{\dag}$, where $Y_{R}$ satisfies $\operatorname{Tr}\{Y_{R}^{\dag
}Y_{R}\}=1$, so that $Y_{R}^{\dag}Y_{R}$ is a density operator $\rho_{R}$. Due
to the fact that the set of pure states $\psi_{RA}$ with full-rank reduced
state $\psi_{R}$ is dense in the set of all pure states, it suffices to
optimize over these. This means that we can rewrite the last line of the first
block (without the negative logarithm) as\begin{widetext}
\begin{equation}
\sup_{Y_{R},X_{RB}}\left\{  \operatorname{Tr}\{Y_{R}^\dag X_{RB}Y
_{R}\mathcal{N}_{A\rightarrow B}(\Gamma_{RA})\}:X_{R}\leq I_{R}%
,\ X_{RB}\geq0,\ \operatorname{Tr}\{Y_{R}^\dag Y_{R}\}=1,\ |Y_{R}|>0\right\}  .
\end{equation}
We now define $X_{RB}^{\prime}\equiv Y_{R}^\dag X_{RB}Y_{R}$, so that
\begin{align}
X_{RB} \geq 0 \quad & \Longleftrightarrow \quad X'_{RB} \geq 0 \\
X_{R} \leq I_R \quad & \Longleftrightarrow \quad X_{R}' \leq Y_R^\dag Y_R = \rho_R.
\end{align}
Then we find that the
above is equal to%
\begin{align}
&  \sup_{\rho_{R},X_{RB}^{\prime}}\left\{  \operatorname{Tr}\{X_{RB}^{\prime
}\mathcal{N}_{A\rightarrow B}(\Gamma_{RA})\}:X_{R}^{\prime}\leq\rho
_{R},\ \operatorname{Tr}\{\rho_{R}\}=1,\ X_{RB}^{\prime}\geq0,\ \rho_{R}%
\geq0\right\} \nonumber\\
&  =\sup_{X_{RB}^{\prime}}\left\{  \operatorname{Tr}\{X_{RB}^{\prime
}\mathcal{N}_{A\rightarrow B}(\Gamma_{RA})\}:\operatorname{Tr}\{X_{R}^{\prime
}\}=1,\ X_{RB}^{\prime}\geq0\right\} \\
&  =\sup_{X_{RB}^{\prime}}\left\{  \operatorname{Tr}\{X_{RB}^{\prime
}\mathcal{N}_{A\rightarrow B}(\Gamma_{RA})\}:\operatorname{Tr}\{X_{RB}%
^{\prime}\}=1,\ X_{RB}^{\prime}\geq0\right\} \\
&  =\left\Vert \mathcal{N}_{A\rightarrow B}(\Gamma_{RA})\right\Vert _{\infty}%
\end{align}
\end{widetext}
The first equality above follows because we are maximizing over both $\rho
_{R}$ and $X^{\prime}_{RB}$, and the objective function only increases by
taking $X_{R}^{\prime}=\rho_{R}$ and with a maximal value one for the trace.
So we conclude that%
\begin{align}
&  \inf_{\psi_{RA}}H_{\min}(B|R)_{\mathcal{N}_{A\rightarrow B}(\psi_{RA}%
)}\nonumber\\
&  =-\log_{2}\left\Vert \mathcal{N}_{A\rightarrow B}(\Gamma_{RA})\right\Vert
_{\infty}\\
&  =-\log_{2}\inf\left\{  \lambda:\mathcal{N}_{A\rightarrow B}(\Gamma
_{RA})\leq\lambda I_{RB}\right\} \\
&  =-\log_{2}\inf\left\{  \lambda:\mathcal{N}_{A\rightarrow B}(\Phi_{RA}%
)\leq\lambda\pi_{R}\otimes I_{B}\right\} \\
&  =H_{\min}(B|R)_{\Phi^{\mathcal{N}}|\Phi^{\mathcal{N}}}\\
&  =H_{\min}(\mathcal{N}),
\end{align}
where $\Phi^{\mathcal{N}}=\mathcal{N}_{A\rightarrow B}(\Phi_{RA})$ and the
last equality follows from \eqref{eq:reduction-channel-min-ent}.
\end{proof}

\subsection{Relation of min-entropy of a channel to its extended min-entropy}

The extended min-entropy of a channel is defined as~\cite{G18}%
\begin{equation}
H_{\min}^{\operatorname{ext}}(\mathcal{N})\equiv H_{\min}(B|R)_{\omega},
\end{equation}
where $\omega_{RA}=\mathcal{N}_{A\rightarrow B}(\Phi_{RA})$, with $\Phi_{RA}$
the maximally entangled state. It is not clear to us whether $H_{\min
}^{\operatorname{ext}}(\mathcal{N})$\ is generally equal to the min-entropy of
a channel $H_{\min}(\mathcal{N})$. However, due to
\eqref{eq:reduction-channel-min-ent}, we conclude that%
\begin{equation}
H_{\min}^{\operatorname{ext}}(\mathcal{N})\geq H_{\min}(\mathcal{N}).
\end{equation}

\section{Asymptotic Equipartition Property}

\label{sec:AEP}

The smoothed conditional min-entropy of a bipartite state $\rho_{AB}$\ is
defined for $\varepsilon\in(0,1)$ as (see, e.g., \cite{T15})%
\begin{equation}
H_{\min}^{\varepsilon}(A|B)_{\rho}\equiv\sup_{P(\rho_{AB},\widetilde{\rho
}_{AB})\leq\varepsilon}H_{\min}(A|B)_{\widetilde{\rho}},
\end{equation}
where the optimization is with respect to all subnormalized states
$\widetilde{\rho}_{AB}$ (satisfying $\widetilde{\rho}_{AB}\geq0$,
$\operatorname{Tr}\{\widetilde{\rho}_{AB}\} \leq1$, and $\widetilde{\rho}%
_{AB}\neq0$) and the sine distance (also called purified distance) of quantum
states $\rho$ and $\sigma$ \cite{R02,R03,GLN04,R06}\ is defined in terms of
the fidelity \cite{U76}\ as%
\begin{align}
P(\rho,\sigma)  &  \equiv\sqrt{1-F(\rho,\sigma)},\\
F(\rho,\sigma)  &  \equiv\left\Vert \sqrt{\rho}\sqrt{\sigma}\right\Vert
_{1}^{2}.
\end{align}
The definition of fidelity is generalized to subnormalized states $\omega$ and
$\tau$ as follows \cite{TCR10}:
\begin{equation}
F(\omega,\tau) \equiv F(\omega\oplus[1-\operatorname{Tr}\{\omega\}],\tau
\oplus[1-\operatorname{Tr}\{\tau\}]),
\end{equation}
where the right-hand side is the usual fidelity of states (that is, we just
add an extra dimension to $\omega$ and $\tau$ and complete them to states).
The smoothed conditional min-entropy satisfies the following asymptotic
equipartition property \cite{TCR09} (see also \cite{T15}), which is one way
that it connects with the conditional entropy of $\rho_{AB}$:%
\begin{equation}
\lim_{n\rightarrow\infty}\frac{1}{n}H_{\min}^{\varepsilon}(A^{n}|B^{n}%
)_{\rho^{\otimes n}}=H(A|B)_{\rho}. \label{eq:AEP-states}%
\end{equation}

The purified channel divergence of two channels $\mathcal{N}_{A\rightarrow B}$
and $\mathcal{M}_{A\rightarrow B}$ is defined as \cite{LKDW18}%
\begin{equation}
P(\mathcal{N},\mathcal{M})\equiv\sup_{\rho_{RA}}P(\mathcal{N}_{A\rightarrow
B}(\rho_{RA}),\mathcal{M}_{A\rightarrow B}(\rho_{RA})),
\label{eq:purified-ch-divergence}%
\end{equation}
Again, due to state purification, the data-processing inequality for
$P(\rho,\sigma)$,\ and the Schmidt decomposition theorem, it suffices to
optimize over states $\rho_{RA}$ that are pure and such that system $R$ is
isomorphic to system $A$. We then use this notion for smoothing the
min-entropy of a channel:

\begin{definition}
[Smoothed min-entropy of a channel]\label{def:smooth-min-ent-channel} The
smoothed min-entropy of a channel is defined for $\varepsilon\in(0,1)$ as%
\begin{equation}
H_{\min}^{\varepsilon}(\mathcal{N})\equiv\sup_{P(\mathcal{N},\widetilde
{\mathcal{N}})\leq\varepsilon}H_{\min}(\widetilde{\mathcal{N}}),
\end{equation}
where $P(\mathcal{N},\widetilde{\mathcal{N}})$ is the purified channel
divergence \cite{LKDW18}.
\end{definition}

In the following theorem, we prove that the smoothed min-entropy of a channel
satisfies an asymptotic equipartition theorem that generalizes \eqref{eq:AEP-states}.

\begin{theorem}
[Asymptotic equipartition property]\label{thm:AEP-channel-entropy}For all
$\varepsilon\in(0,1)$, the following inequality holds%
\begin{equation}
\lim_{n\rightarrow\infty}\frac{1}{n}H_{\min}^{\varepsilon}(\mathcal{N}%
^{\otimes n})\geq H(\mathcal{N}). \label{eq:lower-AEP}%
\end{equation}
We also have that%
\begin{equation}
\lim_{\varepsilon\rightarrow0}\lim_{n\rightarrow\infty}\frac{1}{n}H_{\min
}^{\varepsilon}(\mathcal{N}^{\otimes n})\leq H(\mathcal{N}).
\label{eq:upper-AEP}%
\end{equation}

\end{theorem}

\begin{proof}
We first prove the inequality in \eqref{eq:lower-AEP}. Let $\omega_{R^{n}%
A^{n}}$ denote the de Finetti state \cite{CKR09}, defined as%
\begin{equation}
\omega_{R^{n}A^{n}}\equiv\int d(\sigma_{RA})\ \sigma_{RA}^{\otimes n},
\end{equation}
where $\sigma_{RA}$ is a pure state with system $R$ isomorphic to the channel
input system $A$, and $d(\sigma_{RA})$ denotes the Haar measure on pure
states. This state is the maximally mixed state of the symmetric subspace of
the systems $(RA)^{n}$, and it is permutation invariant \cite{H13}. That is,
for a unitary channel $\mathcal{W}_{R^{n}}^{\pi}\otimes\mathcal{W}_{A^{n}%
}^{\pi}$ corresponding to a permutation $\pi$, we have that $\omega
_{R^{n}A^{n}}=(\mathcal{W}_{R^{n}}^{\pi}\otimes\mathcal{W}_{A^{n}}^{\pi
})(\omega_{R^{n}A^{n}})$ for all $\pi\in S_{n}$, with $S_{n}$ denoting the
symmetric group. Let $\omega_{R^{\prime}R^{n}A^{n}}$ denote the purification
of the de Finetti state, with the purifying system $R^{\prime}$ satisfying the
inequality $\left\vert R^{\prime}\right\vert \leq\left(  n+1\right)
^{\left\vert A\right\vert ^{2}-1}$ \cite{CKR09}. The reduced state
$\omega_{A^{n}}$ is permutation invariant and has full rank. The latter
follows because the set of pure states $\psi_{RA}$ with a full-rank reduced
density operator $\psi_{R}$ is dense in the set of all pure states, and tensor
products of full-rank states are full rank. Let $\omega_{R^{\prime}R^{n}B^{n}%
}^{\widetilde{\mathcal{N}}^{n}}$ denote the state resulting from the action of
the quantum channel $\widetilde{\mathcal{N}}_{A^{n}\rightarrow B^{n}}^{n}$ on
the input state $\omega_{R^{\prime}R^{n}A^{n}}$, and let $\omega_{R^{\prime
}R^{n}B^{n}}^{\mathcal{N}^{\otimes n}}$ denote the state resulting from the
action of the quantum channel $\mathcal{N}_{A\rightarrow B}^{\otimes n}$ on
the input state $\omega_{R^{\prime}R^{n}A^{n}}$. Let CPTP$(A^{n}\rightarrow
B^{n})$ denote the set of all quantum channels from input system $A^{n}$ to
output system $B^{n}$. Let Perm$(A^{n}\rightarrow B^{n})$ denote the set of
all permutation covariant quantum channels from input system $A^{n}$ to output
system $B^{n}$. Define $\psi_{RB^{n}}^{\widetilde{\mathcal{N}}^{n}}$ to be the
state resulting from the action of the channel $\widetilde{\mathcal{N}}%
_{A^{n}\rightarrow B^{n}}^{n}$ on the input state $\psi_{RA^{n}}$. Then
consider that%
\begin{align}
&  H_{\min}^{\varepsilon}(\mathcal{N}^{\otimes n})\nonumber\\
&  =\sup_{\substack{\widetilde{\mathcal{N}}^{n}\in\text{CPTP}(A^{n}\rightarrow
B^{n}):\\P(\mathcal{N}^{\otimes n},\widetilde{\mathcal{N}}^{n})\leq
\varepsilon}}\inf_{\psi_{RA^{n}}}H_{\min}(B^{n}|R)_{\psi^{\widetilde
{\mathcal{N}}^{n}}|\psi^{\widetilde{\mathcal{N}}^{n}}}\\
&  \geq\sup_{\substack{\widetilde{\mathcal{N}}^{n}\in\text{Perm}%
(A^{n}\rightarrow B^{n}):\\P(\mathcal{N}^{\otimes n},\widetilde{\mathcal{N}%
}^{n})\leq\varepsilon}}\inf_{\psi_{RA^{n}}}H_{\min}(B^{n}|R)_{\psi
^{\widetilde{\mathcal{N}}^{n}}|\psi^{\widetilde{\mathcal{N}}^{n}}}\\
&  =\sup_{\substack{\widetilde{\mathcal{N}}^{n}\in\text{Perm}(A^{n}\rightarrow
B^{n}):\\P(\mathcal{N}^{\otimes n},\widetilde{\mathcal{N}}^{n})\leq
\varepsilon}}H_{\min}(B^{n}|R^{n}R^{\prime})_{\omega^{\widetilde{\mathcal{N}%
}^{n}}|\omega^{\widetilde{\mathcal{N}}^{n}}}\\
&  \geq\sup_{\substack{\widetilde{\mathcal{N}}^{n}\in\text{Perm}%
(A^{n}\rightarrow B^{n}):\\P(\omega^{\mathcal{N}^{\otimes n}},\omega
^{\widetilde{\mathcal{N}}^{n}})\leq\varepsilon^{\prime}}}H_{\min}(B^{n}%
|R^{n}R^{\prime})_{\omega^{\widetilde{\mathcal{N}}^{n}}|\omega^{\widetilde
{\mathcal{N}}^{n}}}\label{eq:last-line-1st-blk-AEP}%
\end{align}
The first equality follows from Definition~\ref{def:smooth-min-ent-channel}.
The first inequality follows by restricting the maximization to
permutation-covariant channels. The second equality follows because the
reduced state $\omega_{A^{n}}$ has full rank and by applying the remark after
\eqref{eq:max-rel-ch-max-ent}, to conclude that%
\begin{align}
& \inf_{\psi_{RA^{n}}}H_{\min}(B^{n}|R)_{\psi^{\widetilde{\mathcal{N}}^{n}%
}|\psi^{\widetilde{\mathcal{N}}^{n}}}\nonumber\\
& =-\sup_{\psi_{RA^{n}}}D_{\max}(\widetilde{\mathcal{N}}_{A^{n}\rightarrow
B^{n}}^{n}(\psi_{RA^{n}})\Vert\psi_{R}\otimes I_{B^{n}})\\
& =-\sup_{\psi_{RA^{n}}}D_{\max}(\widetilde{\mathcal{N}}_{A^{n}\rightarrow
B^{n}}^{n}(\psi_{RA^{n}})\Vert\psi_{R}\otimes\pi_{B^{n}})\nonumber\\
& \qquad+n\log_{2}\left\vert B\right\vert \\
& =-\sup_{\psi_{RA^{n}}}D_{\max}(\widetilde{\mathcal{N}}_{A^{n}\rightarrow
B^{n}}^{n}(\psi_{RA^{n}})\Vert\mathcal{R}_{A\rightarrow B}^{\otimes n}%
(\psi_{RA^{n}}))\nonumber\\
& \qquad+n\log_{2}\left\vert B\right\vert \\
& =-D_{\max}(\widetilde{\mathcal{N}}_{A^{n}\rightarrow B^{n}}^{n}%
(\omega_{R^{\prime}R^{n}A^{n}})\Vert\mathcal{R}_{A\rightarrow B}^{\otimes
n}(\omega_{R^{\prime}R^{n}A^{n}}))\nonumber\\
& \qquad+n\log_{2}\left\vert B\right\vert \\
& =H_{\min}(B^{n}|R^{n}R^{\prime})_{\omega^{\widetilde{\mathcal{N}}^{n}%
}|\omega^{\widetilde{\mathcal{N}}^{n}}}.
\end{align}
The second inequality follows by applying the post-selection technique
\cite[Theorem~1]{CKR09}\ with%
\begin{equation}
\varepsilon^{\prime}\equiv\varepsilon\left(  n+1\right)  ^{-2\left(
\left\vert A\right\vert ^{2}-1\right)  }.\label{eq:eps-post-select}%
\end{equation}
(See also Proposition~D.5 of \cite{BCR09}.)\ Note that the factor of two in
the exponent of \eqref{eq:eps-post-select} is necessary because we are
employing the sine distance as the channel distance measure. To be clear, the
statement we are invoking is that if%
\begin{equation}
\mathcal{N}^{n},\mathcal{M}^{n}\in\text{Perm}(A^{n}\rightarrow B^{n})
\end{equation}
satisfy%
\begin{equation}
P(\mathcal{N}_{A^{n}\rightarrow B^{n}}^{n}(\omega_{R^{\prime}R^{n}A^{n}%
}),\mathcal{M}_{A^{n}\rightarrow B^{n}}^{n}(\omega_{R^{\prime}R^{n}A^{n}%
}))\leq\varepsilon^{\prime},
\end{equation}
then%
\begin{equation}
P(\mathcal{N}_{A^{n}\rightarrow B^{n}}^{n},\mathcal{M}_{A^{n}\rightarrow
B^{n}}^{n})\leq\varepsilon.
\end{equation}
Continuing, we have that%
\begin{align}
&  \text{Eq.}~\eqref{eq:last-line-1st-blk-AEP}\nonumber\\
&  =\sup_{\substack{\widetilde{\mathcal{N}}^{n}\in\text{CPTP}(A^{n}\rightarrow
B^{n}):\\P(\omega^{\mathcal{N}^{\otimes n}},\omega^{\widetilde{\mathcal{N}%
}^{n}})\leq\varepsilon^{\prime}}}H_{\min}(B^{n}|R^{n}R^{\prime})_{\omega
^{\widetilde{\mathcal{N}}^{n}}|\omega^{\widetilde{\mathcal{N}}^{n}}}\\
&  =\sup_{\substack{\sigma_{R^{\prime}R^{n}B^{n}}:\\P(\mathcal{N}^{\otimes
n}(\omega_{R^{\prime}R^{n}A^{n}}),\sigma_{R^{\prime}R^{n}B^{n}})\leq
\varepsilon^{\prime},\\\sigma_{R^{\prime}R^{n}}=\omega_{R^{\prime}R^{n}}%
}}H_{\min}(B^{n}|R^{n}R^{\prime})_{\sigma|\sigma}\\
&  \geq\sup_{\substack{\sigma_{R^{\prime}R^{n}B^{n}}:\\P(\mathcal{N}^{\otimes
n}(\omega_{R^{\prime}R^{n}A^{n}}),\sigma_{R^{\prime}R^{n}B^{n}})\leq
2\varepsilon^{\prime}/3}}H_{\min}(B^{n}|R^{n}R^{\prime})_{\sigma|\sigma
}\nonumber\\
&  \qquad-\log_{2}\!\left(  \frac{8+\left[  \varepsilon^{\prime}/3\right]
^{2}}{\left[  \varepsilon^{\prime}/3\right]  ^{2}}\right)  .
\end{align}
The first equality follows from reasoning similar to that given for Lemma~11
in Appendix~B of \cite{FWTB18}, i.e., that a permutation-covariant channel is
optimal among all channels, due to the fact that the original channel
$\mathcal{N}^{\otimes n}$ is permutation covariant. In our case, it follows by
employing the fact that the channel min-entropy does not decrease under the
action of a uniformity preserving superchannel (see the discussion after
\eqref{eq:min-ent-limit}), and the superchannel that randomly performs a
permutation at the channel input and the inverse permutation at the channel
output is one such superchannel. The second equality is a consequence of the
fact that the following two sets are equal:
\begin{multline}
\Big\{\widetilde{\mathcal{N}}_{A^{n}\rightarrow B^{n}}^{n}(\phi_{RA^{n}}):\\
P(\widetilde{\mathcal{N}}_{A^{n}\rightarrow B^{n}}^{n}(\phi_{RA^{n}%
}),\mathcal{N}_{A\rightarrow B}^{\otimes n}(\phi_{RA^{n}}))\leq\varepsilon,\\
\widetilde{\mathcal{N}}_{A^{n}\rightarrow B^{n}}^{n}\in\text{CPTP}\Big\}\\
=\Big\{\widetilde{\omega}_{RB^{n}}\in\mathcal{D}(\mathcal{H}_{RB^{n}%
}):P(\widetilde{\omega}_{RB^{n}},\mathcal{N}_{A\rightarrow B}^{\otimes n}%
(\phi_{RA^{n}}))\leq\varepsilon,\\
\widetilde{\omega}_{R}=\phi_{R}\Big\},
\end{multline}
which follows from applying Lemma~10 in Appendix~B of \cite{FWTB18}. The
inequality follows from Theorem~3 of \cite{ABJT19} (while noting that the
state $\hat{\rho}_{AB}$ defined therein satisfies $\hat{\rho}_{B}=\rho_{B}$,
so that the proof Theorem~3 of \cite{ABJT19} applies to our situation).

Continuing, and by applying \cite[Eq.~(L10)]{WW19states} and definitions, we
find that%
\begin{align}
&  \sup_{\substack{\sigma_{R^{\prime}R^{n}B^{n}}:\\P(\mathcal{N}^{\otimes
n}(\omega_{R^{\prime}R^{n}A^{n}}),\sigma_{R^{\prime}R^{n}B^{n}})\leq
2\varepsilon^{\prime}/3}}H_{\min}(B^{n}|R^{n}R^{\prime})_{\sigma|\sigma
}\nonumber\\
&  \geq H_{\alpha}(B^{n}|R^{n}R^{\prime})_{\omega^{\mathcal{N}^{\otimes n}%
}|\omega^{\mathcal{N}^{\otimes n}}} + f(2\varepsilon^{\prime}/3,\alpha)\\
&  \geq H_{\alpha}(\mathcal{N}^{\otimes n})++ f(2\varepsilon^{\prime}%
/3,\alpha)\\
&  =nH_{\alpha}(\mathcal{N})+ f(2\varepsilon^{\prime}/3,\alpha),
\end{align}
where
\begin{equation}
f(\delta,\alpha) \equiv\log_{2}(1-\delta^{2}) + \frac{\log_{2}(\delta^{2}%
)}{\alpha- 1}.
\end{equation}
The second inequality follows from the definition of the R\'enyi entropy of a
channel (Definition~\ref{def:renyi-ent-ch}), and the equality follows from the
additivity of the R\'enyi entropy of a channel
(Proposition~\ref{prop:renyi-additive}). Putting everything above together, we
conclude the following bound:%
\begin{multline}
\frac{1}{n}H_{\min}^{\varepsilon}(\mathcal{N})\geq H_{\alpha}(\mathcal{N}%
)-\frac{1}{n}\log_{2}\!\left(  \frac{8+\left[  \varepsilon^{\prime}/3\right]
^{2}}{\left[  \varepsilon^{\prime}/3\right]  ^{2}}\right) \\
+ \frac{f(2\varepsilon^{\prime}/3,\alpha)}{n}.
\end{multline}
Taking the limit as $n\rightarrow\infty$, we conclude that the following
inequality holds for all $\alpha>1$:%
\begin{equation}
\lim_{n\rightarrow\infty}\frac{1}{n}H_{\min}^{\varepsilon}(\mathcal{N})\geq
H_{\alpha}(\mathcal{N}).
\end{equation}
Since this inequality holds for all $\alpha>1$, we can take the limit as
$\alpha\rightarrow1$ to conclude that%
\begin{equation}
\lim_{n\rightarrow\infty}\frac{1}{n}H_{\min}^{\varepsilon}(\mathcal{N})\geq
H(\mathcal{N}).
\end{equation}
This concludes the proof of the inequality in \eqref{eq:lower-AEP}.

To arrive at the second inequality in \eqref{eq:upper-AEP}, let $\widetilde
{\mathcal{N}}^{n}$ be a channel such that%
\begin{equation}
P(\mathcal{N}^{\otimes n},\widetilde{\mathcal{N}}^{n})\leq\varepsilon.
\end{equation}
Now let $\phi_{RA^{n}}$ be an arbitrary state. We then have from the
definition in \eqref{eq:purified-ch-divergence} that%
\begin{equation}
P(\mathcal{N}_{A\rightarrow B}^{\otimes n}(\phi_{RA^{n}}),\widetilde
{\mathcal{N}}_{A^{n}\rightarrow B^{n}}^{n}(\phi_{RA^{n}}))\leq\varepsilon.
\end{equation}
Defining the states%
\begin{align}
\widetilde{\omega}_{RB^{n}}  &  \equiv\widetilde{\mathcal{N}}_{A^{n}%
\rightarrow B^{n}}^{n}(\phi_{RA^{n}}),\\
\omega_{RB^{n}}  &  \equiv\mathcal{N}_{A\rightarrow B}^{\otimes n}%
(\phi_{RA^{n}}),
\end{align}
we find that%
\begin{align}
H_{\min}(\widetilde{\mathcal{N}}^{n})  &  \leq H_{\min}(B^{n}|R)_{\widetilde
{\omega}|\widetilde{\omega}}\\
&  \leq H(B^{n}|R)_{\widetilde{\omega}}\\
&  \leq H(B^{n}|R)_{\omega}+\varepsilon2n\log_{2}\left\vert B\right\vert
+g_{2}(\varepsilon),
\end{align}
where%
\begin{equation}
g_{2}(\varepsilon)\equiv\left(  \varepsilon+1\right)  \log_{2}(\varepsilon
+1)-\varepsilon\log_{2}\varepsilon. \label{eq:bosonic-entropy}%
\end{equation}
The second inequality follows from monotonicity of the conditional R\'{e}nyi
entropy with respect to$~\alpha$, and the last from the uniform continuity
bound in \cite[Lemma~2]{Winter15}. The above bound holds for any choice of
$\phi_{RA^{n}}$, and so we conclude that%
\begin{align}
H_{\min}(\widetilde{\mathcal{N}}^{n})  &  \leq H(\mathcal{N}^{\otimes
n})+\varepsilon2n\log_{2}\left\vert B\right\vert +g_{2}(\varepsilon)\\
&  =nH(\mathcal{N})+\varepsilon2n\log_{2}\left\vert B\right\vert
+g_{2}(\varepsilon),
\end{align}
where the equality follows from the additivity of the entropy of a channel
(Proposition~\ref{prop:entropy-additive}). Now, the inequality has been shown
for all $\widetilde{\mathcal{N}}^{n}$ satisfying $P(\mathcal{N}^{\otimes
n},\widetilde{\mathcal{N}}^{n})\leq\varepsilon$, and so we conclude, after
dividing by $n$, that%
\begin{equation}
\frac{1}{n}H_{\min}^{\varepsilon}(\mathcal{N}^{\otimes n})\leq H(\mathcal{N}%
)+2\varepsilon\log_{2}\left\vert B\right\vert +\frac{1}{n}g_{2}(\varepsilon).
\end{equation}
Taking the limit as $n\rightarrow\infty$, we get that%
\begin{equation}
\lim_{n\rightarrow\infty}\frac{1}{n}H_{\min}^{\varepsilon}(\mathcal{N}%
^{\otimes n})\leq H(\mathcal{N})+2\varepsilon\log_{2}\left\vert B\right\vert .
\end{equation}
Now taking the limit as $\varepsilon\rightarrow0$, we arrive at the second
inequality in \eqref{eq:upper-AEP}.
\end{proof}

\bigskip

In Appendix~\ref{app:max-mut-AEP}, we point out how an approach similar to
that in the above proof leads to an alternate proof of the upper bound in
\cite[Theorem~8]{FWTB18}, regarding an asymptotic equipartition property for
the smoothed max-mutual information of a quantum channel.

\section{Generalized channel entropies from generalized divergences}

\label{sec:other-ents}

In this section, we discuss other possibilities for defining generalized
entropies of a quantum channel. One main concern might be how unique or
distinguished our notion of entropy of a channel from
Definition~\ref{def:entropy-channel}\ is, being based on the channel relative
entropy of the channel of interest and the completely randomizing channel. As
a consequence of the fact that there are alternate ways of defining channel
relative entropies, there could be alternate notions of channel entropies.
However, we should recall that one of the main reasons we have chosen the
definition in Definition~\ref{def:entropy-channel}\ is that the channel
relative entropy appearing there has a particularly appealing operational
interpretation in the context of channel discrimination \cite{CMW16}. That is,
for what one might consider the most natural and general setting of quantum
channel discrimination, the optimal rate for distinguishing a channel from the
completely randomizing channel is given by the channel relative entropy in
\eqref{eq:channel-rel-ent} \cite{CMW16}. As we show in what follows, there are
further reasons to focus on our definition of the entropy of a channel from
Definition~\ref{def:entropy-channel}, as well as our definition of the
min-entropy of a channel from Definition~\ref{def:min-entropy-ch}.

To begin the discussion, let $\mathcal{S}(C)$ denote the set of quantum states
for an arbitrary quantum system $C$. Let us recall that a function
$\mathbf{D}:\mathcal{S}(C)\times\mathcal{S}(C)\rightarrow\mathbb{R}%
\cup\{+\infty\}$ is a generalized divergence \cite{PV10,SW12} if for arbitrary
Hilbert spaces $\mathcal{H}_{A}$ and $\mathcal{H}_{B}$, arbitrary states
$\rho_{A},\sigma_{A}\in\mathcal{S}(A)$, and an arbitrary channel
$\mathcal{N}_{A\rightarrow B}$, the following data processing inequality
holds
\begin{equation}
\mathbf{D}(\rho_{A}\Vert\sigma_{A})\geq\mathbf{D}(\mathcal{N}_{A\rightarrow
B}(\rho_{A})\Vert\mathcal{N}_{A\rightarrow B}(\sigma_{A})).
\label{eq:data-processing}%
\end{equation}
Examples of interest are in particular the quantum relative entropy, the
Petz-R\'{e}nyi divergences, the sandwiched R\'{e}nyi divergences, as
considered in this paper.

Based on generalized divergences, one can define at least two different
channel divergences as a measure for the distinguishability of two quantum
channels $\mathcal{N}_{A\rightarrow B}$ and $\mathcal{M}_{A\rightarrow B}$.
Here we consider a function of two quantum channels to be a channel divergence
if it is monotone under the action of a superchannel.

\begin{enumerate}
\item Generalized channel divergence \cite{LKDW18}:%
\begin{equation}
\mathbf{D}(\mathcal{N}\Vert\mathcal{M})\equiv\sup_{\rho_{RA}}\mathbf{D}%
(\mathcal{N}_{A\rightarrow B}(\rho_{RA})\Vert\mathcal{M}_{A\rightarrow B}%
(\rho_{RA})).
\end{equation}
In the above, the optimization can be restricted to pure states of systems $R$
and $A$ with $R$ isomorphic to system $A$. The monotonicity of the generalized
channel divergence under the action of a superchannel was proven in \cite{G18}.

\item Amortized channel divergence \cite{BHKW18}:%
\begin{multline}
\mathbf{D}^{\mathcal{A}}(\mathcal{N}\Vert\mathcal{M})\equiv\\
\sup_{\rho_{RA},\sigma_{RA}}\mathbf{D}(\mathcal{N}_{A\rightarrow B}(\rho
_{RA})\Vert\mathcal{M}_{A\rightarrow B}(\sigma_{RA}))-\mathbf{D}(\rho
_{RA}\Vert\sigma_{RA}).
\end{multline}
The monotonicity of the amortized channel divergence under the action of a
superchannel was proven in \cite{BHKW18}.
\end{enumerate}

We can consider other divergences as follows, but they are not known to be
monotone under the action of a general superchannel, and so we do not label
them as channel divergences:

\begin{enumerate}
\item Choi divergence:%
\begin{equation}
\mathbf{D}^{\Phi}(\mathcal{N}\Vert\mathcal{M})\equiv\mathbf{D}(\mathcal{N}%
_{A\rightarrow B}(\Phi_{RA})\Vert\mathcal{M}_{A\rightarrow B}(\Phi_{RA})).
\end{equation}
As we show in Appendix~\ref{app:choi-div-monotone}, the Choi divergence is
monotone under the action of a superchannel consisting of mixtures of a unital
pre-processing channel and an arbitrary post-processing channel.

\item Adversarial divergence:%
\begin{equation}
\mathbf{D}^{\text{adv}}(\mathcal{N}\Vert\mathcal{M})\equiv\sup_{\rho_{RA}}%
\inf_{\sigma_{RA}}\mathbf{D}(\mathcal{N}_{A\rightarrow B}(\rho_{RA}%
)\Vert\mathcal{M}_{A\rightarrow B}(\sigma_{RA})). \label{eq:adv-ch-div}%
\end{equation}
In the above, due to state purification, data processing, and the Schmidt
decomposition, the maximization can be restricted to pure states $\rho_{RA}$
of systems $R$ and $A$ with $R$ isomorphic to system $A$. The minimization
should be taken over mixed states $\sigma_{RA}$. For a proof of this fact, see
Appendix~\ref{app:adv-ch-div-limited}.

\item Adversarial Choi divergence:%
\begin{equation}
\mathbf{D}^{\text{adv},\Phi}(\mathcal{N}\Vert\mathcal{M})\equiv\inf
_{\sigma_{RA}}\mathbf{D}(\mathcal{N}_{A\rightarrow B}(\Phi_{RA})\Vert
\mathcal{M}_{A\rightarrow B}(\sigma_{RA})).
\end{equation}

\item \textquotedblleft No quantum memory\textquotedblright\ divergence:%
\begin{equation}
\sup_{\rho_{A}}\mathbf{D}(\mathcal{N}_{A\rightarrow B}(\rho_{A})\Vert
\mathcal{M}_{A\rightarrow B}(\rho_{A})).
\end{equation}

\end{enumerate}

There could certainly even be other divergences to consider. In our context,
two effective ways of singling out particular divergences as primary and
others as secondary are 1)\ whether the channel divergence has a compelling
operational interpretation for a channel discrimination task and 2)\ whether
the channel divergence leads to an entropy function that satisfies the axioms
from~\cite{G18}.

Based on the recipe given in the introduction, from a given divergence
$\mathbf{D}^{\prime}(\mathcal{N}\Vert\mathcal{M})$ (any of the choices above),
one could then define a generalized entropy function of a channel
$\mathcal{N}_{A\rightarrow B}$\ as%
\begin{equation}
\mathbf{H}(\mathcal{N})\equiv\log_{2}\left\vert B\right\vert -\mathbf{D}%
^{\prime}(\mathcal{N}\Vert\mathcal{R}),
\end{equation}
where $\mathcal{R}_{A\rightarrow B}$ is the completely randomizing channel from~\eqref{eq:cm-random-ch}.

Taking the above approach to pruning entropy functions, we can already rule
out the last one (\textquotedblleft no quantum memory\textquotedblright), as
done in \cite{G18}, because, after taking $\mathbf{D}$ to be the most
prominent case of quantum relative entropy, the resulting entropy function is
the minimum output entropy of a channel, which is known to be non-additive
\cite{Hastings2009}. While an entropy arising from the Choi divergence leads
to an entropy function satisfying the axioms desired for an entropy function,
the Choi divergence itself does not appear to have a compelling operational
interpretation in the sense of being a \textquotedblleft channel
measure\textquotedblright\ because it simply reduces a channel discrimination
problem to a state discrimination problem (i.e., it does not make use of the
most general approach one could take for discriminating arbitrary channels).
This point could be debated, and we do return to entropy functions derived
from Choi and adversarial Choi divergences in
Section~\ref{sec:ent-funcs-Choi-adv-Choi}\ below.

\subsection{Collapse of entropy functions derived from quantum relative
entropy}

From the list above, by focusing on the operational and axiomatic criteria
listed above, this leaves us with the generalized channel divergence and the
amortized channel divergence. Here we also consider the adversarial
divergence. Interestingly, after taking $\mathbf{D}$ to be the prominent case
of quantum relative entropy and the channel $\mathcal{M}$ to be the completely
randomizing channel, we find the following collapse of the divergences:%
\begin{equation}
D(\mathcal{N}\Vert\mathcal{R})=D^{\mathcal{A}}(\mathcal{N}\Vert\mathcal{R}%
)=D^{\text{adv}}(\mathcal{N}\Vert\mathcal{R}). \label{eq:div-collapse}%
\end{equation}
The first equality was shown in \cite{CMW16,BHKW18}, and we show the second
one now. From the definitions, we have that $\mathbf{D}^{\text{adv}%
}(\mathcal{N}\Vert\mathcal{M})\leq\mathbf{D}(\mathcal{N}\Vert\mathcal{M})$ for
any generalized divergence $\mathbf{D}$ and any channel $\mathcal{M}$. So we
show the opposite inequality for the special case of $\mathbf{D}=D$ and
$\mathcal{M}=\mathcal{R}$. Let $\rho_{RA}$ and $\sigma_{RA}$ be arbitrary
states. Then%
\begin{align}
&  D(\mathcal{N}_{A\rightarrow B}(\rho_{RA})\Vert\mathcal{R}_{A\rightarrow
B}(\sigma_{RA}))\nonumber\\
&  =D(\mathcal{N}_{A\rightarrow B}(\rho_{RA})\Vert\sigma_{R}\otimes\pi
_{B})\label{eq:adv-div-to-ch-div-1}\\
&  =-H(\mathcal{N}_{A\rightarrow B}(\rho_{RA}))\nonumber\\
&  \qquad-\operatorname{Tr}\{\mathcal{N}_{A\rightarrow B}(\rho_{RA})\log
_{2}(\sigma_{R}\otimes\pi_{B})\}\\
&  =-H(\mathcal{N}_{A\rightarrow B}(\rho_{RA}))-\operatorname{Tr}\{\rho
_{R}\log_{2}\sigma_{R}\}\nonumber\\
&  \qquad-\operatorname{Tr}\{\mathcal{N}_{A\rightarrow B}(\rho_{A})\log_{2}%
\pi_{B}\}\\
&  =-H(\mathcal{N}_{A\rightarrow B}(\rho_{RA}))-\operatorname{Tr}\{\rho
_{R}\log_{2}\rho_{R}\}\nonumber\\
&  \qquad+D(\rho_{R}\Vert\sigma_{R})-\operatorname{Tr}\{\mathcal{N}%
_{A\rightarrow B}(\rho_{A})\log_{2}\pi_{B}\}\\
&  =-H(\mathcal{N}_{A\rightarrow B}(\rho_{RA}))\nonumber\\
&  \qquad-\operatorname{Tr}\{\mathcal{N}_{A\rightarrow B}(\rho_{RA})\log
_{2}(\rho_{R}\otimes\pi_{B})\}+D(\rho_{R}\Vert\sigma_{R})\\
&  =D(\mathcal{N}_{A\rightarrow B}(\rho_{RA})\Vert\mathcal{R}_{A\rightarrow
B}(\rho_{RA}))+D(\rho_{R}\Vert\sigma_{R}). \label{eq:adv-div-to-ch-div-last}%
\end{align}
Now taking an infimum over all $\sigma_{RA}$ and invoking the non-negativity
of quantum relative entropy, we conclude that%
\begin{multline}
\inf_{\sigma_{RA}}D(\mathcal{N}_{A\rightarrow B}(\rho_{RA})\Vert
\mathcal{R}_{A\rightarrow B}(\sigma_{RA}))\\
=D(\mathcal{N}_{A\rightarrow B}(\rho_{RA})\Vert\mathcal{R}_{A\rightarrow
B}(\rho_{RA})).
\end{multline}
By taking a supremum over $\rho_{RA}$, we then conclude that $D^{\text{adv}%
}(\mathcal{N}\Vert\mathcal{R})=D(\mathcal{N}\Vert\mathcal{R})$.

Thus, the collapse in \eqref{eq:div-collapse}, as well as the operational
interpretation of $D(\mathcal{N}\Vert\mathcal{R})$ from \cite{CMW16} and the
fact that the resulting entropy function satisfies the axioms from \cite{G18},
indicate that our choice of the entropy of a quantum channel in
Definition~\ref{def:entropy-channel} is cogent.

\subsection{Collapse of entropy functions derived from max-relative entropy}

Interestingly, a similar and further collapse occurs when taking $\mathbf{D}$
to be the max-relative entropy:%
\begin{align}
D_{\max}(\mathcal{N}\Vert\mathcal{R})  &  =D_{\max}^{\Phi}(\mathcal{N}%
\Vert\mathcal{R})\\
&  =D_{\max}^{\mathcal{A}}(\mathcal{N}\Vert\mathcal{R})\\
&  =D_{\max}^{\text{adv}}(\mathcal{N}\Vert\mathcal{R}).
\end{align}
The first two equalities were shown in \cite[Proposition~10]{BHKW18} for
arbitrary channels $\mathcal{N}$ and $\mathcal{M}$. By employing a
semi-definite programming approach as in the proof of
Proposition~\ref{prop:alt-min-ent-ch}, we can conclude the last equality.
Thus, this collapse, as well as the facts that the max-relative entropy
$D_{\max}(\mathcal{N}\Vert\mathcal{M})$ is an upper bound on the rate at which
any two channels can be distinguished in an arbitrary context
\cite[Corollary~18]{BHKW18}\ and the resulting entropy function $H_{\min
}(\mathcal{N})$ satisfies the axioms from \cite{G18}, indicate that our choice
of the min-entropy of a quantum channel in Definition~\ref{def:min-entropy-ch}
is also cogent.

\subsection{Entropy functions derived from R\'enyi relative entropies}

In Section~\ref{sec:renyi-ent-ch}, we defined the R\'{e}nyi entropy of a
channel as in Definition~\ref{def:renyi-ent-ch}, in terms of the sandwiched
R\'{e}nyi relative entropy. The following collapse is known for the sandwiched
R\'{e}nyi relative entropy for $\alpha\in(1,\infty)$ \cite{CMW16,BHKW18}:%
\begin{equation}
D_{\alpha}(\mathcal{N}\Vert\mathcal{R})=D_{\alpha}^{\mathcal{A}}%
(\mathcal{N}\Vert\mathcal{R}).
\end{equation}
However, it is not known whether these quantities are equal for $\alpha
\in(0,1)$ or whether they are equal to the adversarial divergence $D_{\alpha
}^{\text{adv}}(\mathcal{N}\Vert\mathcal{R})$ for any $\alpha\in(0,1)\cup
(1,\infty)$. At the same time, one of the most compelling reasons to fix the
definition of channel R\'{e}nyi entropy as we have done is that the channel
divergence $D_{\alpha}(\mathcal{N}\Vert\mathcal{R})$\ has both a convincing
operational interpretation in channel discrimination as the optimal strong
converse exponent and the entropy function satisfies all of the desired axioms
for an entropy function. Furthermore, the entropy function $H_{\alpha
}(\mathcal{N})$ represents a useful bridge between the entropy and min-entropy
of a quantum channel, due to the facts that $\lim_{\alpha\rightarrow
1}H_{\alpha}(\mathcal{N})=H(\mathcal{N})$, $\lim_{\alpha\rightarrow\infty
}H_{\alpha}(\mathcal{N})=H_{\min}(\mathcal{N})$, and $H_{\alpha}%
(\mathcal{N})\leq H_{\beta}(\mathcal{N})$ for $\alpha\geq\beta\geq1$.

One should notice that we did not define the R\'enyi entropy of a channel in
terms of the Petz--R\'enyi relative entropy and the resulting channel
divergence, amortized channel divergence, or adversarial divergence. One of
the main reasons for this is that it is not known whether the resulting
entropy functions are additive. Furthermore, operational interpetations for
these divergences have not been established, having been open since the paper
\cite{CMW16} appeared. As such, it very well could be the case that one could
derive cogent notions of channel entropy from the Petz--R\'enyi relative
entropy, but this remains the topic of future work.

\subsection{Entropy functions derived from Choi and adversarial Choi
divergences}

\label{sec:ent-funcs-Choi-adv-Choi}In this subsection, we discuss various
entropy functions derived from Choi and adversarial Choi divergences. As
emphasized previously, we note again here that the operational interpretations
for these divergences are really about state discrimination tasks rather than
channel discrimination tasks. Nevertheless, the resulting entropy functions
satisfy the axioms put forward in \cite{G18}.

By picking the divergence $\mathbf{D}$ to be the quantum relative entropy $D$,
we find that the Choi and adversarial Choi divergences are equal when
discriminating an arbitrary channel $\mathcal{N}_{A\rightarrow B}$ from the
completely randomizing channel $\mathcal{R}_{A\rightarrow B}$:%
\begin{equation}
D^{\Phi}(\mathcal{N}\Vert\mathcal{R})=D^{\text{adv},\Phi}(\mathcal{N}%
\Vert\mathcal{R}).
\end{equation}
The proof of this statement follows along the lines of
\eqref{eq:adv-div-to-ch-div-1}--\eqref{eq:adv-div-to-ch-div-last}. There is a
simple operational interpretation for $D^{\Phi}(\mathcal{N}\Vert\mathcal{R})$
in terms of state discrimination \cite{HP91,ON00}, while an operational
interpretation for $D^{\text{adv},\Phi}(\mathcal{N}\Vert\mathcal{R})$ in terms
of state discrimination was given recently in \cite{HT16}.

We could also pick the divergence $\mathbf{D}$ to be Petz--R\'enyi relative
entropy $\overline{D}_{\alpha}$ or the sandwiched R\'enyi relative entropy
$D_{\alpha}$. The resulting Choi and adversarial Choi divergences are then
generally not equal when discriminating an arbitrary channel $\mathcal{N}%
_{A\rightarrow B}$ from the completely randomizing channel $\mathcal{R}%
_{A\rightarrow B}$. There is an operational interpretation for $\overline
{D}_{\alpha}^{\Phi}(\mathcal{N}\Vert\mathcal{R})$ for $\alpha\in(0,1)$ in
terms of state discrimination \cite{hayashi2007error,Nagaoka06}\ (error
exponent problem), and there is an operational interpretation for $D_{\alpha
}^{\Phi}(\mathcal{N}\Vert\mathcal{R})$ for $\alpha\in(1,\infty)$ in terms of
state discrimination \cite{MO13}\ (strong converse exponent problem).
Interestingly, \cite{HT16} has given a meaningful operational interpretation
for the adversarial Choi divergences $\overline{D}_{\alpha}^{\text{adv},\Phi
}(\mathcal{N}\Vert\mathcal{R})$ for $\alpha\in(0,1)$ and $D_{\alpha
}^{\text{adv},\Phi}(\mathcal{N}\Vert\mathcal{R})$ for $\alpha\in(1,\infty)$ in
terms of error exponent and strong converse exponent state discrimination
problems, respectively.

For $\mathcal{N}_{A\rightarrow B}$ a quantum channel and $\Phi_{RB}%
^{\mathcal{N}}\equiv\mathcal{N}_{A\rightarrow B}(\Phi_{RA})$ the Choi state,
the resulting channel entropy functions are then as follows:%
\begin{align}
H^{\Phi}(\mathcal{N})  &  \equiv H(B|R)_{\Phi^{\mathcal{N}}}=H^{\text{adv}%
,\Phi}(\mathcal{N}),\label{eq:Choi-ent-func-1}\\
H_{\alpha}^{\Phi}(\mathcal{N})  &  \equiv H_{\alpha}(B|R)_{\Phi^{\mathcal{N}%
}|\Phi^{\mathcal{N}}},\label{eq:Choi-ent-func-2}\\
\overline{H}_{\alpha}^{\Phi}(\mathcal{N})  &  \equiv\overline{H}_{\alpha
}(B|R)_{\Phi^{\mathcal{N}}|\Phi^{\mathcal{N}}},\label{eq:Choi-ent-func-3}\\
H_{\alpha}^{\text{adv},\Phi}(\mathcal{N})  &  \equiv H_{\alpha}(B|R)_{\Phi
^{\mathcal{N}}},\label{eq:Choi-ent-func-4}\\
\overline{H}_{\alpha}^{\text{adv},\Phi}(\mathcal{N})  &  \equiv\overline
{H}_{\alpha}(B|R)_{\Phi^{\mathcal{N}}}. \label{eq:Choi-ent-func-last}%
\end{align}
It then follows that all of the above entropy functions are additive for
$\alpha\in(0,1)\cup(1,\infty)$ (with the exception of additivity holding for
$H_{\alpha}^{\text{adv},\Phi}(\mathcal{N})$ for $\alpha\in\lbrack
1/2,1)\cup(1,\infty)$), due to the facts that the Choi state of a
tensor-product channel is equal to the tensor product of the Choi states of
the individual channels, as well as the additivity of the underlying
conditional entropies, for $H_{\alpha}^{\text{adv},\Phi}(\mathcal{N})$ shown
in \cite{T15} and for $\overline{H}_{\alpha}^{\text{adv},\Phi}(\mathcal{N})$
following from the quantum Sibson identity \cite[Lemma~7]{SW12} (see also
\cite[Lemma~1]{TBH14}). Normalization and reduction to states (as in
Proposition~\ref{prop:reduction-to-states}) follows for all of the above
quantities. What remains is monotonicity under random unitary superchannels,
and what we can show is something stronger:\ monotonicity under doubly
stochastic superchannels, the latter defined in \cite{G18} as superchannels
$\Theta$\ such that their adjoint $\Theta^{\dag}$\ is also a superchannel,
where the adjoint is defined with respect to the inner product for supermaps
considered in \cite{G18}.

\begin{theorem}
Let $\Theta$ be a doubly stochastic superchannel given by
\begin{equation}
\Theta\left[  \mathcal{N}_{A\rightarrow B}\right]  \equiv\Omega_{BE\rightarrow
D}\circ\mathcal{N}_{A\rightarrow B}\circ\Lambda_{C\rightarrow AE}\;.
\end{equation}
with $\Omega_{BE\rightarrow D}$ and $\Lambda_{C\rightarrow AE}$ quantum
channels, $E$ a quantum memory system, $|A|=|C|$, and $|B|=|D|$. Then, for
$\mathbf{H}$ any of the entropy functions in
\eqref{eq:Choi-ent-func-1}--\eqref{eq:Choi-ent-func-last}, the following
inequality holds%
\begin{equation}
\mathbf{H}\left(  \Theta\left[  \mathcal{N}_{A\rightarrow B}\right]  \right)
\geq\mathbf{H}(\mathcal{N}_{A\rightarrow B})\;. \label{eq:mono-doubly-stoc}%
\end{equation}
The inequality above holds for $\alpha\in\lbrack1/2,1)\cup(1,\infty)$ for the
functions in \eqref{eq:Choi-ent-func-2} and \eqref{eq:Choi-ent-func-4} and for
$\alpha\in(0,1)\cup(1,2]$ for the functions in \eqref{eq:Choi-ent-func-3} and \eqref{eq:Choi-ent-func-last}.
\end{theorem}

\begin{proof}
Recall from \cite{G18}\ that, since $\Theta$ is doubly stochastic, we have
that%
\begin{align}
\operatorname{Tr}_{E}\{\Lambda_{C\rightarrow AE}(I_{C})\}  &  =I_{A}%
,\label{c1}\\
\Omega_{BE\rightarrow D}\left(  I_{B}\otimes\rho_{ER}\right)   &
=I_{D}\otimes\rho_{R}\;.
\end{align}
Let $\Theta$ be as above, and let us begin by considering the adversarial
quantities for the ranges of $\alpha$ for which data processing holds. Let
$\omega_{R}$ be an arbitrary state. Let $\xi_{AER}\equiv\Lambda_{C\rightarrow
AE}\left(  \Phi_{CR}\right)  $, and note that the marginal $\xi_{A}$ is the
maximally mixed state due to~\eqref{c1} and the dimension constraint
$|A|=|C|$. Therefore, there exists a quantum channel $\mathcal{E}%
_{R\rightarrow ER}$ such that
\begin{equation}
\xi_{AER}=\mathcal{E}_{R\rightarrow ER}(\Phi_{AR})\;.
\label{eq:key-relation-doubly-stoch}%
\end{equation}
Let $\sigma_{ER}\equiv\mathcal{E}_{R\rightarrow ER}(\omega_{R})$. With these
notations set, and working with the specific entropy function in
\eqref{eq:Choi-ent-func-4}, we find that\begin{widetext}
\begin{align}
H_{\alpha}^{\text{adv},\Phi}(\Theta\left[  \mathcal{N}_{A\rightarrow
B}\right]  )  &  \geq-D_{\alpha}(\Theta\left[  \mathcal{N}_{A\rightarrow
B}\right]  \left(  \Phi_{CR}\right)  \Vert I_{D}\otimes\sigma_{R})\\
&  =-D_{\alpha}(\Omega_{BE\rightarrow D}\circ\mathcal{N}_{A\rightarrow B}%
\circ\Lambda_{C\rightarrow AE}\left(  \Phi_{CR}\right)  \Vert I_{D}%
\otimes\sigma_{R})\\
&  =-D_{\alpha}(\Omega_{BE\rightarrow D}\circ\mathcal{N}_{A\rightarrow B}%
\circ\Lambda_{C\rightarrow AE}\left(  \Phi_{CR}\right)  \Vert\Omega
_{BE\rightarrow D}\left(  I_{B}\otimes\sigma_{ER}\right)  )\\
&  \geq-D_{\alpha}(\mathcal{N}_{A\rightarrow B}\circ\Lambda_{C\rightarrow
AE}\left(  \Phi_{CR}\right)  \Vert I_{B}\otimes\sigma_{ER})\\
&  =-D_{\alpha}(\mathcal{N}_{A\rightarrow B}\left(  \xi_{AER}\right)  \Vert
I_{B}\otimes\sigma_{ER})\\
&  =-D_{\alpha}(\mathcal{E}_{R\rightarrow ER}\circ\mathcal{N}_{A\rightarrow
B}\left(  \Phi_{AR}\right)  \Vert\mathcal{E}_{R\rightarrow ER}(I_{B}%
\otimes\omega_{R}))\\
&  \geq-D_{\alpha}(\mathcal{N}_{A\rightarrow B}\left(  \Phi_{AR}\right)  \Vert
I_{B}\otimes\omega_{R})\ .
\end{align}
\end{widetext}
Since the inequality holds for an arbitrary state $\omega_{R}$, we conclude
that%
\begin{equation}
H_{\alpha}^{\text{adv},\Phi}(\Theta\left[  \mathcal{N}_{A\rightarrow
B}\right]  )\geq H_{\alpha}^{\text{adv},\Phi}(\mathcal{N}_{A\rightarrow B}),
\end{equation}
which is the inequality in \eqref{eq:mono-doubly-stoc} for the adversarial
Choi R\'{e}nyi entropy $H_{\alpha}^{\text{adv},\Phi}(\mathcal{N})$. The proof
for the entropy functions in \eqref{eq:Choi-ent-func-1} and
\eqref{eq:Choi-ent-func-last} goes the same way, since the above proof only
relied upon the data processing inequality.

To arrive at the inequality in \eqref{eq:mono-doubly-stoc} for the entropy
functions in \eqref{eq:Choi-ent-func-2}--\eqref{eq:Choi-ent-func-3}, we
exploit the same proof, but we choose $\omega_{R}$ to be the maximally mixed
state. By tracing over systems $AE$ in \eqref{eq:key-relation-doubly-stoch},
we find that%
\begin{align}
\pi_{R}=\operatorname{Tr}_{AE}\{\xi_{AER}\} &  =\operatorname{Tr}%
_{AE}\{\mathcal{E}_{R\rightarrow ER}(\Phi_{AR})\}\\
&  =(\operatorname{Tr}_{E}\circ\mathcal{E}_{R\rightarrow ER})(\pi_{A}),
\end{align}
and so we conclude that the reduced channel $\operatorname{Tr}_{E}%
\circ\mathcal{E}_{R\rightarrow ER}$ is unital. This means that, by choosing
$\sigma_{ER}\equiv\mathcal{E}_{R\rightarrow ER}(\omega_{R})$ again, we can
conclude that $\sigma_{R}=\pi_{R}$. By applying the same steps as above, we
then find that%
\begin{align}
H_{\alpha}^{\Phi}(\Theta\left[  \mathcal{N}_{A\rightarrow B}\right]  )  &
=-D_{\alpha}(\Theta\left[  \mathcal{N}_{A\rightarrow B}\right]  \left(
\Phi_{CR}\right)  \Vert I_{D}\otimes\pi_{R})\\
&  \geq-D_{\alpha}(\mathcal{N}_{A\rightarrow B}\left(  \Phi_{AR}\right)  \Vert
I_{B}\otimes\pi_{R})\\
&  =H_{\alpha}^{\Phi}(\mathcal{N}_{A\rightarrow B}),
\end{align}
which is the inequality in \eqref{eq:mono-doubly-stoc} for the entropy
function in \eqref{eq:Choi-ent-func-2}. The proof for the entropy function in
\eqref{eq:Choi-ent-func-3} then goes the same way.
\end{proof}

\bigskip

As a final remark to conclude this section, we note that the following limit
holds%
\begin{equation}
\lim_{\alpha\rightarrow\infty}H_{\alpha}^{\text{adv},\Phi}(\mathcal{N}%
_{A\rightarrow B})=H_{\min}^{\mathrm{ext}}(\mathcal{N}_{A\rightarrow B}),
\end{equation}
as a consequence of \eqref{eq:max-rel-as-limit}, and so the proof given above
represents a different way, from that given in \cite{G18}, for arriving at the
conclusion that the extended min-entropy of a channel is non-decreasing under
the action of a doubly stochastic superchannel.

\section{Conclusion and outlook}

\label{sec:concl}

In this paper, we have introduced a definition for the entropy of a quantum
channel, based on the channel relative entropy between the channel of interest
and the completely randomizing channel. Building on this approach, we defined
the R\'{e}nyi and min-entropy of a channel. We proved that these channel
entropies satisfy the axioms for entropy functions, recently put forward in
\cite{G18}. We also proved that the entropy of a channel is equal to the
completely bounded entropy of \cite{DJKR06}, and the R\'{e}nyi entropy of a
channel is related to the completely bounded $1\rightarrow p$ norm considered
in \cite{DJKR06}. The smoothed min-entropy of a channel satisfies an
asymptotic equipartition property that generalizes the same property for
smoothed min-entropy of quantum states \cite{TCR09}. We showed that the
entropy of a channel has an operational interpretation in terms of a task
called quantum channel merging, in which the goal is for the receiver to merge
his share of the channel with the environment's share, and this task is a
dynamical counterpart of the known task of quantum state merging
\cite{HOW05,HOW07}. We evaluated the entropy of a channel for several common
channel models. Finally, we considered other generalized entropies of a
quantum channel and gave further evidence that
Definition~\ref{def:entropy-channel}\ is a cogent approach for defining
entropy of a quantum channel.

Going forward from here, one of the most interesting open questions is to
determine if there is a set of axioms that uniquely identifies the entropy of
a quantum channel, similar to how there is a set of axioms that uniquely
characterizes Shannon entropy \cite{C08}. We wonder the same for the R\'{e}nyi
entropy of a channel, given that the R\'{e}nyi entropies were originally
identified \cite{Renyi61}\ by removing one of the axioms that uniquely
characterizes Shannon entropy. On a different front, one could alternatively
define the entropy of $n$ uses of a quantum channel in terms of an
optimization over quantum co-strategies \cite{GW07,G09}\ or quantum combs
\cite{CDP08a}, and for analyzing the asymptotic equipartition property in this
scenario, one could alternatively smooth with respect to the strategy norm of
\cite{CDP08a,G12}. The results of \cite{CMW16} suggest that the asymptotic
equipartition property might still hold in this more complex scenario, but
further analysis is certainly required. Note that a related scenario has been
considered recently in \cite{CE16}. Finally, if the Petz--R\'{e}nyi channel
divergence between an arbitrary channel and the completely depolarizing
channel is additive, then a R\'{e}nyi channel entropy defined from it would be
convincing. This question about Petz--R\'{e}nyi channel divergence has been
open since \cite{CMW16}.

\begin{acknowledgments}
We are grateful to the local organizers (Graeme Smith and Felix Leditzky)\ of
the Rocky Mountain Summit on Quantum Information, held at JILA, Boulder,
Colorado during June 2018. We are especially grateful to Xiao Yuan for
notifying us of a gap in a previous proof of \eqref{eq:lower-AEP}. GG
acknowledges support from the Natural Sciences and Engineering Research
Council of Canada (NSERC). MMW\ acknowledges support from the National Science
Foundation under grant nos.~1714215 and 1907615.
\end{acknowledgments}


\bibliography{Ref}
\bibliographystyle{unsrt}

\appendix

\section{Quantum channel merging capacity proof}

\label{sec:channel-merging-proof}

This appendix details the proof of Theorem~\ref{thm:opt-int-ent-ch}.

\subsection{Converse bound}

Let us begin by considering the converse part, following the approach given in
\cite{HOW07}\ for quantum state merging.

\begin{proposition}
\label{prop:1-shot-converse}Fix $n,L,K\in\mathbb{N}$ and $\varepsilon
\in\lbrack0,1]$. Let $\mathcal{N}_{A\rightarrow B}$ be a quantum channel. Then
an $(n,L/K,\varepsilon)$ quantum channel merging protocol for $\mathcal{N}%
_{A\rightarrow B}$ satisfies the following bound:%
\begin{multline}
\frac{1}{n}\left[  (1-\sqrt{\varepsilon})\log_{2}L-\log_{2}K\right] \\
\leq H(\mathcal{N})+\sqrt{\varepsilon}\log_{2}\left\vert A\right\vert
+g_{2}(\sqrt{\varepsilon}).
\end{multline}

\end{proposition}

\begin{proof}
We closely follow the approach given in \cite[Section~IV-B]{HOW07}, which
established the converse part of the quantum state merging theorem. Consider
an arbitrary $(n,L/K,\varepsilon)$ quantum channel merging protocol of the
form described above. To prove the converse, we can really employ any
entanglement measure that reduces to the entropy of entanglement for pure
states and is asymptotically continuous. So let us choose the entanglement of
formation \cite{BDSW96}, which is defined for a bipartite state $\rho_{AB}$
as
\begin{multline}
E_{F}(A;B)_{\rho}\equiv\\
\inf\left\{  \sum_{x}p_{X}(x)H(A)_{\psi^{x}}:\rho_{AB}=\sum_{x}p_{X}%
(x)\psi_{AB}^{x}\right\}  , \label{eq:EoF}%
\end{multline}
where the infimum is with respect to all convex decompositions of $\rho_{AB}$
into pure states $\psi_{AB}^{x}$. The entanglement of formation does not
increase under the action of an LOCC\ channel \cite{BDSW96}. For the purposes
of the converse, as in \cite[Section~IV-B]{HOW07}, we imagine that the
reference party$~R$ is working together with Bob$~B$, and they are spatially
separated from Eve$~E$. Let $\omega_{R\widetilde{B}_{E}^{n}\widetilde{E}%
^{n}\overline{B}_{1}\overline{E}_{1}}$ and $\widetilde{\omega}_{R\widetilde
{B}_{E}^{n}\widetilde{E}^{n}\overline{B}_{1}\overline{E}_{1}}$\ denote the
following respective states:%
\begin{multline}
\omega_{R\widetilde{B}_{E}^{n}\widetilde{E}^{n}\overline{B}_{1}\overline
{E}_{1}} \equiv\\
[\operatorname{id}_{BE\rightarrow\widetilde{B}_{E}\widetilde{E}}^{\otimes
n}\circ(\mathcal{U}_{A\rightarrow BE}^{\mathcal{N}})^{\otimes n}](\psi
_{RA^{n}})\otimes\Phi_{\overline{B}_{1}\overline{E}_{1}}^{L},
\end{multline}
\begin{multline}
\widetilde{\omega}_{R\widetilde{B}_{E}^{n}\widetilde{E}^{n}\overline{B}%
_{1}\overline{E}_{1}} \equiv\\
\mathcal{P}_{B^{n}E^{n}\overline{B}_{0}\overline{E}_{0}\rightarrow
\widetilde{B}_{E}^{n}\widetilde{E}^{n}\overline{B}_{1}\overline{E}_{1}%
}([(\mathcal{U}_{A\rightarrow BE}^{\mathcal{N}})^{\otimes n}(\psi_{RA^{n}%
})]\otimes\Phi_{\overline{B}_{0}\overline{E}_{0}}^{K}).
\end{multline}
Define%
\begin{equation}
f(n,\varepsilon,\left\vert A\right\vert ,L)\equiv\sqrt{\varepsilon}n\log
_{2}\left\vert A\right\vert +\sqrt{\varepsilon}\log_{2}L+g_{2}(\sqrt
{\varepsilon}).
\end{equation}
We then have that%
\begin{align}
& \log_{2}L+H(R)_{\omega}\nonumber\\
&  =H(\overline{B}_{1})_{\omega}+H(R)_{\omega}\\
&  =H(R\overline{B}_{1})_{\omega}\\
&  =E_{F}(R\overline{B}_{1};\widetilde{B}_{E}^{n}\widetilde{E}^{n}\overline
{E}_{1})_{\omega}\\
&  \leq E_{F}(R\overline{B}_{1};\widetilde{B}_{E}^{n}\widetilde{E}%
^{n}\overline{E}_{1})_{\widetilde{\omega}}+f(n,\varepsilon,\left\vert
A\right\vert ,L).
\end{align}
The first equality follows because $\log_{2}L=H(\overline{B}_{1})_{\omega}$
for a maximally entangled state of Schmidt rank $L$. The second equality
follows because quantum entropy is additive with respect to product states.
The third equality follows because the entanglement of formation reduces to
entropy of entanglement for pure states. The inequality is a consequence of
the uniform continuity bound from \cite[Corollary~4]{Winter15}. Continuing, we
have that%
\begin{align}
&  E_{F}(R\overline{B}_{1};\widetilde{B}_{E}^{n}\widetilde{E}^{n}\overline
{E}_{1})_{\widetilde{\omega}}\nonumber\\
&  \leq E_{F}(RB^{n}\overline{B}_{0};E^{n}\overline{E}_{0})_{(\mathcal{U}%
^{\mathcal{N}})^{\otimes n}(\psi)\otimes\Phi^{K}}\\
&  =H(RB^{n}\overline{B}_{0})_{(\mathcal{U}^{\mathcal{N}})^{\otimes n}%
(\psi)\otimes\Phi^{K}}\\
&  =H(RB^{n})_{(\mathcal{U}^{\mathcal{N}})^{\otimes n}(\psi)}+H(\overline
{B}_{0})_{\Phi^{K}}\\
&  =H(RB^{n})_{(\mathcal{U}^{\mathcal{N}})^{\otimes n}(\psi)}+\log_{2}K.
\end{align}
The first inequality follows from LOCC\ monotonicity of the entanglement of
formation under the action of the one-way LOCC\ channel $\mathcal{P}%
_{B^{n}E^{n}\overline{B}_{0}\overline{E}_{0}\rightarrow\widetilde{B}_{E}%
^{n}\widetilde{E}^{n}\overline{B}_{1}\overline{E}_{1}}$. The last three
equalities follow for reasons similar to what have been given above. Putting
everything together, we find that%
\begin{align}
\log_{2}M  &  =\log_{2}L-\log_{2}K\\
&  \leq H(B^{n}|R)_{(\mathcal{U}^{\mathcal{N}})^{\otimes n}(\psi
)}+f(n,\varepsilon,\left\vert A\right\vert ,L).
\end{align}
Since the protocol is required to work for every possible input state
$\psi_{RA^{n}}$, we conclude the following bound%
\begin{align}
\log_{2}M  &  \leq\inf_{\psi_{RA^{n}}}H(B^{n}|R)_{(\mathcal{U}^{\mathcal{N}%
})^{\otimes n}(\psi)}+f(n,\varepsilon,\left\vert A\right\vert ,L)\\
&  =nH(\mathcal{N})+f(n,\varepsilon,\left\vert A\right\vert ,L),
\end{align}
with the equality following from the additivity of the entropy of a channel
\cite{DJKR06} (recalled here as Proposition~\ref{prop:entropy-additive}). The
inequality in the statement of the proposition follows by dividing by $n$ and rearranging.
\end{proof}

\subsection{Achievability bound}

Now let us consider the achievability part.

\begin{proposition}
\label{prop:achievability}Fix $n,L,K\in\mathbb{N}$ and $\varepsilon\in(0,1)$.
Let $\mathcal{N}_{A\rightarrow B}$ be a quantum channel. Then there exists an
$(n,L/K,\varepsilon)$ channel merging protocol for $\mathcal{N}_{A\rightarrow
B}$ such that its entanglement gain satisfies the following inequality for all
$\alpha>1$:
\begin{multline}
\frac{1}{n}\left[  \log_{2}L-\log_{2}K\right]  \geq H_{\alpha}(\mathcal{N})\\
-\frac{\alpha}{n\left(  \alpha-1\right)  }\left[  4\log_{2}(1/\varepsilon
)+4(\left\vert A\right\vert ^{2}-1)\log_{2}(n+1)\right] \\
-\frac{\alpha}{n\left(  \alpha-1\right)  }\left[ 1/\alpha+2\log_{2}13\right]
.
\end{multline}

\end{proposition}

\begin{proof}
For the achievability part, we employ ideas used in the theory of quantum
channel simulation \cite{BCR09,BBCW13,BRW14,B13}. In particular, the main
challenge of quantum channel merging over quantum state merging is that it is
necessary for the protocol to work for every possible state $\psi_{RA^{n}}%
$\ that could be input, not merely for a fixed state input. In prior work on
quantum channel simulation \cite{BCR09,BBCW13,BRW14,B13}, this challenge has
been met by appealing to the post-selection technique \cite[Theorem~1]{CKR09}.
Here, we use the same approach. In the context of the post-selection
technique, it is helpful to consult the unpublished note \cite{M10} for
further details.

Let $\zeta_{A^{n}\hat{A}^{n}}$ denote the maximally mixed state of the
symmetric subspace of the $A^{n}\hat{A}^{n}$ systems \cite{H13}, where
$\hat{A}$ is isomorphic to the channel input system $A$. Note that this state
can be written as \cite[Proposition~6]{H13}%
\begin{equation}
\zeta_{A^{n}\hat{A}^{n}}=\int d\psi_{A\hat{A}}\ \psi_{A\hat{A}}^{\otimes n},
\end{equation}
where $\psi_{A\hat{A}}$ denotes a pure state and $d\psi_{A\hat{A}}$ is the
Haar measure over the pure states. This state is permutation invariant; i.e.,
for a unitary channel $\mathcal{W}_{A^{n}}^{\pi} \otimes\mathcal{W}_{\hat
{A}^{n}}^{\pi}$ corresponding to a permutation $\pi$, we have that
$\zeta_{A^{n}\hat{A}^{n}}=(\mathcal{W}_{A^{n}}^{\pi}\otimes\mathcal{W}%
_{\hat{A}^{n}}^{\pi})(\zeta_{A^{n}\hat{A}^{n}})$ for all $\pi\in S_{n}$, with
$S_{n}$ denoting the symmetric group. Let $\zeta_{R^{\prime}\hat{A}^{n}A^{n}}$
be a purification of $\zeta_{A^{n}\hat{A}^{n}}$, and note that it can be
chosen such that \cite{M10}%
\begin{equation}
\zeta_{R^{\prime}\hat{A}^{n}A^{n}}=(\mathcal{W}_{R^{\prime}}^{\pi}%
\otimes\mathcal{W}_{\hat{A}^{n}}^{\pi}\otimes\mathcal{W}_{A^{n}}^{\pi}%
)(\zeta_{R^{\prime}\hat{A}^{n}A^{n}}),
\end{equation}
where $\mathcal{W}_{R^{\prime}}^{\pi}$ is some unitary, which implies that%
\begin{equation}
(\mathcal{W}_{R^{\prime}}^{\pi^{-1}} \otimes\mathcal{W}_{\hat{A}^{n}}%
^{\pi^{-1}})(\zeta_{R^{\prime}\hat{A}^{n}A^{n}})=\mathcal{W}_{A^{n}}^{\pi
}(\zeta_{R^{\prime}\hat{A}^{n}A^{n}}). \label{eq:symmetry-de-finetti-state}%
\end{equation}

The first goal is to show the existence of a state merging protocol for the
state $(\mathcal{U}_{A\rightarrow BE}^{\mathcal{N}})^{\otimes n}%
(\zeta_{R^{\prime}\hat{A}^{n}A^{n}})$. As shown in \cite[Theorem~5.2]%
{DBWR14}\ (see also the earlier \cite[Proposition~4.7]{B09} in this context),
there exists a state merging protocol with error $\sqrt{13\varepsilon^{\prime
}}$, with the entanglement gain satisfying%
\begin{multline}
\log_{2}L-\log_{2}K\geq H_{\min}^{\varepsilon^{\prime}}(B^{n}|\hat{A}%
^{n}R^{\prime})_{(\mathcal{U}^{\mathcal{N}})^{\otimes n}(\zeta)}\\
-2\log_{2}\!\left(  \frac{1}{\varepsilon^{\prime}}\right)  .
\label{eq:berta-one-shot}%
\end{multline}
(To arrive at the inequality in \eqref{eq:berta-one-shot}, one needs to use
the fact that $P(\rho,\sigma) \geq\frac{1}{2}\left\Vert \rho-\sigma\right\Vert
_{1}$ for any two states.) That is, there exists a one-way LOCC\ channel
$\mathcal{P}_{B^{n}E^{n}\overline{B}_{0}\overline{E}_{0}\rightarrow
\widetilde{B}_{E}^{n}\widetilde{E}^{n}\overline{B}_{1}\overline{E}_{1}}$ such
that the following inequality holds \begin{widetext}
\begin{multline}
\frac{1}{2}\bigg\Vert[\operatorname{id}_{BE\rightarrow\widetilde{B}%
_{E}\widetilde{E}}^{\otimes n}\circ(\mathcal{U}_{A\rightarrow BE}%
^{\mathcal{N}})^{\otimes n}](\zeta_{R'\hat{A}^{n}A^{n}})\otimes\Phi
_{\overline{B}_{1}\overline{E}_{1}}^{L}\label{eq:de-finetti-performance}\\
-\mathcal{P}_{B^{n}E^{n}\overline{B}_{0}\overline{E}_{0}\rightarrow
\widetilde{B}_{E}^{n}\widetilde{E}^{n}\overline{B}_{1}\overline{E}_{1}%
}([(\mathcal{U}_{A\rightarrow BE}^{\mathcal{N}})^{\otimes n}(\zeta_{\hat
R'{A}^{n}A^{n}})]\otimes\Phi_{\overline{B}_{0}\overline{E}_{0}}^{K}%
)\bigg\Vert_{1}\leq\sqrt{13\varepsilon^{\prime}}.
\end{multline}
Now our goal is for \eqref{eq:criterion-CM} to be satisfied for all possible
states $\psi_{RA^{n}}$. As a first step toward this goal, note that we can
symmetrize the protocol $\mathcal{P}_{B^{n}E^{n}\overline{B}_{0}\overline
{E}_{0}\rightarrow\widetilde{B}_{E}^{n}\widetilde{E}^{n}\overline{B}%
_{1}\overline{E}_{1}}$ as follows%
\begin{equation}
\overline{\mathcal{P}}_{B^{n}E^{n}\overline{B}_{0}\overline{E}_{0}%
\rightarrow\widetilde{B}_{E}^{n}\widetilde{E}^{n}\overline{B}_{1}\overline
{E}_{1}}\equiv\frac{1}{n!}\sum_{\pi\in S_{n}}\left(  \mathcal{W}%
_{\widetilde{B}_{E}^{n}}^{\pi^{-1}}\otimes\mathcal{W}_{\widetilde{E}^{n}}%
^{\pi^{-1}}\right)  \circ\mathcal{P}_{B^{n}E^{n}\overline{B}_{0}\overline
{E}_{0}\rightarrow\widetilde{B}_{E}^{n}\widetilde{E}^{n}\overline{B}%
_{1}\overline{E}_{1}}\circ\left(  \mathcal{W}_{B^{n}}^{\pi}\otimes
\mathcal{W}_{E^{n}}^{\pi}\right)  ,
\end{equation}
and the inequality in \eqref{eq:de-finetti-performance} is still satisfied,
i.e.,%
\begin{multline}
\frac{1}{2}\bigg\Vert[\operatorname{id}_{BE\rightarrow\widetilde{B}%
_{E}\widetilde{E}}^{\otimes n}\circ(\mathcal{U}_{A\rightarrow BE}%
^{\mathcal{N}})^{\otimes n}](\zeta_{R'\hat{A}^{n}A^{n}})\otimes\Phi
_{\overline{B}_{1}\overline{E}_{1}}^{L}\\
-\overline{\mathcal{P}}_{B^{n}E^{n}\overline{B}_{0}\overline{E}_{0}%
\rightarrow\widetilde{B}_{E}^{n}\widetilde{E}^{n}\overline{B}_{1}\overline
{E}_{1}}\!\left(  [(\mathcal{U}_{A\rightarrow BE}^{\mathcal{N}})^{\otimes
n}(\zeta_{R'\hat{A}^{n}A^{n}})]\otimes\Phi_{\overline{B}_{0}\overline{E}_{0}%
}^{K}\right)  \bigg\Vert_{1}\leq\sqrt{13\varepsilon^{\prime}}.
\end{multline}
This follows from the unitary invariance and convexity of the trace norm, the
permutation covariance of the maps $[\operatorname{id}_{BE\rightarrow
\widetilde{B}_{E}\widetilde{E}}^{\otimes n}\circ(\mathcal{U}_{A\rightarrow
BE}^{\mathcal{N}})^{\otimes n}]$ and $(\mathcal{U}_{A\rightarrow
BE}^{\mathcal{N}})^{\otimes n}$:%
\begin{align}
\forall\pi &  \in S_{n}:\left(  \mathcal{W}_{\widetilde{B}_{E}^{n}}^{\pi^{-1}%
}\otimes\mathcal{W}_{\widetilde{E}^{n}}^{\pi^{-1}}\right)  \circ
\operatorname{id}_{BE\rightarrow\widetilde{B}_{E}\widetilde{E}}^{\otimes
n}\circ(\mathcal{U}_{A\rightarrow BE}^{\mathcal{N}})^{\otimes n}%
\circ\mathcal{W}_{A^{n}}^{\pi}=\operatorname{id}_{BE\rightarrow\widetilde
{B}_{E}\widetilde{E}}^{\otimes n}\circ(\mathcal{U}_{A\rightarrow
BE}^{\mathcal{N}})^{\otimes n},\\
\forall\pi &  \in S_{n}:\left(  \mathcal{W}_{B^{n}}^{\pi^{-1}}\otimes
\mathcal{W}_{E^{n}}^{\pi^{-1}}\right)  \circ(\mathcal{U}_{A\rightarrow
BE}^{\mathcal{N}})^{\otimes n}\circ\mathcal{W}_{A^{n}}^{\pi}=(\mathcal{U}%
_{A\rightarrow BE}^{\mathcal{N}})^{\otimes n},
\end{align}
and the equality in \eqref{eq:symmetry-de-finetti-state}. Furthermore, the
symmetrization can be accomplished by one-way LOCC (Bob randomly picks $\pi$,
applies $\mathcal{W}_{B^{n}}^{\pi}$, communicates the value to Eve, who
applies $\mathcal{W}_{E^{n}}^{\pi}$ at the input and $\mathcal{W}%
_{\widetilde{B}_{E}^{n}}^{\pi^{-1}}\otimes\mathcal{W}_{\widetilde{E}^{n}}%
^{\pi^{-1}}$ at the output), and is thus free in our model. Since the
symmetrized protocol, the target channel $[\operatorname{id}_{BE\rightarrow
\widetilde{B}_{E}\widetilde{E}}^{\otimes n}\circ(\mathcal{U}_{A\rightarrow
BE}^{\mathcal{N}})^{\otimes n}]$, and the channel $(\mathcal{U}_{A\rightarrow
BE}^{\mathcal{N}})^{\otimes n}$ are permutation covariant, we can now invoke
the post-selection technique \cite[Theorem~1]{CKR09}\ to conclude that as long
as we choose $\varepsilon^{\prime}=\varepsilon\left(  n+1\right)  ^{-2\left(
\left\vert A\right\vert ^{2}-1\right)  }$, then it is guaranteed that%
\begin{multline}
\sup_{\psi_{RA^{n}}}\frac{1}{2}\bigg\Vert[\operatorname{id}_{BE\rightarrow
\widetilde{B}_{E}\widetilde{E}}^{\otimes n}\circ(\mathcal{U}_{A\rightarrow
BE}^{\mathcal{N}})^{\otimes n}](\psi_{RA^{n}})\otimes\Phi_{\overline{B}%
_{1}\overline{E}_{1}}^{L}\\
-\overline{\mathcal{P}}_{B^{n}E^{n}\overline{B}_{0}\overline{E}_{0}%
\rightarrow\widetilde{B}_{E}^{n}\widetilde{E}^{n}\overline{B}_{1}\overline
{E}_{1}}([(\mathcal{U}_{A\rightarrow BE}^{\mathcal{N}})^{\otimes n}%
(\psi_{RA^{n}})]\otimes\Phi_{\overline{B}_{0}\overline{E}_{0}}^{K}%
)\bigg\Vert_{1}\leq\sqrt{13\varepsilon}.
\end{multline}
Propagating this choice of $\varepsilon^{\prime}$ to the quantity in
\eqref{eq:berta-one-shot}, this means that we require%
\begin{equation}
\log_{2}L-\log_{2}K\geq H_{\min}^{\varepsilon\left(  n+1\right)  ^{-2\left(
\left\vert A\right\vert ^{2}-1\right)  }}(B^{n}|\hat{A}^{n}R')_{(\mathcal{U}%
^{\mathcal{N}})^{\otimes n}(\zeta)}-2\log_{2}\!\left(  \frac{1}{\varepsilon
}\right)  -4\left(  \left\vert A\right\vert ^{2}-1\right)  \log_{2}\left(
n+1\right)  .
\end{equation}
At this point, we invoke \cite[Eq.~(6.92)]{T15}, as well as the inequality
$1-\sqrt{1-\delta^{2}}\geq\delta^{2}/2$ holding for all $\delta\in\left[
0,1\right]  $, to conclude the following bound for $\alpha>1$:%
\begin{align}
&  \!\!\!\! H_{\min}^{\varepsilon\left(  n+1\right)  ^{-2\left(  \left\vert
A\right\vert ^{2}-1\right)  }}(B^{n}|\hat{A}^{n}R')_{(\mathcal{U}^{\mathcal{N}%
})^{\otimes n}(\zeta)}\nonumber\\
&  \geq H_{\alpha}(B^{n}|\hat{A}^{n}R')_{(\mathcal{U}^{\mathcal{N}})^{\otimes
n}(\zeta)|(\mathcal{U}^{\mathcal{N}})^{\otimes n}(\zeta)}+\frac{2\log
_{2}(\varepsilon\left(  n+1\right)  ^{-2\left(  \left\vert A\right\vert
^{2}-1\right)  })-1}{\alpha-1}\\
&  \geq\inf_{\phi_{RA^{n}}}H_{\alpha}(B^{n}|R)_{\omega|\omega}-\frac{2\log
_{2}(1/\varepsilon)+4\left(  \left\vert A\right\vert ^{2}-1\right)  \log
_{2}\left(  n+1\right)  -1}{\alpha-1}\\
&  =H_{\alpha}(\mathcal{N}^{\otimes n})-\frac{2\log_{2}(1/\varepsilon
)+4\left(  \left\vert A\right\vert ^{2}-1\right)  \log_{2}\left(  n+1\right)
-1}{\alpha-1}\\
&  =nH_{\alpha}(\mathcal{N})-\frac{2\log_{2}(1/\varepsilon)+4\left(
\left\vert A\right\vert ^{2}-1\right)  \log_{2}\left(  n+1\right)  -1}%
{\alpha-1}.
\end{align}
\end{widetext}
where the first equality follows from
Proposition~\ref{prop:identities-for-renyi}, with $\omega_{RB^{n}}%
\equiv\mathcal{N}_{A\rightarrow B}^{\otimes n}(\phi_{RA^{n}})$,\ and the last
equality critically relies upon the additivity $H_{\alpha}(\mathcal{N}%
^{\otimes n})=nH_{\alpha}(\mathcal{N})$\ from
Proposition~\ref{prop:renyi-additive}, which in turn directly follows from the
main result of \cite{DJKR06}. Putting everything together, we conclude that
for $\varepsilon\in(0,1/13)$, there exists an $(n,L/K,\sqrt{13\varepsilon})$
channel merging protocol for $\mathcal{N}_{A\rightarrow B}$ such that its
entanglement gain satisfies the following inequality for all $\alpha>1$:
\begin{multline}
\frac{1}{n}\left[  \log_{2}L-\log_{2}K\right]  \geq H_{\alpha}(\mathcal{N})\\
-\frac{\alpha}{n\left(  \alpha-1\right)  }\left[  2\log_{2}(\frac
{1}{\varepsilon} )+4\left(  \left\vert A\right\vert ^{2}-1\right)  \log
_{2}(n+1)+\frac{1}{\alpha} \right]  .
\end{multline}
We arrive at the statement of the proposition by a final substitution
$\varepsilon^{\prime\prime}=\sqrt{13\varepsilon}\in(0,1)$, which implies that
$\varepsilon=(\varepsilon^{\prime\prime})^{2}/13$ and $2\log_{2}%
(1/\varepsilon)=4\log_{2}(1/\varepsilon^{\prime\prime})+2\log_{2}13$.
\end{proof}

\subsection{Quantum channel merging capacity is equal to the entropy of a
channel}

We can now put together the previous two propositions to conclude the
following theorem:

\begin{proof}
[Proof of Theorem~\ref{thm:opt-int-ent-ch}]By applying the limits
$n\rightarrow\infty$ and $\varepsilon\rightarrow0$, the following bound is a
consequence of Proposition~\ref{prop:1-shot-converse}:%
\begin{equation}
C_{M}(\mathcal{N})\leq H(\mathcal{N}).
\end{equation}
For an arbitrary $\alpha>1$, $\varepsilon\in(0,1)$, and $\delta>0$, we can
conclude from Proposition~\ref{prop:achievability} that there exists an
$(n,2^{n\left[  H_{\alpha}(\mathcal{N})-\delta\right]  },\varepsilon)$ channel
merging protocol by taking $n$ sufficiently large. This implies that
$H_{\alpha}(\mathcal{N})$ is an achievable rate for all $\alpha>1$. However,
since this statement is true for all $\alpha>1$, we can conclude that the rate
$\sup_{\alpha>1}H_{\alpha}(\mathcal{N})=H(\mathcal{N})$ is achievable also.
This establishes that $C_{M}(\mathcal{N})\geq H(\mathcal{N})$.
\end{proof}

\section{Max-mutual information of a channel and the asymptotic equipartition
property}

\label{app:max-mut-AEP}In this appendix, we point out how the max-mutual
information of a quantum channel is a limit of the sandwiched R\'{e}nyi mutual
information of a channel, the latter having been defined in \cite{GW15}. We
then show how to arrive at an alternate proof of the asymptotic equipartition
property in \cite[Theorem~8]{FWTB18}\ by making use of this connection.

First recall that the sandwiched R\'{e}nyi mutual information of a channel is
defined for $\alpha\in(0,1)\cup(1,\infty)$ as \cite[Eq.~(3.5)]{GW15}%
\begin{equation}
I_{\alpha}(\mathcal{N})\equiv\max_{\psi_{RA}}I_{\alpha}(R;B)_{\omega},
\end{equation}
where%
\begin{align}
\omega_{RB}  &  \equiv\mathcal{N}_{A\rightarrow B}(\psi_{RA}),\\
I_{\alpha}(R;B)_{\omega}  &  \equiv\min_{\sigma_{B}}D_{\alpha}(\omega
_{RB}\Vert\omega_{R}\otimes\sigma_{B}),
\end{align}
where $D_{\alpha}$ is the sandwiched R\'{e}nyi relative entropy from
\eqref{eq:sandwiched-Renyi}. It was subsequently used in \cite{CMW16}. The
max-mutual information of a channel is equal to \cite[Definition~4]{FWTB18}
\begin{align}
I_{\max}(\mathcal{N})  &  \equiv\max_{\psi_{RA}}I_{\max}(R;B)_{\omega},\\
I_{\max}(R;B)_{\omega}  &  \equiv\min_{\sigma_{B}}D_{\max}(\omega_{RB}%
\Vert\omega_{R}\otimes\sigma_{B}).
\end{align}

\begin{proposition}
For a quantum channel $\mathcal{N}_{A\rightarrow B}$, the following limit
holds
\begin{equation}
I_{\max}(\mathcal{N})=\lim_{\alpha\rightarrow\infty}I_{\alpha}(\mathcal{N}).
\end{equation}

\end{proposition}

\begin{proof}
To see this, consider that%
\begin{align}
&  \lim_{\alpha\rightarrow\infty}I_{\alpha}(\mathcal{N})\nonumber\\
&  =\lim_{\alpha\rightarrow\infty}\max_{\psi_{RA}}\min_{\sigma_{B}}D_{\alpha
}(\mathcal{N}_{A\rightarrow B}(\psi_{RA})\Vert\psi_{R}\otimes\sigma_{B})\\
&  =\sup_{\alpha>1}\max_{\psi_{RA}}\min_{\sigma_{B}}D_{\alpha}(\mathcal{N}%
_{A\rightarrow B}(\psi_{RA})\Vert\psi_{R}\otimes\sigma_{B})\\
&  \leq\max_{\psi_{RA}}\min_{\sigma_{B}}\sup_{\alpha>1}D_{\alpha}%
(\mathcal{N}_{A\rightarrow B}(\psi_{RA})\Vert\psi_{R}\otimes\sigma_{B})\\
&  =\max_{\psi_{RA}}\min_{\sigma_{B}}D_{\max}(\mathcal{N}_{A\rightarrow
B}(\psi_{RA})\Vert\psi_{R}\otimes\sigma_{B})\\
&  =I_{\max}(\mathcal{N}).
\end{align}
Now consider that%
\begin{align}
&  \lim_{\alpha\rightarrow\infty}I_{\alpha}(\mathcal{N})\nonumber\\
&  =\sup_{\alpha>1}\max_{\psi_{RA}}\min_{\sigma_{B}}D_{\alpha}(\mathcal{N}%
_{A\rightarrow B}(\psi_{RA})\Vert\psi_{R}\otimes\sigma_{B})\\
&  \geq\sup_{\alpha>1}\min_{\sigma_{B}}D_{\alpha}(\mathcal{N}_{A\rightarrow
B}(\Phi_{RA})\Vert\Phi_{R}\otimes\sigma_{B})\\
&  =\min_{\sigma_{B}}\sup_{\alpha>1}D_{\alpha}(\mathcal{N}_{A\rightarrow
B}(\Phi_{RA})\Vert\Phi_{R}\otimes\sigma_{B})\\
&  =\min_{\sigma_{B}}D_{\max}(\mathcal{N}_{A\rightarrow B}(\Phi_{RA})\Vert
\Phi_{R}\otimes\sigma_{B})\\
&  =I_{\max}(\mathcal{N}).
\end{align}
with the exchange of $\min$ and $\sup$ in the last line following from
\cite[Corollary~A.2]{MH11}. The last equality follows from the remark after
\cite[Definition~4]{FWTB18}.
\end{proof}

\bigskip

The smoothed max-mutual information of a quantum channel $\mathcal{N}%
_{A\rightarrow B}$ is then defined for $\varepsilon\in(0,1)$ as
\cite[Definition~5]{FWTB18}%
\begin{equation}
I_{\max}^{\varepsilon}(\mathcal{N})\equiv\inf_{\widetilde{\mathcal{N}%
}\ :\ P(\mathcal{N},\widetilde{\mathcal{N}})\leq\varepsilon}I_{\max
}(\widetilde{\mathcal{N}}).
\end{equation}
(Here we smooth with respect to purified distance for convenience.) We then
have that \cite[Theorem~8]{FWTB18}%
\begin{equation}
\lim_{n\rightarrow\infty}\frac{1}{n}I_{\max}^{\varepsilon}(\mathcal{N}%
^{\otimes n})\leq I(\mathcal{N}), \label{eq:AEP-max-mut}%
\end{equation}
where $I(\mathcal{N})$ is the mutual information of a channel \cite{AC97},
defined as%
\begin{equation}
I(\mathcal{N})=\lim_{\alpha\rightarrow1}\widetilde{I}_{\alpha}(\mathcal{N}%
)=\sup_{\psi_{RA}}I(R;B)_{\omega}, \label{eq:mut-inf-ch-def}%
\end{equation}
where $\omega_{RB}\equiv\mathcal{N}_{A\rightarrow B}(\psi_{RA})$.

To arrive at an alternate proof of the upper bound in \cite[Theorem~8]%
{FWTB18}, consider that an application of \cite[Eq.~(6.92)]{T15}, definitions,
and arguments similar to those in the first part of
Theorem~\ref{thm:AEP-channel-entropy} imply the following inequality for all
$\alpha>1$ and $\varepsilon\in(0,1)$:%
\begin{align}
I_{\max}^{\varepsilon}(\mathcal{N}^{\otimes n})  &  \leq I_{\alpha
}(\mathcal{N}^{\otimes n}) + f(\varepsilon,\alpha)\\
&  =nI_{\alpha}(\mathcal{N})+ f(\varepsilon,\alpha),
\end{align}
where the equality follows from \cite[Lemma~6]{GW15} and $f(\varepsilon
,\alpha)$ is a function of $\varepsilon$ and $\alpha$ that vanishes when
dividing by $n$ and taking the large $n$ limit. Dividing by $n$ and taking the
limit $n\rightarrow\infty$, we find that%
\begin{equation}
\lim_{n\rightarrow\infty}\frac{1}{n}I_{\max}^{\varepsilon}(\mathcal{N}%
^{\otimes n})\leq I_{\alpha}(\mathcal{N}).
\end{equation}
Since the inequality holds for all $\alpha>1$, we can take the limit
$\alpha\rightarrow1$, apply \eqref{eq:mut-inf-ch-def}, and conclude the bound
in \eqref{eq:AEP-max-mut}.

\section{Data processing of the Choi divergence under particular
superchannels}

\label{app:choi-div-monotone}

\begin{proposition}
Let $\Theta$ be a superchannel of the following form:
\begin{equation}
\Theta(\mathcal{N}_{A\rightarrow B})=\sum_{x}p(x)\;\Omega_{BE\rightarrow
D}^{x}\circ\mathcal{N}_{A\rightarrow B}\circ\Lambda_{C\rightarrow AE}^{x}
\label{maintheta}%
\end{equation}
where $p(x)$ is a probability distribution, and for each $x$ the map
$\Omega_{BE\rightarrow D}^{x}$ is an arbitrary quantum channel, and
$\Lambda_{C\rightarrow AE}^{x}$ is a unital quantum channel (hence
$|C|=|A||E|$). Then the Choi divergence is monotone under such superchannels:
\begin{equation}
\mathbf{D}^{\Phi}(\mathcal{N}\Vert\mathcal{M})\geq\mathbf{D}^{\Phi}%
(\Theta(\mathcal{N})\Vert\Theta(\mathcal{M})).
\end{equation}

\end{proposition}

\begin{proof}
To prove it, we first prove the monotonicity under any superchannel of the
form
\begin{equation}
\Upsilon(\mathcal{N}_{A\rightarrow B})=\Omega_{BE\rightarrow D}\circ
\mathcal{N}_{A\rightarrow B}\circ\Lambda_{C\rightarrow AE} \label{upsilon}%
\end{equation}
with $\Omega_{BE\rightarrow D}$ an arbitrary quantum channel, and
$\Lambda_{C\rightarrow AE}$ a unital quantum channel. Indeed, denoting by
$\Lambda_{AE\rightarrow C}^{t}$ the quantum channel obtained from
$\Lambda_{C\rightarrow AE}$ by taking the transpose on each of its Kraus
operators, and denoting by $\tilde{A}$, $\tilde{C}$, and $\tilde{E}$, replicas
of systems $A$, $C$, and $E$, we find that\begin{widetext}
\begin{align}
&  \mathbf{D}^{\Phi}(\Upsilon(\mathcal{N})\Vert\Upsilon(\mathcal{M}%
))\nonumber\\
&  =\mathbf{D}(\Upsilon(\mathcal{N}_{A\rightarrow B})(\Phi_{\tilde{C}C}%
)\Vert\Upsilon(\mathcal{M}_{A\rightarrow B})(\Phi_{\tilde{C}C}))\\
&  =\mathbf{D}((\Omega_{BE\rightarrow D}\circ\mathcal{N}_{A\rightarrow B}%
\circ\Lambda_{C\rightarrow AE})(\Phi_{\tilde{C}C})\Vert(\Omega_{BE\rightarrow
D}\circ\mathcal{M}_{A\rightarrow B}\circ\Lambda_{C\rightarrow AE}%
)(\Phi_{\tilde{C}C}))\\
&  \leq\mathbf{D}((\mathcal{N}_{A\rightarrow B}\circ\Lambda_{C\rightarrow
AE})(\Phi_{\tilde{C}C})\Vert(\mathcal{M}_{A\rightarrow B}\circ\Lambda
_{C\rightarrow AE})(\Phi_{\tilde{C}C}))\\
&  =\mathbf{D}((\Lambda_{\tilde{A}\tilde{E}\rightarrow\tilde{C}}^{t}%
\circ\mathcal{N}_{A\rightarrow B})(\Phi_{\tilde{A}A}\otimes\Phi_{\tilde{E}%
E})\Vert(\Lambda_{\tilde{A}\tilde{E}\rightarrow\tilde{C}}^{t}\circ
\mathcal{M}_{A\rightarrow B})(\Phi_{\tilde{A}A}\otimes\Phi_{\tilde{E}E}))\\
&  \leq\mathbf{D}(\mathcal{N}_{A\rightarrow B}(\Phi_{\tilde{A}A}\otimes
\Phi_{\tilde{E}E})\Vert\mathcal{M}_{A\rightarrow B}(\Phi_{\tilde{A}A}%
\otimes\Phi_{\tilde{E}E}))\\
&  =\mathbf{D}(\mathcal{N}_{A\rightarrow B}(\Phi_{\tilde{A}A})\Vert
\mathcal{M}_{A\rightarrow B}(\Phi_{\tilde{A}A}))\\
&  =\mathbf{D}^{\Phi}(\mathcal{N}\Vert\mathcal{M}).
\end{align}
\end{widetext}
where, in both inequalities, we used data processing of the divergence
$\mathbf{D}$, for the second equality we used the relation $\Lambda
_{C\rightarrow AE}(\Phi_{\tilde{C}C})=\Lambda_{\tilde{A}\tilde{E}%
\rightarrow\tilde{C}}^{t}(\Phi_{\tilde{A}A}\otimes\Phi_{\tilde{E}E})$, and for
the third equality we used the property $\mathbf{D}(\rho\otimes\omega
\Vert\sigma\otimes\omega)=\mathbf{D}(\rho\Vert\sigma)$ \cite{WWY13}. Now, to
prove the monotonicity under $\Theta$ as in~\eqref{maintheta}, we write
$\Theta=\sum_{x}p(x)\Upsilon_{x}$, where each $\Upsilon_{x}$ has the
form~\eqref{upsilon}. With this notation, we find that\begin{widetext}
\begin{align}
&  \mathbf{D}^{\Phi}(\Theta(\mathcal{N})\Vert\Theta(\mathcal{M}))=\mathbf{D}%
\left(  \sum_{x}p(x)\Upsilon_{x}(\mathcal{N}_{A\rightarrow B})(\Phi_{\tilde
{C}C})\Big\Vert\sum_{x}p(x)\Upsilon_{x}(\mathcal{M}_{A\rightarrow B}%
)(\Phi_{\tilde{C}C})\right)  =\label{g1}\\
&  \mathbf{D}\left(  \mathrm{Tr}_{X}\Big[\sum_{x}p(x)\Upsilon_{x}%
(\mathcal{N}_{A\rightarrow B})(\Phi_{\tilde{C}C})\otimes|x\rangle\langle
x|_{X}\Big]\Big\Vert\mathrm{Tr}_{X}\Big[\sum_{x}p(x)\Upsilon_{x}%
(\mathcal{M}_{A\rightarrow B})(\Phi_{\tilde{C}C})\otimes|x\rangle\langle
x|_{X}\Big]\right) \\
&  \leq\mathbf{D}\left(  \sum_{x}p(x)\Upsilon_{x}(\mathcal{N}_{A\rightarrow
B})(\Phi_{\tilde{C}C})\otimes|x\rangle\langle x|_{X}\Big\Vert\sum
_{x}p(x)\Upsilon_{x}(\mathcal{M}_{A\rightarrow B})(\Phi_{\tilde{C}C}%
)\otimes|x\rangle\langle x|_{X}\right) \\
&  \leq\mathbf{D}\left(  \sum_{x}p(x)\mathcal{N}_{A\rightarrow B}(\Phi
_{\tilde{A}A})\otimes|x\rangle\langle x|_{X}\Big\Vert\sum_{x}p(x)\mathcal{M}%
_{A\rightarrow B}(\Phi_{\tilde{A}A})\otimes|x\rangle\langle x|_{X}\right) \\
&  =\mathbf{D}\left(  \mathcal{N}_{A\rightarrow B}(\Phi_{\tilde{A}%
A})\Big\Vert\mathcal{M}_{A\rightarrow B}(\Phi_{\tilde{A}A})\right) \\
&  =\mathbf{D}^{\Phi}(\mathcal{N}\Vert\mathcal{M}). \label{g2}%
\end{align}
\end{widetext}
where, in the first inequality, we used the monotonicity of the divergence
under data processing, and for the second inequality, we used the monotonicity
under maps of the form in~\eqref{upsilon}.
\end{proof}

\section{Optimizing the adversarial channel divergence}

\label{app:adv-ch-div-limited}By definition, we always have that%
\begin{equation}
\mathbf{D}^{\text{adv}}(\mathcal{N}\Vert\mathcal{M})\geq\sup_{\psi_{RA}}%
\inf_{\sigma_{RA}}\mathbf{D}(\mathcal{N}_{A\rightarrow B}(\psi_{RA}%
)\Vert\mathcal{M}_{A\rightarrow B}(\sigma_{RA})),
\end{equation}
where $\psi_{RA}$ is pure with system$~R$ isomorphic to system$~A$.

To see the claim after \eqref{eq:adv-ch-div}, let $\rho_{RA}$ be an arbitrary
state with purification $\phi_{R^{\prime}RA}$. It thus holds that
$\phi_{R^{\prime}RA}$ is a purification of $\rho_{A}$, with $R^{\prime}R$
acting as the purifying systems. By taking a \textquotedblleft
canonical\textquotedblright\ purification of $\rho_{A}$ that is in direct
correspondence with its eigendecomposition, there exists a purification
$\varphi_{SA}$\ of $\rho_{A}$ with system $S$ isomorphic to system $A$. Since
the purification $\phi_{R^{\prime}RA}$ is related by an isometric channel
$\mathcal{U}_{S\rightarrow R^{\prime}R}$ to the purification $\varphi_{SA}$ as
$\phi_{R^{\prime}RA}=\mathcal{U}_{S\rightarrow R^{\prime}R}(\varphi_{SA})$ and
applying the isometric invariance of generalized divergences \cite{TWW14}, we
conclude for an arbitrary state $\omega_{SA}$ that%
\begin{align}
&  \mathbf{D}(\mathcal{N}_{A\rightarrow B}(\varphi_{SA})\Vert\mathcal{M}%
_{A\rightarrow B}(\omega_{SA}))\nonumber\\
&  =\mathbf{D}((\mathcal{U}_{S\rightarrow R^{\prime}R}\circ\mathcal{N}%
_{A\rightarrow B})(\varphi_{SA})\Vert(\mathcal{U}_{S\rightarrow R^{\prime}%
R}\circ\mathcal{M}_{A\rightarrow B})(\omega_{SA}))\\
&  =\mathbf{D}(\mathcal{N}_{A\rightarrow B}(\phi_{R^{\prime}RA})\Vert
\mathcal{M}_{A\rightarrow B}(\mathcal{U}_{S\rightarrow R^{\prime}R}%
(\omega_{SA})))\\
&  \geq\mathbf{D}(\mathcal{N}_{A\rightarrow B}(\rho_{RA})\Vert\mathcal{M}%
_{A\rightarrow B}((\operatorname{Tr}_{R^{\prime}}\circ\mathcal{U}%
_{S\rightarrow R^{\prime}R})(\omega_{SA})))\\
&  \geq\inf_{\sigma_{RA}}\mathbf{D}(\mathcal{N}_{A\rightarrow B}(\rho
_{RA})\Vert\mathcal{M}_{A\rightarrow B}(\sigma_{RA})).
\end{align}
The first inequality is from data processing under the partial trace over
$R^{\prime}$. Since the inequality holds for arbitrary $\omega_{SA}$, we
conclude that%
\begin{multline}
\inf_{\sigma_{RA}}\mathbf{D}(\mathcal{N}_{A\rightarrow B}(\rho_{RA}%
)\Vert\mathcal{M}_{A\rightarrow B}(\sigma_{RA}))\leq\\
\inf_{\omega_{SA}}\mathbf{D}(\mathcal{N}_{A\rightarrow B}(\varphi_{SA}%
)\Vert\mathcal{M}_{A\rightarrow B}(\omega_{SA})).
\end{multline}
We can then take a supremum to conclude that%
\begin{multline}
\inf_{\sigma_{RA}}\mathbf{D}(\mathcal{N}_{A\rightarrow B}(\rho_{RA}%
)\Vert\mathcal{M}_{A\rightarrow B}(\sigma_{RA}))\leq\\
\sup_{\varphi_{SA}}\inf_{\omega_{SA}}\mathbf{D}(\mathcal{N}_{A\rightarrow
B}(\varphi_{SA})\Vert\mathcal{M}_{A\rightarrow B}(\omega_{SA})).
\end{multline}
Since the inequality holds for an arbitrary choice of $\rho_{RA}$, we conclude
that%
\begin{equation}
\mathbf{D}^{\text{adv}}(\mathcal{N}\Vert\mathcal{M})\leq\sup_{\varphi_{SA}%
}\inf_{\omega_{SA}}\mathbf{D}(\mathcal{N}_{A\rightarrow B}(\varphi_{SA}%
)\Vert\mathcal{M}_{A\rightarrow B}(\omega_{SA})).
\end{equation}
This concludes the proof.

\end{document}